\documentclass[a4paper,11pt]{article}

\usepackage{geometry}
\usepackage{fullpage}
\usepackage{times}
\usepackage{soul}
\usepackage{url}
\usepackage[utf8]{inputenc}
\usepackage[small]{caption}
\usepackage{graphicx}
\usepackage{xcolor}
\usepackage{amsmath}
\usepackage{booktabs}
\usepackage{algorithm}
\usepackage{dsfont}
\usepackage{multirow}
\usepackage{amsthm}

\usepackage{titlesec}

\titlespacing*{\paragraph}
  {0pt}    
  {0.5em}  
  {1em}    

    \usepackage{longtable}
    \usepackage{nicefrac}
	\usepackage[T1]{fontenc}
	\usepackage{xspace}
    \usepackage{enumerate}
	\usepackage{tikz}
 
 \usetikzlibrary{positioning,arrows,shapes,decorations,calc,fit,arrows.meta} 
	\usepackage{colortbl}
	\usepackage{booktabs}
	\usepackage{cleveref}
	\usepackage{array}
	\usepackage{xcolor}
	\usepackage{enumitem}
	\usepackage{tabularx}
	\usepackage{thm-restate}
	\usepackage{ragged2e}
    \usepackage{microtype}
    \usepackage{mathrsfs} 
    \usepackage{caption}
    \usepackage{makecell}

\usepackage{amssymb}
\usepackage{algpseudocode}
\usepackage{amsfonts}
\usepackage{pifont}
\usepackage{cleveref}

\usepackage{mathrsfs}
\usepackage{enumitem}
\usepackage{nicefrac}

\usepackage{thm-restate}
\usepackage{bbm}

\newtheorem{theorem}{Theorem}

\newtheorem{proposition}{Proposition}
\newtheorem{lemma}{Lemma}

\theoremstyle{definition}
\newtheorem{definition}{Definition}
\newtheorem{example}{Example}

\usepackage{natbib}
\newcounter{remarkcount}
\crefname{remarkcount}{Remark}{Remarks}

\crefname{remark}{Remark}{Remarks}

\newcolumntype{L}[1]{>{\raggedright\let\newline\\\arraybackslash\hspace{0pt}}m{#1}}
\newcolumntype{C}[1]{>{\centering\let\newline\\\arraybackslash\hspace{0pt}}m{#1}}
\newcolumntype{R}[1]{>{\raggedleft\let\newline\\\arraybackslash\hspace{0pt}}m{#1}}
\renewcommand{\arraystretch}{1.2}

\renewcommand{\epsilon}{\varepsilon}

\usepackage{natbib}

\allowdisplaybreaks

\usepackage{authblk}

\title{\bf The Impossibility of Strategyproof Rank Aggregation\thanks{A preliminary version of this paper was published in the Proceedings of the 25th International Conference on Autonomous Agents and Multiagent Systems (AAMAS'26) \citep{EbLe26d}. This version extends the discussion of our results and the related work, includes all proofs, and significantly extends the analysis of the incentive ratio (\Cref{subsec:ApproxResults}). In particular, the incentive ratio for distance scoring rules is now tight, the Minimal Compromise rule is now discussed in depth, and the result on rank vector scoring rules is new. We further note that Manuel Eberl requested to be removed as a coauthor from this manuscript after leaving academia.
}}

\author{Patrick Lederer}

\affil{ILLC, University of Amsterdam, The Netherlands}

\date{\vspace{-10mm}}
\renewcommand{\arraystretch}{1.3}

\begin{document}

\maketitle

\begin{abstract}
In rank aggregation, the goal is to combine multiple input rankings into a single output ranking. In this paper, we analyze rank aggregation methods, so-called social welfare functions (SWFs), with respect to strategyproofness, which requires that no agent can misreport her ranking to obtain an output ranking that is closer to her true ranking in terms of the Kemeny distance. As our main result, we show that no anonymous SWF satisfies unanimity and strategyproofness if there are at least four alternatives. This result is proved using SAT solving, a computer-aided theorem-proving technique, and verified using Isabelle/HOL, a highly trustworthy interactive proof assistant. Moreover, we show by hand that strategyproofness is incompatible with majority-consistency, a variant of Condorcet-consistency for SWFs. Lastly, we analyze the incentive ratio of SWFs, which measures the worst-case ratio between a voter's utility when manipulating and when voting honestly. Unfortunately, we show that many natural SWFs have a large incentive ratio and are thus severely manipulable. Specifically, we prove that, if there are $m\geq 4$ alternatives, all positional scoring rules have an unbounded incentive ratio, and all distance scoring rules have an incentive ratio of roughly ${m\choose 2}$. Moreover, we identify an SWF called the Minimal Compromise rule, which has an incentive ratio of $m-2$.
\end{abstract}

\section{Introduction}
An important problem for multi-agent systems and artificial intelligence is rank aggregation: multiple input rankings need to be aggregated into a single output ranking.
For instance, this task arises when a hiring committee is asked to produce a ranking of the applicants based on the preferences of the committee members \citep{KuRu20a,CaRu24a}, when aggregating the outputs of multiple ranking algorithms in ensemble learning \citep{Prat12a,LDW22a}, or when recommender systems infer an output ranking based on the preferences of multiple users \citep{ADM+10a,PSK22a}.
Furthermore, rank aggregation is applied in computational biology \citep{Lin10a,KLAV12a}, engineering \citep{Gaer18a}, and meta-search \citep{DKNS01a,ReSt03a}. While rank aggregation has been studied independently for many of these applications, this problem also lies at the heart of social choice theory \citep{ASS11a,BCE+14a}. In this field, rank aggregation is formalized as a preference aggregation problem: a set of voters report rankings to indicate their preferences over the alternatives, and these rankings need to be aggregated into a collective output ranking. The central objects of study for this problem are \emph{social welfare functions (SWFs)}, which map every profile of the voters' input rankings to a single output ranking.

One fundamental property of such SWFs---and, more broadly, of many applications of collective decision-making---is \emph{strategyproofness}. This axiom requires that the SWF incentivizes the voters to reveal their preferences truthfully. We believe strategy\-proofness to be important for many applications of rank aggregation, as it ensures that voters cannot game the rank aggregation rule by strategically misreporting their true ranking. For example, if an SWF that violates strategyproofness is used to aggregate the preferences of a hiring committee, a committee member may misreport her preferences to ensure that her preferred candidates are more likely to get the job. Similarly, in recommender systems, a user may try to manipulate the output ranking so that the final recommendations are closer to her own preferences. Lastly, even in technical applications such as ensemble learning, strategyproofness may be desirable as it offers resistance against malicious behavior from individual algorithms.

However, while both SWFs and strategyproofness are individually well understood in social choice theory \citep[see, e.g.,][]{Tayl02a,Barb10a,ASS11a}, the study of strategyproof SWFs has only recently gained attention \citep{BoSp14a,Atha16a,Atha19a,ABE25a}.
One possible reason for this is that it is challenging to define strategyproofness for SWFs because it is unclear how voters compare different output rankings. For instance, if a voter's true ranking is $a\succ b\succ c$, does she prefer the ranking $b\succ a\succ c$ or the ranking $c\succ a\succ b$? Following the recent literature \citep[e.g.,][]{Atha16a,LPW24a,ABE25a}, we will address this issue by using the Kemeny distance to define the voters' preferences over rankings. This distance counts the number of pairs of alternatives on which two rankings disagree, and we suppose that voters prefer rankings that have a smaller Kemeny distance to their true ranking. This assumption is motivated by the fact that, in many applications, it is common to minimize the total Kemeny distance to the input rankings \citep[e.g.,][]{DKNS01a,LDW22a,CaRu24a}, which implicitly assumes that voters prefer rankings closer to their own ranking. Moreover, the Kemeny distance itself is theoretically well understood and has been repeatedly characterized \citep[e.g.,][]{Keme59a,Diac88a,CaSt18a}.

Based on the Kemeny distance, \emph{\mbox{(Kemeny-)}strategyproofness} requires of an SWF that voters cannot obtain a ranking that is closer to their true ranking by strategically reporting another ranking. The central question of this paper is then whether appealing strategyproof SWFs exist. To this end, we note that, when there are $m\leq 3$ alternatives, it is known that, e.g., the Kemeny rule is strategyproof \citep{Atha16a,ABE25a}, thus answering our question in the affirmative in this case. However, even if there are only $m=4$ alternatives, it is open whether such SWFs exist or whether an impossibility theorem similar to the Gibbard-Satterthwaite theorem \citep{Gibb73a,Satt75a} holds.

\subsection{Our Contribution}

In this paper, we show that, under mild side conditions, no strategyproof SWF exists, thereby establishing an analogue of the Gibbard-Satterthwaite theorem for rank aggregation. For a concise overview of our results, we refer to \Cref{tab:summary}. In more detail, as our main theorem, we show that every SWF that satisfies two standard conditions called anonymity and unanimity fails strategyproofness when there are $m\geq 5$ alternatives and an even number of voters $n$, or when there are $m=4$ alternatives and $n$ is a multiple of~$4$ (\Cref{thm:mainimp}). Anonymity and unanimity denote two basic properties: anonymity requires that all voters are treated equally, and unanimity requires that the output ranking puts one alternative $x$ ahead of another alternative $y$ if all voters prefer $x$ to $y$. Both of these conditions are rather minimal requirements on an SWF, so our result shows that no SWF that seems acceptable in practice can satisfy strategyproofness.

The proof of our main theorem is obtained via SAT solving, a computer-aided theorem-proving technique. Specifically, we encode the problem of deciding whether an SWF satisfies anonymity, unanimity, and strategyproofness in a logical formula and show with the help of a computer that this formula is unsatisfiable when there are $m=5$ alternatives and $n=2$ voters, and when there are $m=4$ alternatives and $n=4$ voters. This proves two base cases for our impossibility theorem, which we lift to our final statement by applying inductive arguments. Following standard practice, we also extract a proof of one of our base cases in a human-readable format (cf. \Cref{app:mainproof}). However, since this proof spans over 20 pages, we additionally verify our main theorem with Isabelle \citep{NPW02a}, a highly trustworthy computer program designed to verify mathematical proofs. In particular, our verification \citep{EbLe26a} derives our impossibility theorem as stated in \Cref{thm:mainimp} directly from the axioms, thereby releasing us from the need to check any intermediate steps. Moreover, we also published the code for generating our SAT encoding \citep{EbLe26b}, which makes it possible for other researchers to reproduce the result.

As a second result, we prove that no strategyproof SWF satisfies a variant of Condorcet-consistency for rankings. To introduce this condition, we define the majority relation of a profile as the binary relation on the alternatives that prefers an alternative $x$ to another alternative $y$ if a majority of the voters prefers $x$ to $y$. Our axiom, which is called majority-consistency, then requires that, when the majority relation corresponds to a ranking, the SWF chooses this ranking. In \Cref{thm:PWC}, we show that no SWF satisfies both majority-consistency and strategyproofness if there are $m\geq 4$ alternatives and sufficiently many voters. This result can be seen as a counterpart to the impossibility of strategyproof and Condorcet-consistent social choice functions, which return single alternatives instead of rankings \citep[e.g.,][]{Gard76a,GePe17a,BBL21b}. Moreover, since we have shown this result by hand, it allows us to showcase the proof of an impossibility theorem based on strategyproofness.

Lastly, we analyze the incentive ratio of several SWFs to measure how manipulable they are. Roughly, the incentive ratio of an SWF quantifies the worst-case ratio between the utility of a voter when manipulating and when voting honestly. In particular, we define this concept in terms of the Kemeny utility function, which counts the pairs of alternatives on which two rankings agree, instead of the Kemeny distance. The incentive ratio has been successfully applied in private-goods settings \citep[e.g.,][]{CDZZ12a,WWZ20a,LSX24a} to show that many desirable mechanisms that are manipulable still limit the gains from manipulation. Since manipulations are usually costly (e.g., voters need to spend resources or cognitive effort to learn the preferences of the other voters and to compute a beneficial manipulation), a bounded incentive ratio may suffice to disincentivize strategic behavior in practice.
\noindent\begin{table*}[t]
\centering
\setlength{\tabcolsep}{0cm}
\resizebox{\linewidth}{!}{
\begin{tabular}{p{0.5\textwidth}@{\hspace{1cm}}p{0.5\textwidth}}
\toprule
Negative results & Positive results \\\toprule
\begin{minipage}[t]{0.5\textwidth}
\begin{itemize}[leftmargin = *, topsep=0pt,itemsep=0pt]
\item[$\ominus$] No majority-consistent and strategyproof SWF when $m \geq 4$ (\Cref{thm:PWC})
\item[$\ominus$] No unanimous, anonymous, and strategyproof SWF when $m\geq 4$ (\Cref{thm:mainimp})
\item[$\ominus$] Positional scoring rules have an unbounded incentive ratio when $m\geq 3$ (\Cref{prop:posscoring}).
\item[$\ominus$] The Kemeny rule as well as all distance scoring rules have an incentive ratio between ${m\choose 2}-m$ and ${m\choose 2}$ when $m\geq 4$ (\Cref{prop:distancescoring}).
\item[$\ominus$] The Minimal Compromise rule has an incentive ratio of $m-2$ when $m\geq 4$, which is optimal within the class of unanimous rank vector scoring rules (\Cref{prop:MC}).
\end{itemize}
\end{minipage}
&
\begin{minipage}[t]{0.5\textwidth}
\begin{itemize}[leftmargin = *, topsep=0pt,itemsep=0pt]
\item[$\oplus$] The Kemeny rule is strategyproof, unanimous, anonymous, and majority-consistent if $m\leq 3$ \citep{Atha16a,ABE25a}
\item[$\oplus$] There are non-dictatorial, unanimous, and strategyproof SWFs \citep{ABE25a}
\item[$\oplus$] The Kemeny rule is betweenness strategyproof for all $m$ \citep{BoSp14a}
\end{itemize}
\end{minipage}\\
\bottomrule
\end{tabular}
}
\vspace{3pt}
\caption{Summary of our results and comparison to related work. All negative results are proved in this paper. While there have been impossibility results for strategyproof SWFs before our work \citep[e.g.,][]{BoSt92a,ABE25a}, \Cref{thm:mainimp} supersedes all of them.}
\label{tab:summary}
\end{table*}

Unfortunately, we show for three natural classes of SWFs that all corresponding rules have a high incentive ratio. In more detail, we first consider positional scoring rules and show that all such SWFs have an unbounded incentive ratio, as voters who obtain utility $0$ when voting truthfully can manipulate (\Cref{prop:posscoring}). Second, we show that the Kemeny rule and every distance scoring rule have an incentive ratio of approximately ${m\choose 2}$ (\Cref{prop:distancescoring}). Here, distance scoring rules refer to a generalization of the Kemeny rule, where we apply a convex and strictly increasing score function to the Kemeny distance of all voters and then choose the ranking that minimizes the total score. Lastly, in an attempt to optimize the incentive ratio, we consider the Minimal Compromise rule, which sorts the alternatives in decreasing order of the worst rank assigned to them by any voter. We show that this SWF has an incentive ratio of $m-2$ (\Cref{prop:MC}). Moreover, we prove that an incentive ratio of $m-2$ is optimal within the class of unanimous rank vector scoring rules, which compute a score for each alternative based on how often it appears at each rank and then order the alternatives by their scores. In summary, these results show that many natural SWFs are severely manipulable as they have large incentive ratios, while the Minimal Compromise rule minimizes these incentives among the rules in our analysis.

\subsection{Related Work}

Our work relates to two influential lines of work in social choice theory: the study of SWFs and the analysis of strategyproofness. For in-depth introductions to these topics, we refer to the book by \citet{ASS11a} and the survey by \citet{Barb10a}. In the following, we discuss the works in these research lines that are most relevant to our paper. Moreover, we review prior work on strategyproof rank aggregation in detail and highlight a relation between our model and judgement aggregation.

\paragraph{Social Welfare Functions.} Social welfare functions are among the oldest concepts in social choice theory and have received considerable attention. Indeed, even Arrow's impossibility theorem was originally stated for SWFs \citep{Arro51a}. To date, a large range of SWFs has been proposed, including the Kemeny rule \citep{Keme59a,KeSn60a}, various types of scoring rules \citep[e.g.,][]{Smit73a,Youn75a,CRX09a,Lede23a,BBP26a}, and Condorcet-style rules \citep[e.g.,][]{Cope51a,Slat61a,Schu11a}. Essentially all of these rules satisfy anonymity and unanimity and are thus covered by our main impossibility theorem. In the existing literature, SWFs are primarily studied with respect to consistency notions such as population consistency or in\-de\-pen\-dence of irrelevant alternatives. Prominent results of this type include the characterizations of the Kemeny rule \citep{YoLe78a,Youn88a,CaSt13a} and of the class of positional scoring rules \citep{Smit73a,Youn74b,Youn75a}. Furthermore, more recent contributions in this line of work provide, e.g., characterizations of the Slater rule \citep{CPS21a} and of sequential scoring rules \citep{FBC14a,BBP26a}. Finally, the idea of proportional representation has recently been explored for SWFs \citep{LPW24a,ALP+25b,Lede25b}. Analogous to our definition of strategyproofness, these authors use the Kemeny distance to quantify how well a voter is represented by the output ranking. However, while these papers share our framework of SWFs, they do not focus on strategyproofness.

\paragraph{Strategyproofness.} Similar to SWFs, strategyproofness has attracted significant attention in social choice theory, although the results in this area are more negative. In particular, Gibbard and Satterthwaite \citep{Gibb73a,Satt75a} have shown that no reasonable deterministic single-winner voting rule is strategyproof. Motivated by this result, numerous works have aimed to circumvent this impossibility theorem, for instance by allowing randomized or set-valued outcomes \citep[e.g.,][]{Gard76a,Gibb77a,Kell77a,Bran17a,BSS19a,BBL21b,BrLe24a} or by restricting the feasible input rankings \citep[e.g.,][]{Moul80a,ChZe21a}. Except for domain restrictions, these approaches have mostly led to strengthened impossibility theorems. Our paper can also be interpreted in this line of work: since rankings contain more information than a single winner, one may attempt to escape the Gibbard-Satterthwaite theorem by studying SWFs and a suitable notion of strategyproofness. As our results show, this approach does not work when using Kemeny-strategyproofness.

We furthermore note that, similar to our main theorem, many results on strategyproof social choice have been shown via computer-aided theorem-proving methods \citep{BGP16c,BBEG16a,BrGe15a,BBGH18a,Pete18a,KdV+20a,BBPS21a,BSS19a,DDE+22a,SSU25a}. For a general overview of this technique, we refer to the survey by \citet{GePe17a}. In more detail, most of these papers study strategyproofness in voting settings where the outcome is a set of winners, which is interpreted either as a selected committee or as a set of shortlisted alternatives. Similar to our approach for rankings, computer-aided methods are helpful for these problems since strategyproofness can be naturally encoded in these settings. However, for many of the listed papers, the authors managed to extract a reasonably sized human-readable proof from the computer proof, which often reasons about fewer than 20 profiles. By contrast, the smallest proof that we found for our impossibility of strategyproof, unanimous, and anonymous SWFs reasons about roughly 200 profiles.

Lastly, we also consider the incentive ratio in this paper to measure how severely particular SWFs can be manipulated. To our knowledge, this ratio was first introduced for Fisher markets, where goods are sold to agents \citep{CDZ11a,CDZZ12a,CDTZ16a,CDT+22a}. For example, for linear utilities, the incentive ratio of the corresponding market mechanism is $2$. Moreover, the incentive ratio has been analyzed for various assignment problems, such as the probabilistic serial rule \citep{HWWZ24a,LSX24a}, picking sequences \citep{XiLi20a}, and cake cutting \citep{BTWY25a}. More recently, the incentive ratio was also studied for approval-based budget division \citep{ALV26a}, where the Nash product rule has an incentive ratio of $2$. Hence, all these results show that, in their respective settings, there are desirable methods with bounded incentive ratios, so the gains from manipulating are limited.

\paragraph{Strategyproofness for Social Welfare Functions.} More directly related, there are several works that study Kemeny-strategyproofness for SWFs. To our knowledge, \citet{BoSt92a} were the first to study this condition. Specifically, these authors show that a variant of Kemeny-strategyproofness for groups of voters leads to an impossibility when requiring a technical auxiliary property called weak extrema independence. Moreover, \citet{Atha16a} and \citet{ABE25a} prove that, when $m\leq 3$, the Kemeny rule (with suitable tie-breaking) and other SWFs are Kemeny-strategyproof, but these results cease to hold when $m\geq 4$. Further, \citet{ABE25a} show that no anonymous SWF satisfies Kemeny-strategyproofness and a technical property called preference selection. This demanding condition requires an SWF to always return a ranking that is present in the input profile, and it is violated by every commonly considered SWF. 
Furthermore, several alternative notions of strategyproofness have been studied for SWFs. Most notably, \citet{BoSp14a} introduce a notion called betweenness strategyproofness, which requires that, by manipulating, a voter cannot obtain a ranking that lies on a single-crossing sequence of rankings from her true ranking to the output ranking chosen when voting truthfully. Betweenness strategyproofness is weaker than Kemeny-strategyproofness, and \citet{BoSp14a} show that, e.g., the Kemeny rule always satisfies it. This condition was further analyzed by \citet{Sato15a}, \citet{Harl16a}, and \citet{Atha19a}, the latter two of whom use it to characterize so-called status-quo rules. Further notions of strategyproofness for SWFs have been studied by \citet[][Chapter 3]{Bonk18b} and \citet{DDL24a}.

\paragraph{Strategyproofness for Judgement Aggregation.} Lastly, rank aggregation can be seen as a special case of judgement aggregation, where we need to aggregate the voters' judgments on logically related formulas. In this setting, strategyproofness has been studied before \citep[e.g.,][]{DiLi07a,DiLi07b,TeEn19a,TeEn22a}, and we refer to \citet{BRS17a} for a survey of this problem. In particular, \citet{DiLi07a} show that, for their preference-free notion of manipulation, strategyproofness is equivalent to monotonicity and independence of irrelevant issues. Based on this insight, one can infer an impossibility theorem similar to the Gibbard-Satterthwaite impossibility theorem. However, their corresponding notion of strategyproofness requires immunity to manipulation for all closeness-respecting preferences and is thus much stronger than Kemeny-strategyproofness, which makes their results incomparable to ours.

\section{Preliminaries}\label{sec:preliminaries}

We are now introducing our formal model of rank aggregation (cf. \Cref{subsec:model}) and then define strategyproofness for SWFs (cf. \Cref{subsec:SPdef}). 

\subsection{The Model}\label{subsec:model}

Let $A=\{a,b,c,\dots\}$ be a set of $m$ alternatives and $N=\{1,\dots,n\}$ be a set of $n$ voters. Every voter $i\in N$ reports a \emph{ranking} $\succ_i$ over the alternatives to indicate her preferences. Formally, a ranking $\succ_i$ is a transitive, asymmetric, and complete binary relation on $A$. In particular, this means that every voter has to rank all alternatives without ties. The set of all rankings is denoted by $\mathcal{R}$. A \emph{(ranking) profile} $R=(\succ_1,\dots,\succ_n)$ is the collection of the rankings of all voters in $N$, and the set of all profiles is $\mathcal{R}^N$. We will write rankings as sequences of alternatives and indicate the voter submitting a ranking directly before it. For example, $3:abc$ means that voter $3$ prefers $a$ to $b$ to $c$.

To decide which output ranking to choose for a given profile, we use \emph{social welfare functions (SWFs)} $f$, which map every ranking profile $R$ to a single output ranking $\rhd=f(R)$. Thus, the signature of a social welfare function is $\mathcal{R}^N\rightarrow\mathcal{R}$. To clearly distinguish between input and output rankings, we denote the former by $\succ$ and the latter by $\rhd$. However, both input and output rankings are elements of $\mathcal{R}$, so an output ranking can also be part of a profile and an input ranking can be chosen by our SWF.
Furthermore, we note that all commonly considered SWFs may sometimes be unable to differentiate between multiple winning rankings. For instance, if the input profile contains every ranking ${\succ}\in\mathcal{R}$ once, there is no fair way to distinguish between the output rankings. For such cases, we assume that the winning ranking is chosen based on an external ranking $\mathbf{>}$ over the alternatives that is lexicographically extended to rankings. Specifically, given two rankings $\rhd_1=x_1\dots x_m$ and $\rhd_2=y_1\dots y_m$, we say that $\rhd_1$ is \emph{lexicographically more preferred} than $\rhd_2$ if and only if there is $\ell\in\{1,\dots,m\}$ such that $x_\ell\,\mathbf{>}\,y_\ell$ and $x_i=y_i$ for all $i\in\{1,\dots,\ell-1\}$. When defining SWFs, we assume that the lexicographically most preferred ranking among the eligible rankings is chosen. We emphasize, however, that this assumption is only used for defining SWFs and that our impossibility theorems hold for any tie-breaking scheme.

To give one example of an SWF, we next introduce the Kemeny rule. This SWF was first suggested by \citet{Keme59a} and is perhaps the most prominent method in rank aggregation. To introduce it, we define the \emph{Kemeny distance} (which is also known as swap distance or Kendall-tau distance) between two rankings $\succ$ and $\rhd$ by $\Delta(\succ,\rhd)=|\{(x,y)\in A^2\colon x\succ y\land y\rhd x\}|$.
Less formally, $\Delta(\succ,\rhd)$ is the number of pairs of alternatives on which $\succ$ and $\rhd$ disagree. Then, the \emph{Kemeny rule} chooses the lexicographically most preferred ranking that minimizes the total Kemeny distance to the input rankings, i.e., $f_\mathit{Kemeny}(R)=\arg\min_{\rhd\in \mathcal{R}} \sum_{i\in N} \Delta(\succ_i,\rhd)$.

\subsection{Strategyproofness}\label{subsec:SPdef}

The central axiom in our analysis is strategyproofness, which requires that voters cannot benefit by lying about their true ranking. Following the literature \citep[e.g.,][]{BoSt92a,Atha16a,ABE25a}, we define this axiom by assuming that the voters' preferences over rankings are induced by the Kemeny distance $\Delta(\succ,\rhd)=|\{(x,y)\in A^2\colon x\succ y\land y\rhd x\}|$: a voter with ranking $\succ_i$ prefers a ranking $\rhd$ to another ranking $\rhd'$ if $\Delta(\succ_i,\rhd)<\Delta(\succ_i,\rhd')$. This formalizes the idea that voters prefer rankings that are more similar to their true ranking.

In other words, we use the Kemeny distance to lift the voters' preferences over alternatives to preferences over rankings. We believe that this approach is reasonable because the Kemeny distance is one of the most commonly considered distances over rankings. In particular, in many applications of rank aggregation \citep[e.g.,][]{DKNS01a,KLAV12a,LDW22a,CaRu24a}, it is common to minimize the total Kemeny distance to find good output rankings. This implicitly assumes that voters prefer rankings that have a smaller Kemeny distance to their input ranking. Moreover, many papers in social choice theory rely on the Kemeny distance to define their axioms and evaluate rankings \citep[e.g.,][]{OEFS13a,ZSW14a,LPW24a,Lede25b}. Lastly, the Kemeny distance itself is also theoretically well understood and allows for appealing characterizations \citep{Keme59a,Diac88a,CaSt18a}, which give clear normative arguments for when this distance is the right choice.

Based on the Kemeny distance, we can define strategyproofness as usual by requiring that voters cannot obtain a more preferred ranking by voting strategically.

\begin{definition}[Strategyproofness]\label{def:strategyproofness}
An SWF $f$ is \emph{(Kemeny-)strategyproof} if $\Delta(\succ_i,f(R))\leq \Delta(\succ_i,f(R'))$ for all voters $i\in N$ and profiles $R,R'\in\mathcal{R}^N$ such that ${\succ_j}={\succ_j'}$ for all $j\in N\setminus\{i\}$.
\end{definition}

We note that the definition of strategyproofness only depends on the ordinal comparisons between pairs of possible output rankings but not on the actual values $\Delta(\succ,\rhd)$. As a consequence, every affine transformation of the Kemeny distance induces the same preferences over rankings and thus the same strategyproofness notion. In particular, this means that we can equivalently define Kemeny-strategyproofness based on the \emph{Kemeny utility function} given by $u(\succ,\rhd)=|\{(x,y)\in A^2\colon x\succ y\land x\rhd y\}|$ for all $\succ,\rhd\in\mathcal{R}$. Indeed, because $u(\succ,\rhd)={m\choose 2}-\Delta(\succ,\rhd)$ for all rankings $\succ,\rhd\in\mathcal{R}$, our strategyproofness notion is equivalent to requiring that no voter can increase her Kemeny utility by lying about her true ranking. On the other hand, there are numerous distances and utility functions between rankings \citep{Diac88a,KuVa10a} that are not affine transformations of the Kemeny distance. When using such distance functions, one obtains an alternative strategyproofness notion, which may lead to different results.

To further illustrate strategyproofness for SWFs, we next discuss an example showing that the Kemeny rule is manipulable.

\begin{example}
Consider the following profile $R$ with $m=4$ alternatives $A=\{a,b,c,d\}$ and $n=5$ voters $N=\{1,\dots,5\}$.
\begin{center}
\begin{tabular}{cccccc}
$R$: & $1$: $abcd$ & $2$: $cdab$ & $3$: $dbac$ & $4$: $bcda$ & $5$: $adcb$
\end{tabular}
\end{center}
For this profile, the Kemeny rule chooses the ranking $\rhd=dabc$, which has a total Kemeny distance of $\sum_{i\in N}\Delta(\succ_i,\rhd)=3+3+1+4+2=13$. By contrast, if voter $4$ misreports her true ranking by swapping $b$ and $c$, i.e., if $\succ_4'=cbda$, the Kemeny rule chooses the ranking $\rhd'=cdab$. Indeed, after this deviation, the total Kemeny distance of $\rhd$ is $14$, while that of $\rhd'$ is $\sum_{i\in N}\Delta(\succ_i,\rhd')=4+0+4+2+3=13$. Further, it holds that $\Delta(\succ_4,\rhd)=\Delta(bcda,dabc)=4>3=\Delta(bcda,cdab)=\Delta(\succ_4,\rhd')$. This shows that voter $4$ prefers the ranking $\rhd'$ to the ranking $\rhd$ selected for $R$, so the Kemeny rule fails strategyproofness when $m=4$ and $n=5$.
\end{example}


\section{Strategyproofness and Majority-consistency}\label{sec:PWCresults}

As our first result, we show that strategyproofness is in conflict with a Condorcet-style axiom for rankings, which we call majority-consistency. Intuitively, this axiom is based on the majority principle: if a majority of the voters prefers an alternative $x$ to another alternative $y$, then $x$ should be ranked ahead of $y$. To make this more formal, we define the \emph{majority relation} $\succsim_R$ of a profile $R$ by $x\succsim_R y$ if and only if $|\{i\in N\colon x\succ_i y\}|\geq |\{i\in N\colon y\succ_i x\}|$. Then, majority-consistency requires an SWF to return the majority relation whenever this relation corresponds to a ranking.

\begin{definition}[Majority-consistency]
An SWF $f$ is \emph{majority-consistent} if $f(R)={\succsim_R}$ for all profiles $R$ such that the majority relation $\succsim_R$ is transitive and antisymmetric.
\end{definition}

We note that majority-consistency does not impose any constraints on the outcome of the SWF if the majority relation does not correspond to a ranking. Moreover, our axiom is conceptually related to, e.g., the Condorcet-consistency axiom for rankings by \citet{YoLe78a} and the notion of popular rankings \citep{ZSW14a,KCM23a}. Furthermore, the Kemeny rule and many SWFs that are based on the (weighted) majority relation satisfy this condition.

As we show next, no majority-consistent SWF satisfies strategyproofness once there are $m\geq 4$ alternatives. While we primarily include this result to showcase the proof of an impossibility theorem involving strategyproofness, it is also one of the strongest impossibility theorems regarding strategyproofness in rank aggregation. Moreover, our result mirrors the well-known fact that no single-winner voting rule that guarantees the selection of the Condorcet winner (i.e., an alternative that beats every other alternative in pairwise comparisons) can be strategyproof \citep[e.g.,][]{Gard76a,GePe17a,BBL21b}.

\begin{theorem}\label{thm:PWC}
No strategyproof SWF satisfies majority-consistency if there are $m\geq 4$ alternatives and $n\geq 9$ voters with $n\notin\{10,12,14,16\}$.
\end{theorem}
\begin{proof}
We first study the case that $m=4$ and $n=9$ and later generalize the result to larger numbers of voters and alternatives. Hence, assume for contradiction that there is a strategyproof and majority-consistent SWF $f$ for $4$ alternatives and $9$ voters. We focus on the following two profiles $\bar R$ and $\widehat R$ to derive a contradiction.
\begin{center}
\renewcommand{\arraystretch}{1}
    \begin{tabular}{llllll llll}
        $\bar R$ & $1$: $cdba$ & $2$: $badc$ & $3$: $dbac$ & $4$: $cbad$ & $5$: $adcb$
         & $6$: $cadb$ & $7$: $dcba$ & $8$: $dabc$ & $9$: $abcd$\smallskip\\
         $\widehat R$ & $1$: $cdba$ & $2$: $badc$ & \textcolor{black}{$3$: $dbca$} & $4$: $cbad$ & $5$: $adcb$
         & $6$: $cadb$ & $7$: $dcba$ & $8$: $dabc$ & $9$: $abcd$
    \end{tabular}
\end{center}

We show that $f$ has to choose ${\bar\rhd}=adcb$ for $\bar R$ and ${\widehat\rhd}=dcba$ for $\widehat R$. Since the profiles $\bar R$ and $\widehat R$ differ only in the ranking of voter~$3$, this means that $f$ is manipulable. Indeed, it holds that $\Delta(\bar\succ_3,{\bar\rhd})=3>2=\Delta(\bar\succ_3,{\widehat\rhd})$, so voter $3$ prefers ${\widehat\rhd}$ to ${\bar\rhd}$. It therefore remains to show that $f$ chooses $adcb$ and $dcba$ for $\bar R$ and $\widehat R$, respectively.

\paragraph{Claim 1: $f(\bar R)=adcb$.}
To prove this claim, we consider the following five profiles $\bar R^1,\dots,\bar R^5$, all of which differ from $\bar R$ in the ranking of a single voter. This ranking is highlighted in blue.
\begin{center}
\renewcommand{\arraystretch}{1}
    \begin{tabular}{llllll llll}
        $\bar R^1$ & \textcolor{blue}{$1$: $cdab$} & $2$: $badc$ & $3$: $dbac$ & $4$: $cbad$ & $5$: $adcb$
         & $6$: $cadb$ & $7$: $dcba$ & $8$: $dabc$ & $9$: $abcd$\smallskip\\
        $\bar R^2$ & $1$: $cdba$ & \textcolor{blue}{$2$: $abdc$} & $3$: $dbac$ & $4$: $cbad$ & $5$: $adcb$
         & $6$: $cadb$ & $7$: $dcba$ & $8$: $dabc$ & $9$: $abcd$\smallskip\\
        $\bar R^3$ & $1$: $cdba$ & $2$: $badc$ & \textcolor{blue}{$3$: $dabc$} & $4$: $cbad$ & $5$: $adcb$
         & $6$: $cadb$ & $7$: $dcba$ & $8$: $dabc$ & $9$: $abcd$\smallskip\\
        $\bar R^4$ & $1$: $cdba$ & $2$: $badc$ & $3$: $dbac$ & \textcolor{blue}{$4$: $cabd$} & $5$: $adcb$
         & $6$: $cadb$ & $7$: $dcba$ & $8$: $dabc$ & $9$: $abcd$\smallskip\\
        $\bar R^5$ & $1$: $cdba$ & $2$: $badc$ & $3$: $dbac$ & $4$: $cbad$ & \textcolor{blue}{$5$: $dabc$}
         & $6$: $cadb$ & $7$: $dcba$ & $8$: $dabc$ & $9$: $abcd$
    \end{tabular}
\end{center}

In all five profiles, the majority relation is transitive and antisymmetric and therefore constitutes a ranking. In more detail, in $\bar R^1$ to $\bar R^4$, all of which arise from $\bar R$ by swapping $a$ and $b$ in the ranking of a single voter, the majority relation corresponds to the ranking $adcb$. Further, in $\bar R^5$, the majority relation is given by $dbac$. Hence, majority-consistency requires that $f(\bar R^1)=\dots=f(\bar R^4)=adcb$ and $f(\bar R^5)=dbac$. Consequently, strategyproofness implies the following constraints on the ranking $\bar\rhd$ chosen for $\bar R$:
\begin{align}
    \Delta(cdba,\bar\rhd)&\leq \Delta(cdba,f(\bar R^1))=\Delta(cdba,adcb)=4\\
    \Delta(badc,\bar\rhd)&\leq \Delta(badc,f(\bar R^2))=\Delta(badc,adcb)=3\\
    \Delta(dbac,\bar\rhd)&\leq \Delta(dbac,f(\bar R^3))=\Delta(dbac,adcb)=3\\
    \Delta(cbad,\bar\rhd)&\leq \Delta(cbad,f(\bar R^4))=\Delta(cbad,adcb)=4\\
    \Delta(adcb,\bar\rhd)&\leq \Delta(adcb,f(\bar R^5))=\Delta(adcb,dbac)=3.
\end{align}

We claim that only the ranking $adcb$ satisfies these constraints. To prove this, we consider several cases. First, $d$ cannot be bottom-ranked in $\bar\rhd$ because $\Delta(dbac,\rhd)\geq 4$ for all rankings $\rhd$ that rank $d$ last and are not equal to $bacd$. Hence, all these rankings fail Condition~(3). Further, it holds that $\Delta(adcb,bacd)=4$, so the ranking $bacd$ fails Condition~(5).
Next, $\bar\rhd$ cannot bottom-rank $a$: every ranking $\rhd$ that places $a$ last and is not equal to $dcba$ violates Condition~(5), since $\Delta(adcb,\rhd)\geq 4$, and the ranking $dcba$ fails Condition~(2), since $\Delta(badc,dcba)=4$.
Third, we show that alternative $c$ cannot be bottom-ranked by $\bar\rhd$. For this, we observe that Condition~(1) implies that ${\bar\rhd}\notin\{abdc,badc,adbc\}$ because all these rankings have a Kemeny distance of at least $5$ from $cdba$. Further, Condition~(4) requires that ${\bar\rhd}\notin\{dabc,dbac\}$ because these rankings have a Kemeny distance of at least $5$ from $cbad$. Lastly, Condition~(5) shows that ${\bar\rhd}\neq bdac$ because $\Delta(adcb,bdac)=4$. We therefore conclude that $c$ is not bottom-ranked by $\bar\rhd$.
Because all other options have been ruled out, $b$ must be bottom-ranked by $\bar\rhd$. This means that $\Delta(badc,\bar\rhd)\geq 3$. Moreover, this inequality needs to be tight due to Condition~(2), which is only possible if $f(\bar R)={\bar\rhd}=adcb$.

\paragraph{Claim 2: $f(\widehat R)=dcba$.}
We next show that $f(\widehat R)=dcba$, for which we consider the following five profiles. All of these profiles differ from $\widehat R$ only in the highlighted ranking.
\begin{center}
\renewcommand{\arraystretch}{1}
\begin{tabular}{llllll llll}
$\widehat R^1$ & $1$: $cdba$ & $2$: $badc$ & $3$: $dbca$ & \textcolor{blue}{$4$: $cbda$} & $5$: $adcb$
& $6$: $cadb$ & $7$: $dcba$ & $8$: $dabc$ & $9$: $abcd$\smallskip\\
$\widehat R^2$ & $1$: $cdba$ & $2$: $badc$ & $3$: $dbca$ & $4$: $cbad$ & \textcolor{blue}{$5$: $dacb$}
& $6$: $cadb$ & $7$: $dcba$ & $8$: $dabc$ & $9$: $abcd$\smallskip\\
$\widehat R^3$ & $1$: $cdba$ & $2$: $badc$ & $3$: $dbca$ & $4$: $cbad$ & $5$: $adcb$
& \textcolor{blue}{$6$: $cdab$} & $7$: $dcba$ & $8$: $dabc$ & $9$: $abcd$\smallskip\\
$\widehat R^4$ & $1$: $cdba$ & \textcolor{blue}{$2$: $bdac$} & $3$: $dbca$ & $4$: $cbad$ & $5$: $adcb$
& $6$: $cadb$ & $7$: $dcba$ & $8$: $dabc$ & $9$: $abcd$\smallskip\\
$\widehat R^5$ & $1$: $cdba$ & $2$: $badc$ & $3$: $dbca$ & $4$: $cbad$ & $5$: $adcb$
& $6$: $cadb$ & \textcolor{blue}{$7$: $cdab$} & $8$: $dabc$ & $9$: $abcd$
\end{tabular}
\end{center}

The profiles $\widehat R^1,\dots,\widehat R^4$ differ from $\widehat R$ only in that a voter swaps $d$ and $a$. Consequently, the majority relation of these profiles corresponds to the ranking $dcba$. By contrast, in $\widehat R^5$, the majority relation is given by the ranking $cadb$. Hence, majority-consistency requires that $f(\widehat R^1)=\dots=f(\widehat R^4)=dcba$ and $f(\widehat R^5)=cadb$. In turn, strategyproofness between $\widehat R$ and our five profiles requires the following constraints on the ranking ${\widehat\rhd}$ chosen for $\widehat R$.
\begin{align}
    \Delta(cbad,{\widehat\rhd})&\leq \Delta(cbad,f(\widehat R^1))=\Delta(cbad,dcba)=3\\
    \Delta(adcb,{\widehat\rhd})&\leq \Delta(adcb,f(\widehat R^2))=\Delta(adcb,dcba)=3\\
    \Delta(cadb,{\widehat\rhd})&\leq \Delta(cadb,f(\widehat R^3))=\Delta(cadb,dcba)=3\\
    \Delta(badc,{\widehat\rhd})&\leq \Delta(badc,f(\widehat R^4))=\Delta(badc,dcba)=4\\
    \Delta(dcba,{\widehat\rhd})&\leq \Delta(dcba,f(\widehat R^5))=\Delta(dcba,cadb)=3.
\end{align}

Analogously to the previous claim, these constraints entail that ${\widehat\rhd}=dcba$. To see this, we first note that Conditions~(7) and~(10) show that $d$ cannot be bottom-ranked by ${\widehat\rhd}$. In more detail, every ranking $\rhd$ other than $cbad$ that bottom-ranks $d$ violates Condition~(10), since $\Delta(dcba,\rhd)\geq 4$. On the other hand, the ranking $cbad$ violates Condition~(7), as $\Delta(adcb,cbad)=4$. Next, Conditions~(6) and~(8) show that $c$ cannot be bottom-ranked by ${\widehat\rhd}$: the only ranking $\rhd$ that bottom-ranks $c$ and satisfies $\Delta(cbad,\rhd)\leq 3$ (Condition~(6)) is $badc$, but this ranking fails Condition~(8), since $\Delta(cadb,badc)=5$. Third, ${\widehat\rhd}$ cannot bottom-rank $b$: Condition~(9) rules out ${\widehat\rhd}\in\{cdab,dcab,cadb\}$, since all these rankings have a distance of at least $5$ from $badc$; Condition~(6) shows that ${\widehat\rhd}\notin\{adcb,dacb\}$, since these rankings have a Kemeny distance of at least $4$ from $cbad$; and Condition~(10) shows that ${\widehat\rhd}\neq acdb$, since $\Delta(dcba,acdb)=4$. Hence, we conclude that $a$ must be bottom-ranked in ${\widehat\rhd}$. In turn, we infer from Condition~(7) that ${\widehat\rhd}=dcba$, since every other ranking that bottom-ranks $a$ satisfies $\Delta(adcb,\rhd)\geq 4$.

\paragraph{Extension to larger values of $m$ and $n$.} Lastly, we explain how to generalize our result to larger numbers of voters $n$ and alternatives $m$. First, to increase $m$, we can add new alternatives in the same order at the bottom of the rankings of all voters. After this extension, the majority relation is still transitive and antisymmetric for all profiles $\widehat R^i$ and $\bar R^i$ with $i\in\{1,\dots,5\}$. Using majority-consistency, the inequalities derived above therefore remain intact in both cases. Finally, these inequalities still imply that we need to choose rankings for $\widehat R$ and $\bar R$ that permit a manipulation by voter $3$. To extend our construction to larger numbers of voters, we apply two different techniques. First, we can generalize our impossibility to every odd $n>9$ by adding pairs of voters with inverse rankings. These voters cancel each other out with respect to the majority relation and therefore do not affect our analysis. Second, to extend our impossibility to an even number of voters, we can double all voters in our profiles. While this requires intermediate profiles $\widetilde R$ to go from, e.g., $\bar R$ to $\bar R^1$, we can still infer the same inequalities by chaining the strategyproofness conditions. For instance, for $\bar R^1$, strategyproofness implies that $\Delta(cdba,f(\bar R))\leq\Delta(cdba,f(\widetilde R^1))\leq\Delta(cdba,f(\bar R^1))$, where $\widetilde R^1$ denotes the intermediate profile. Hence, our analysis remains intact after this extension. Lastly, for any even $n>18$, we can again add pairs of voters with inverse rankings.
\end{proof}

\begin{remark}\label{rem:Condm3}
When $m\leq 3$, \Cref{thm:PWC} ceases to hold, as the Kemeny rule (with suitable tie-breaking) is strategyproof and majority-consistent in this case \citep{ABE25a,YoLe78a}. Moreover, under mild additional conditions, namely anonymity, cancellation (i.e., adding pairs of voters with inverse rankings does not affect the outcome), and a somewhat technical variant of neutrality, it can be shown that strategyproofness and majority-consistency require the choice of a ranking that minimizes the total Kemeny distance to the input rankings when there are $3$ alternatives. We refer to \Cref{app:m3} for more details on this result.
\end{remark}

\begin{remark}
We did not minimize the number of voters in \Cref{thm:PWC}, as we aimed for a simple proof. However, with the help of a computer, we showed that this impossibility theorem already holds when there are $m=4$ alternatives and $n\in\{3,4\}$ voters. Based on our inductive arguments for $n$, it follows that no majority-consistent SWF is strategyproof if there are $m=4$ alternatives and $n\geq 3$ voters. By contrast, our inductive argument for $m$ is specific to the profiles in \Cref{thm:PWC}, and it cannot be straightforwardly applied to show that the computer proof extends to more alternatives. Unfortunately, the computer reasons about nearly one hundred profiles for the cases $m=4$ and $n\in\{3,4\}$, which makes it impossible to infer an intuitive proof from the computer-generated argument. Instead, we have verified the correctness of the computer proof, as well as the proof presented in this section, with the interactive theorem prover Isabelle/HOL \citep{EbLe26a}.
\end{remark}

\section{Strategyproofness, Unanimity, and Anonymity}\label{sec:UnanimityResults}

We now turn to our main theorem: there is no anonymous SWF that satisfies strategyproofness and unanimity if there are $m\geq 5$ alternatives and an even number of voters $n$, or when there are $m=4$ alternatives and the number of voters $n$ is a multiple of~$4$. To formally state this result, we next introduce anonymity and unanimity. Intuitively, anonymity requires that the identities of the voters do not matter for the outcome of the SWF. On the other hand, unanimity demands that if all voters prefer $x$ to $y$, the output ranking should rank $x$ ahead of $y$.

\begin{definition}[Anonymity]
    An SWF $f$ is \emph{anonymous} if $f(R)=f(R')$ for all profiles $R,R'\in\mathcal{R}^N$ for which there is a permutation $\pi:N\rightarrow N$ such that ${\succ'_i}={\succ_{\pi(i)}}$ for all $i\in N$.
\end{definition}

\begin{definition}[Unanimity]
    An SWF $f$ is \emph{unanimous} if, for all profiles $R\in\mathcal{R}^N$ and alternatives $x,y\in A$, it holds for the output ranking ${\rhd}=f(R)$ that $x\rhd y$ whenever $x\succ_i y$ for all voters $i\in N$.  
\end{definition}

We note that unanimity is sometimes also called Pareto optimality. Indeed, when voting for a single winner, an alternative is commonly considered Pareto-dominated by another alternative if all voters prefer the second alternative to the first one. However, in rank aggregation, this notion is commonly called unanimity \citep[e.g.,][]{Atha16a,ABE25a}, because we assume that voters compare rankings based on the Kemeny distance. In particular, the natural notion of Pareto optimality in this context is to require that $f$ must choose a ranking $\rhd$ for which there is no other ranking $\rhd'$ such that $\Delta(\succ_i,\rhd')\leq \Delta(\succ_i,\rhd)$ for all voters $i\in N$ and $\Delta(\succ_i,\rhd')< \Delta(\succ_i,\rhd)$ for at least one voter $i\in N$. This is a strictly stronger efficiency notion than unanimity, so an impossibility theorem based on unanimity is a stronger result.

Now, we observe that anonymity and unanimity are very basic properties that are satisfied by essentially every commonly considered SWF. Moreover, both axioms seem relevant in practice: anonymity merely describes that the SWF should not be biased towards or against particular voters, and unanimity demands that the output ranking should not contradict the voters on pairs of alternatives on which all voters agree. Thus, by showing that every anonymous and unanimous SWF fails strategyproofness, we demonstrate that every reasonable SWF is manipulable. 

\begin{theorem}\label{thm:mainimp}
    No anonymous SWF satisfies strategyproofness and unanimity if there are $m\geq 5$ alternatives and the number of voters $n$ is even, or if there are $m=4$ alternatives and $n$ is a multiple of $4$.
\end{theorem}

We have shown \Cref{thm:mainimp} based on a computer-aided theorem-proving technique called SAT solving. In the context of social choice theory, such computer-aided techniques were pioneered by \citet{TaLi09a} and have since been used to show a large number of results \citep[e.g.,][]{GeEn11a,BrGe15a,BBEG16a,End20a,BBPS21a,Pete18a,BSS19a,DDE+22a}. We refer to the survey by \citet{GePe17a} for an introduction to these techniques. In the following three sections, we outline how we apply SAT solving to obtain \Cref{thm:mainimp} (cf. \Cref{subsec:SATsolving,subsec:inductive}) and how we verify our result (cf. \Cref{subsec:verification}).

\begin{remark}
All axioms of \Cref{thm:mainimp} are necessary for the impossibility. Specifically, constant SWFs, which always return a fixed ranking, satisfy strategyproofness and anonymity but violate unanimity. Dictatorships, which return the ranking of a fixed voter, satisfy unanimity and strategyproofness but violate anonymity. It is also not possible to weaken anonymity to non-dictatorship, as \citet{ABE25a} design non-dictatorial (and non-anonymous) SWFs that are strategyproof and unanimous. Moreover, both dictatorships and the SWFs by \citet{ABE25a} satisfy a symmetry condition for the alternatives called neutrality, so it is also not possible to replace anonymity (i.e., symmetry for the voters) with symmetry for the alternatives. Third, e.g., the Kemeny rule satisfies unanimity and anonymity but violates strategyproofness. Further, we cannot significantly weaken strategyproofness, as \citet{BoSp14a} show that the Kemeny rule satisfies betweenness strategyproofness, which is a weaker condition than our Kemeny-strategyproofness. Finally, when $m\leq 3$, \Cref{thm:mainimp} ceases to hold, as the Kemeny rule is strategyproof in this case \citep{ABE25a}.
\end{remark}

\begin{remark}
A drawback of \Cref{thm:mainimp} is that we cannot extend this result to an odd number of voters. The primary reason for this is technical: we could not find an inductive argument that generalizes our theorem from an even number of voters to an odd one. Moreover, based on our SAT approach, we showed that there are SWFs that satisfy all axioms of \Cref{thm:mainimp} when $m=4$ and $n\in\{3,5\}$, which indicates that such an argument may not exist. Similar problems are common for impossibility theorems in social choice theory \citep[e.g.,][]{BrGe15a,Pete18a,BLS22c,DDE+22a,BrLe24a}, as it is often challenging to generalize such results from a fixed number of voters $n$ to arbitrary values of $n$. We note, however, that we can extend \Cref{thm:mainimp} to odd $n$ when strengthening unanimity. Specifically, based on \Cref{thm:mainimp}, one can show that no anonymous SWF satisfies strategyproofness and a property called near unanimity when $m\geq 5$ and $n\geq 3$ is odd. This latter condition requires that the output ranking puts $x$ ahead of $y$ whenever all but one voter prefer $x$ to $y$ \citep{Beno02a,Lede21c}. In particular, if we had an SWF that satisfies anonymity, strategyproofness, and near unanimity for odd $n\geq 3$, we could construct an SWF that satisfies anonymity, strategyproofness, and unanimity for $n-1$~voters by fixing the ranking of a single voter. We further observe that near unanimity becomes less demanding as $n$ increases, and all common SWFs satisfy this condition when $n$ is sufficiently large relative to $m$.
\end{remark}

\begin{remark}
It is possible to compare \Cref{thm:mainimp} to the Gibbard-Satterthwaite impossibility theorem \citep{Gibb73a,Satt75a} when using the weighted Kemeny utility functions of \citet{ALS26a}. Given a weight vector $\lambda=(\lambda_1,\dots,\lambda_m)$, these authors define the weighted Kemeny utility function of a voter with input ranking $\succ$ for the output ranking ${\rhd}=x_1\dots x_m$ by $u_\lambda(\succ,\rhd)=\sum_{j=1}^m\lambda_j\cdot |\{x_k\in A\colon x_j\succ x_k\land x_j\rhd x_k\}|$. When $\lambda=(1,\dots,1)$, the weighted Kemeny utility function precisely reduces to the unweighted Kemeny utility function $u(\succ,\rhd)=|\{(x,y)\in A^2\colon x\succ y\land x\rhd y\}|$. Since the Kemeny utility function and the Kemeny distance induce the same ordinal preferences over rankings, \Cref{thm:mainimp} shows that no SWF satisfies strategyproofness, unanimity, and anonymity in this case. By contrast, the Gibbard-Satterthwaite impossibility theorem considers voting rules that always pick a single winner. These can be captured by the weighted Kemeny utility function induced by the weight vector $\lambda=(1,0,\dots,0)$, which only cares about the top-ranked candidate of the output ranking. Hence, for this weight vector as well, we know that strategyproofness is incompatible even with weakenings of unanimity and anonymity called non-imposition and non-dictatorship. Given these negative results, it is an interesting question whether there is any weight vector $\lambda$ such that attractive strategyproof SWFs exist when the voters' utilities over rankings are given by the weighted Kemeny utility function induced by $\lambda$.
\end{remark}

\subsection{SAT Solving}\label{subsec:SATsolving}

To show \Cref{thm:mainimp}, we rely on SAT solving, a computer-aided theorem-proving technique. The central idea of this approach is that, for fixed numbers of voters $n$ and alternatives $m$, there is a large but finite number of ranking profiles and possible outcomes. For instance, when $n=2$ and $m=5$, there are $(5!)^2=14,400$ ranking profiles, for each of which one of $5!=120$ rankings must be chosen. Based on this observation, it is possible to write a large logical formula that is satisfiable if and only if there is an anonymous SWF that satisfies unanimity and strategyproofness for the given values of $n$ and $m$. We then prove two base cases of our theorem by letting a computer program, a so-called SAT solver, show that our formula is unsatisfiable when there are $n=2$ voters and $m=5$ alternatives, and when there are $n=4$ voters and $m=4$ alternatives.

In our logical formula, we follow the standard encoding of voting rules. Specifically, our formula uses variables $x_{R,\rhd}$ for all profiles $R$ and rankings $\rhd$, which will encode whether the ranking $\rhd$ is chosen for the profile $R$. Moreover, since we focus on anonymous SWFs, we treat profiles as multisets of rankings. Formally, this means that the variable $x_{R,\rhd}$ states whether the ranking $\rhd$ is chosen for all non-anonymous profiles that can be obtained by assigning the rankings in the multiset $R$ to the voters in $N$. This is possible because anonymity necessitates that we choose the same ranking for all such profiles. By representing profiles as multisets, we reduce the number of variables in our formula, as a single multiset corresponds to up to $n!$ non-anonymous ranking profiles. Also, due to this representation, anonymity is implicitly encoded, so we do not need to add constraints for this axiom.

Next, we have to ensure that our variables $x_{R,\rhd}$ indeed encode an (anonymous) SWF. This requires us to formalize that, for every profile $R$, exactly one ranking $\rhd$ is chosen. To this end, we require for each profile $R$ that exactly one variable $x_{R,\rhd}$ is true. Moreover, to further simplify our formula, we encode unanimity in this step by only introducing variables $x_{R,\rhd}$ for profiles $R$ and rankings $\rhd$ such that $\rhd$ satisfies unanimity for $R$. To make this more formal, let $X(R)=\{(x,y)\in A^2\colon \forall i\in N\colon x\succ_i y\}$ denote the pairs of alternatives $(x,y)$ such that all voters prefer $x$ to $y$ in $R$. Then, $U(R)=\{{\rhd}\in\mathcal{R}\colon \forall (x,y)\in X(R)\colon x\rhd y\}$ is the set of rankings that satisfy unanimity for $R$. We only introduce variables $x_{R,\rhd}$ for each profile $R$ and ranking $\rhd\in U(R)$, as rankings outside of $U(R)$ are not allowed to be chosen by unanimity. We can now enforce that our variables encode a unanimous and anonymous SWF by adding the following constraints for every profile $R$, which state, respectively, that at least and at most one ranking in $U(R)$ must be chosen for $R$.
\begin{align*}
    &\bigvee_{{\rhd}\in U(R)} x_{R,\rhd}\qquad \text{and}
    \qquad
    \bigwedge_{\rhd,\rhd'\in U(R)\colon {\rhd}\neq{\rhd'}} (\neg x_{R,\rhd}\lor \neg x_{R,\rhd'})
\end{align*}

Lastly, we need to encode strategyproofness. We recall that this axiom requires that, for every profile $R$ and voter $i$, it is not possible for the voter to deviate such that the ranking chosen when lying has a smaller Kemeny distance to her truthful ranking than the one chosen when reporting her truthful ranking. Put differently, if $R$ and $R'$ differ only in the ranking of voter $i$, we cannot choose rankings $\rhd$ and $\rhd'$ for these profiles such that $\Delta(\rhd,\succ_i)>\Delta(\rhd',\succ_i)$. To make this more formal, we denote by $D(R,\succ)$ the set of profiles that can be derived from $R$ by letting a voter with ranking $\succ$ deviate to an arbitrary other ranking. We note that $D(R,\succ)=\emptyset$ if no voter in $R$ reports $\succ$. Further, given two rankings $\succ$ and $\rhd$, we let $B(\succ,\rhd)=\{{\rhd'}\in\mathcal{R}\colon \Delta(\succ,\rhd')<\Delta(\succ,\rhd)\}$ denote the set of rankings that have a smaller Kemeny distance to $\succ$ than $\rhd$. Based on this notation, strategyproofness can be formalized via implications: if we choose a ranking $\rhd$ for $R$, we cannot choose a ranking $\rhd'\in B(\succ,\rhd)\cap U(R')$ for any profile $R'\in D(R,\succ)$ and ranking ${\succ}\in\mathcal{R}$. Using our variables $x_{R,\rhd}$, this results in the following constraints.
\begin{align*}
\bigwedge_{R\in\mathcal{R}^N}\,
\bigwedge_{{\succ}\in\mathcal{R}}\,
\bigwedge_{R'\in D(R,\succ)}\,
\bigwedge_{{\rhd}\in U(R)}\,
\bigwedge_{{\rhd'}\in B(\succ,\rhd)\cap U(R')}
(\neg x_{R,\rhd}\lor \neg x_{R',\rhd'})
\end{align*}

We can now write a computer program that generates this logical formula for given values of $m$ and $n$. For instance, for $m=5$ and $n=2$, our program produces a formula with $227,880$ variables and $59,445,060$ clauses. We then pass the formulas for $m=5$, $n=2$ and $m=n=4$ to a SAT solver (e.g., Glucose \citep{AuSi18a} or CaDiCaL \citep{BFF+20a}), which proves both formulas unsatisfiable in less than a minute. We hence derive the following result.

\begin{proposition}\label{prop:basecase}
No anonymous SWF satisfies both strategyproofness and unanimity if $m=5$ and $n=2$ or $m=4$ and $n=4$.
\end{proposition}

\subsection{Inductive Arguments}\label{subsec:inductive}

\Cref{prop:basecase} shows that no anonymous SWF satisfies strategyproofness and unanimity for only two cases, namely when there are $m=5$ alternatives and $n=2$ voters or $m=4$ alternatives and $n=4$ voters. By contrast, \Cref{thm:mainimp} claims that the impossibility holds for a large range of combinations of $m$ and $n$. To close this gap, we will next present a lemma that generalizes our impossibility theorem from fixed numbers of voters and alternatives to a large range. In combination with \Cref{prop:basecase}, this lemma proves \Cref{thm:mainimp}.

\begin{lemma}\label{lem:inductiveArgs}
Assume there is no anonymous SWF that satisfies strategyproofness and unanimity for $m$ alternatives and $n$ voters. The following claims are true.
\begin{enumerate}[label=(\arabic*), leftmargin = *]
\item For every $m'>m$, there is no anonymous SWF that satisfies strategyproofness and unanimity for $m'$ alternatives and $n$ voters.
\item For every $\ell\in\mathbb{N}$, there is no anonymous SWF that satisfies strategyproofness and unanimity for $m$ alternatives and $\ell n$ voters.
\end{enumerate}
\end{lemma}
\begin{proof}
Fix some values $m$ and $n$ and assume that there is no anonymous SWF satisfying strategyproofness and unanimity for $m$ alternatives and $n$ voters. For both claims of our lemma, we will show that if there was an SWF satisfying our axioms for the given parameters, then there would also be one for $m$ alternatives and $n$ voters, which contradicts our previous assumption.

\paragraph{Claim (1):} Assume for contradiction that there is an integer $m'>m$ such that there is an anonymous SWF $f$ for $m'$ alternatives and $n$ voters that satisfies strategyproofness and unanimity. Moreover, let $A'$ denote the set of $m'$ alternatives for which $f$ is defined, and let $A\subseteq A'$ be a subset of this set with $|A|=m$. We will next construct an SWF $g$ for $n$ voters and the alternatives $A$ that satisfies all our axioms. To this end, given a profile $R$ on the alternatives $A$, we define by $R^{A\rightarrow A'}$ the profile obtained from $R$ by adding the alternatives in $A'\setminus A$ in a fixed order at the bottom of the rankings of all voters. Moreover, given a ranking $\rhd$ on $A'$, we define by ${\rhd}|_A$ the ranking obtained by deleting the alternatives in $A'\setminus A$ from $\rhd$. Then, given a profile $R$ on the alternatives $A$, our new SWF $g$ first constructs the profile $R^{A\rightarrow A'}$, computes the output ranking $\rhd$ of $f$ on this profile, and finally returns the restriction of $\rhd$ to $A$. Or, more compactly, $g$ is given by $g(R)=f(R^{A\rightarrow A'})|_A$.

We next will show that $g$ satisfies anonymity, unanimity, and strategyproofness. For anonymity, we note that $\pi(R^{A\rightarrow A'})=\pi(R)^{A\rightarrow A'}$ for every permutation $\pi:N\rightarrow N$ since we extend the ranking of every voter in the same way, regardless of her reported ranking. Since $f$ is by assumption anonymous, it holds that $g(\pi(R))=f(\pi(R)^{A\rightarrow A'})|_A=f(\pi(R^{A\rightarrow A'}))|_A=f(R^{A\rightarrow A'})|_A=g(R)$ for all permutations $\pi:N\rightarrow N$, thus showing that $g$ is anonymous. Next, for unanimity, we observe that, if $x\succ_i y$ for all voters $i$ in some profile $R$ on $A$, then the same holds for the profile $R^{A\rightarrow A'}$. Since $f$ is unanimous, it therefore follows that the output ranking $g(R)=f(R^{A\rightarrow A'})|_A$ ranks $x$ ahead of $y$.

Lastly, for strategyproofness, we note that if $g$ is manipulable, then so must be $f$. To see this, assume that there are profiles $R$ and $\widehat R$ over the set of alternatives $A$ and a voter $i$ in $N$ such that $R$ and $\widehat R$ only differ in the ranking of voter $i$ and $\Delta(\succ_i, g(R))>\Delta(\succ_i, g(\widehat R))$. Further, let $\bar\succ_i$ denote voter $i$'s extended ranking in $R^{A\rightarrow A'}$. Since $f$ is unanimous and all voters in $R^{A\rightarrow A'}$ (resp. $\widehat R^{A\rightarrow A'}$) rank all alternatives in $A$ ahead of those in $A'\setminus A$ and agree on the order of the alternatives in $A'\setminus A$, the same must be true for the output ranking $f(R^{A\rightarrow A'})$  (resp. $f(\widehat R^{A\rightarrow A'})$). This implies that $\Delta({\bar \succ_i}, f(R^{A\rightarrow A'}))=\Delta({\succ_i}, g(R))>\Delta(\succ_i,g(\widehat R))=\Delta(\bar\succ_i, f(\widehat R^{A\rightarrow A'}))$. Since $R^{A\rightarrow A'}$ and $\widehat R^{A\rightarrow A'}$ only disagree in the ranking of voter $i$, we conclude that $f$ is manipulable, contradicting our assumptions. Hence, if $f$ exists, we could also construct an SWF $g$ for $m$ alternatives and $n$ voters that satisfies anonymity, unanimity, and strategyproofness, which contradicts the premise of this lemma.

\paragraph{Claim (2):} For the second case, we fix some integer $\ell\in\mathbb{N}$ and suppose that there is an anonymous SWF $f$ that satisfies strategyproofness and unanimity for $m$ voters and $\ell n$ voters. This time, we define the following transformation: given a profile $R$ on $n$ voters and $m$ alternatives, we define by $\ell R$ the profile that contains $\ell$ copies of each voter in $R$. Then, we define an SWF $g$ for $n$ voters and $m$ alternatives by $g(R)=f(\ell R)$. We first note that it is again easy to see that $g$ inherits both anonymity and unanimity from $f$. In more detail, permuting our input profile $R$ corresponds to permuting the extended profile $\ell R$ accordingly, which suffices to prove that $g$ inherits anonymity. Similarly, if all voters in $R$ prefer $x$ to $y$, the same holds for $\ell R$. Hence, the output ranking $g(R)=f(\ell R)$ also has to rank $x$ ahead of $y$ by the unanimity of $f$.

Lastly, for strategyproofness, consider two profile $R$ and $R'$ (on $n$ voters) such that $R$ and $R'$ only differ in the ranking of voter $i$. Consequently, $\ell R$ and $\ell R'$ only differ in the $\ell$ clones of voter $i$. Now, consider the sequence of profile $R^0=\ell R$, $R^1, \dots, R^\ell=\ell R'$ derived by letting the clones of voter $i$ one after another change their ranking from $\succ_i$ to $\succ_i'$. By the strategyproofness of $f$, it holds that $\Delta(\succ_i, f(R^j))\leq \Delta(\succ_i, f(R^{j+1}))$ for all $j\in \{0,\dots, \ell-1\}$. By chaining these inequalities, we get that $\Delta(\succ_i, g(R))=\Delta(\succ_i, f(\ell R))\leq \Delta(\succ_i, f(\ell R'))=\Delta(\succ_i, g(R'))$. This proves that $g$ is strategyproof if $f$ satisfies this condition. Hence, if $f$ satisfies anonymity, unanimity, and strategyproofness, so does~$g$, which again contradicts the assumption that no SWF for $m$ alternatives and $n$ voters simultaneously satisfies these conditions. 
\end{proof}

\subsection{Verification}\label{subsec:verification}

Because \Cref{prop:basecase} has been shown by SAT solving, it is not immediately clear how to verify the correctness of this result. We have chosen a threefold approach to address this issue.

First, following prior works \citep[e.g.,][]{BrGe15a,BBEG16a,Pete18a,BSS19a}, we provide a proof of \Cref{prop:basecase} for the case $m=5$ and $n=2$ in a human-readable format in \Cref{app:mainproof}. This proof was obtained by analyzing minimal unsatisfiable subsets (MUSes) of the original formula, i.e., inclusion-minimal subsets of the formula that are unsatisfiable. More intuitively, such MUSes can be seen as the reason why a formula is unsatisfiable, and they tend to be much smaller than the original formula. For instance, for their Theorem 2, \citet{BSS19a} obtain a MUS that reasons about only 13 profiles, which allows them to give a compact human-readable proof. Unfortunately, our MUSes, which we obtained using the program MUSer2 \citep{BeMa12a}, are much larger: even after several optimizations, the smallest MUS we found requires roughly 200 profiles and uses intricate reasoning. Although we were able to extract a human-readable proof using a custom computer program that translates MUSes into a readable format, the resulting proof spans more than 20 pages and requires the verification of more than 3000 applications of strategyproofness. As such, the proof allows readers to build confidence in the correctness of our SAT-based approach by inspecting some intermediate steps, but manually verifying the full argument would be extremely tedious. 
We moreover note that we cannot provide a human-readable proof for the case $m=n=4$, as the computer reasons about thousands of profiles in this case.

Because of these issues, we offer two further means of verification. Specifically, we have published our code for creating the SAT formula described in \Cref{subsec:SATsolving} \citep[][]{EbLe26b}. This enables researchers to directly check that our code correctly constructs the desired formula, which means that we only need to trust the correctness of the SAT solvers. We also note that our implementation is rather standard, so we expect researchers familiar with this type of work to be able to verify our code in less than a day.

Lastly, following more recent works \citep[e.g.,][]{BBEG16a,BSS19a,DDE+22a}, we have verified \Cref{thm:mainimp}---including both base cases and the inductive arguments---with the interactive theorem prover Isabelle/HOL \citep{NPW02a}. Such interactive theorem provers are computer programs designed to verify the correctness of mathematical proofs and have a high degree of trustworthiness. Specifically, Isabelle/HOL offers a rich mathematical logic that makes it simple to formalize our setting. Thus, our Isabelle verification directly derives \Cref{prop:basecase} as well as \Cref{lem:inductiveArgs} from our axioms. This formal verification removes the need to check the intermediate steps manually because Isabelle verifies the correctness of each deduction step based on a small and highly trustworthy set of logical operations. Consequently, to trust our results, one only needs to trust the formalization of our axioms in Isabelle. We further note that experts in verification see such formal proofs as the ``gold standard'' for increasing the trustworthiness of mathematical results \citep{HAB+17a}. Our Isabelle proof development is available in the \emph{Archive of Formal Proofs} \citep{EbLe26a}.

\section{Approximate Strategyproofness}\label{subsec:ApproxResults}

As our last contribution, we analyze how severely particular SWFs can be manipulated. Indeed, while \Cref{thm:mainimp} shows that all reasonable SWFs must be manipulable, it may be the case that voters can only gain a small amount of utility by lying about their true ranking under some rules. To formally analyze this question, we introduce the incentive ratio in \Cref{subsec:IR}, which quantifies in a multiplicative sense how much more utility a voter can gain by manipulating. Furthermore, in \Cref{subsec:distancescoring,subsec:positionalscoring,subsec:mincomp}, we analyze the incentive ratios of three natural classes of SWFs. Unfortunately, it turns out that all the SWFs we consider have a large incentive ratio and are thus severely manipulable.

\subsection{Incentive Ratio}\label{subsec:IR}

To study how severely an SWF can be manipulated, we use the incentive ratio, which measures the worst-case multiplicative utility gain of a voter when manipulating under a given SWF. This notion has been used in several private- and public-goods settings to show that appealing but manipulable mechanisms can still bound the utility gains of voters \citep[e.g.,][]{CDZZ12a,WWZ20a,LSX24a,ALV26a}. In practice, this may suffice to deter manipulations, as their costs may exceed their benefits. In particular, to identify a successful manipulation, a voter needs to spend resources and effort to gain information about the rankings of the other voters and to compute a successful manipulation. If these efforts only result in a small utility gain, it may not be beneficial for voters to manipulate.

Following prior work, we define the incentive ratio in terms of utilities and thus use the Kemeny utility function $u(\succ,\rhd)=|\{(x,y)\in A^2\colon x\succ y\land x\rhd y\}|$ instead of the Kemeny distance for this concept. Moreover, when analyzing the incentive ratio of an SWF, we allow an arbitrary number of voters, which allows us to focus on the behavior of the rule instead of the construction of a particular worst-case scenario with few voters. Finally, to define the incentive ratio, we denote by $(\succ_i',R_{-i})$ the profile derived from another profile $R$ by replacing the ranking of voter $i$ with $\succ_i'$.

\begin{definition}[Incentive Ratio]
The \emph{incentive ratio} of an SWF $f$ for $m$ alternatives is defined by
\[
\gamma_m(f)=\sup_{R,i,\succ_i'}\frac{u(\succ_i,f(\succ_i',R_{-i}))}{u(\succ_i,f(R))},
\]
where the supremum ranges over all profiles $R$ with $m$ alternatives and an arbitrary number of voters, all voters $i$ in the considered profile $R$, and all possible deviations $\succ_i'$ of voter $i$.
\end{definition}

We note that $u(\succ_i,f(R))$ can be $0$, so we use the convention that $\frac{0}{0}=1$ and $\frac{x}{0}=\infty$ for all $x>0$. Less formally, the incentive ratio measures the smallest value $\gamma$ such that a voter can obtain at most $\gamma$ times her original utility by manipulating. Hence, a smaller incentive ratio means that an SWF is less severely manipulable, and an SWF has an incentive ratio of $1$ if and only if it is strategyproof. Conversely, a large incentive ratio means that, in the worst case, voters can benefit substantially by voting strategically under the considered SWF. Lastly, we observe that an SWF has an incentive ratio of $\infty$ if and only if a voter with utility $0$ can manipulate. Otherwise, its incentive ratio for $m$ alternatives is upper-bounded by ${m\choose 2}$, since the minimal non-zero utility of a voter is $1$ and her maximal possible utility is ${m\choose 2}$. This means that the incentive ratio is also related to the notion of strategyproofness for unrepresented voters \citep{DDE+22a}, which demands that voters with $0$ utility should not be able to manipulate. Strategyproofness for unrepresented voters is desirable because voters who obtain utility $0$ from the outcome have, in contrast to other voters, no risk of obtaining a worse outcome when acting strategically.

\subsection{Positional Scoring Rules}\label{subsec:positionalscoring}

As a warm-up, we first consider the class of positional scoring rules. For these rules, the voters assign points to the alternatives depending on their positions in the input ranking, and the output ranking orders the alternatives in decreasing order of their total score. To formalize this, we define the \emph{rank} of an alternative $x$ in a ranking $\succ$ by $r(\succ,x)=1+|\{y\in A\setminus\{x\}\colon y\succ x\}|$. Then, positional scoring rules are defined by \emph{positional scoring functions} $p:\{1,\dots,m\}\rightarrow\mathbb{R}$, and a voter who places an alternative at rank $k$ assigns a score of $p(k)$ to this alternative. We require positional scoring functions to satisfy $p(1)\geq p(2)\geq\dots\geq p(m)$ and $p(1)>p(m)$, i.e., voters give at least as many points to higher-ranked alternatives and do not assign the same score to all alternatives. Finally, an SWF $f$ is a \emph{positional scoring rule} if there is a positional scoring function $p$ such that, for every profile $R$, $f$ returns the lexicographically most preferred ranking ${\rhd}$ such that $x\rhd y$ implies $\sum_{i\in N}p(r(\succ_i,x))\geq \sum_{i\in N}p(r(\succ_i,y))$. For example, the Borda rule is induced by the positional scoring function $p(k)=m-k$ for all $k\in \{1,\dots,m\}$.

As we show next, every positional scoring rule has an unbounded incentive ratio when $m\geq 3$, as voters with utility $0$ can manipulate these SWFs.

\begin{proposition}\label{prop:posscoring}
    For all $m\geq 3$, the incentive ratio of every positional scoring rule $f_{\mathit{pos}}$ is $\gamma_m(f_{\mathit{pos}})=\infty$.
\end{proposition}
\begin{proof}
To prove this claim, we fix a positional scoring rule $f$ and let $p$ denote its positional scoring function. To show that $\gamma_m(f)=\infty$, it suffices to find two profiles $R$ and $R'$ on $m$ alternatives such that $R$ and $R'$ differ only in the ranking of a single voter $i$, this voter obtains a utility of $0$ from the ranking ${\rhd}=f(R)$, and she obtains a non-zero utility from the ranking ${\rhd'}=f(R')$. To construct such profiles, we let $R^x$ denote a profile on $m-1$ voters such that \emph{(i)} alternative $x$ is top-ranked by all voters and \emph{(ii)} for every alternative $y\in A\setminus\{x\}$ and each rank $k\in\{2,\dots,m\}$, there is one voter such that $r(\succ_j,y)=k$. For these profiles, the total score of $x$, denoted by $p(R^x,x)$, is $p(R^x,x)=(m-1)p(1)$, and the total score of every other alternative $y\in A\setminus\{x\}$ is $p(R^x,y)=\sum_{k=2}^m p(k)$. Since $p(1)\geq p(2)\geq\dots\geq p(m)$ and $p(1)>p(m)$, this means that $p(R^x,x)>p(R^x,y)$ for all $y\in A\setminus\{x\}$. Next, we define $c=p(R^x,y)$ and $\delta=p(R^x,x)-p(R^x,y)$ for some pair of alternatives $x,y\in A$ with $x\neq y$ and note that $c$ and $\delta$ are independent of the choice of $x$ and $y$.

We proceed with a case distinction and first suppose that $p(1)>p(m-1)$. In this case, we enumerate the alternatives by $x_1,\dots,x_m$ and consider the profile $\widehat R$ that consists of $i$ copies of $R^{x_i}$ for all $i\in\{1,\dots,m-1\}$ and $m-1$ copies of $R^{x_m}$. For instance, if $m=3$, $\widehat R$ consists of one copy of $R^{x_1}$, two copies of $R^{x_2}$, and two copies of $R^{x_3}$. In this profile, every alternative $x_i\in A\setminus\{x_m\}$ has a score of $p(\widehat R,x_i)=(\frac{m(m+1)}{2}-1)c+i\delta$, and alternative $x_m$ has a score of $p(\widehat R,x_m)=(\frac{m(m+1)}{2}-1)c+(m-1)\delta$. Hence, it holds that $p(\widehat R,x_m)=p(\widehat R,x_{m-1})>p(\widehat R,x_{m-2})>\dots>p(\widehat R,x_1)$. Next, let $\lambda\in\mathbb{N}$ denote an integer such that $\delta\lambda>2(p(1)-p(m))$ and let $R$ denote the profile that consists of $\lambda$ copies of $\widehat R$, one voter reporting ${\succ_1}=x_1\dots x_m$, and one voter reporting ${\succ_2}=x_1\dots x_{m-2}x_mx_{m-1}$. We first note that $p(R,x_m)=\lambda p(\widehat R,x_m)+p(m-1)+p(m)=\lambda p(\widehat R,x_{m-1})+p(m-1)+p(m)=p(R,x_{m-1})$. Further, by the choice of $\lambda$, it holds that $p(R,x_{m-1})>p(R,x_{m-2})>\dots>p(R,x_1)$ because $p(R,x_i)-p(R,x_{i-1})=\lambda\delta+p(r({\succ_1},x_i))-p(r({\succ_1},x_{i-1}))+p(r({\succ_2},x_i))-p(r({\succ_2},x_{i-1}))\geq\lambda\delta-2(p(1)-p(m))>0$ for all $i\in\{2,\dots,m-1\}$. Lastly, we suppose without loss of generality that $f$ chooses the ranking ${\rhd}=x_m\dots x_1$, which means that the voter reporting $\succ_1$ obtains a utility of $0$. On the other hand, if this voter deviates to report ${\succ_1'}=x_{m-1}x_1\dots x_{m-2}x_m$, it holds for the resulting profile $R'$ that $p(R',x_{m-1})>p(R,x_{m-1})=p(R,x_m)=p(R',x_m)$ because $p(1)>p(m-1)$. This means that $x_{m-1}\rhd' x_m$ for the ranking ${\rhd'}=f(R')$ and therefore $u(\succ_1,f(R'))>0$. This completes our proof in this case.

Second, we suppose that $p(1)=p(m-1)$. Since $p(1)\geq p(2)\geq\dots\geq p(m-1)$, this means that $p(1)=p(2)=\dots=p(m-1)$. Moreover, because $p(1)>p(m)$, we conclude that $p(m-1)>p(m)$. In this case, we consider the profile $\widehat R$ that consists of one copy of $R^{x_1}$ and $i-1$ copies of $R^{x_i}$ for all $i\in\{2,\dots,m\}$. Hence, it holds in this profile that $p(\widehat R,x_1)=(\frac{m(m-1)}{2}+1)c+\delta$ and $p(\widehat R,x_i)=(\frac{m(m-1)}{2}+1)c+(i-1)\delta$ for all $i\in\{2,\dots,m\}$. This shows that $p(\widehat R,x_m)>\dots>p(\widehat R,x_2)=p(\widehat R,x_1)$. Further, we again let $\lambda$ denote an integer such that $\delta\lambda>2(p(1)-p(m))$ and define $R$ as the profile that consists of $\lambda$ copies of $\widehat R$, one voter reporting ${\succ_1}=x_1\dots x_m$, and another voter reporting ${\succ_2}=x_2x_1x_3\dots x_m$. Just as in the previous case, it can be shown that $p(R,x_m)>\dots>p(R,x_2)=p(R,x_1)$. Without loss of generality, we suppose that $f(R)={\rhd}=x_m\dots x_1$, which means that the voter reporting $\succ_1$ obtains a utility of $0$. Lastly, let $R'$ denote the profile where this voter deviates to the ranking ${\succ_1'}=x_1x_3\dots x_mx_2$. It holds for this profile that $p(R',x_1)=p(R,x_1)$, as the manipulator does not change the position of this alternative, and $p(R',x_2)<p(R,x_2)$ because $p(m)<p(2)$. Hence, we now have $p(R',x_1)>p(R',x_2)$, which implies that $x_1\rhd' x_2$ for the ranking ${\rhd'}=f(R')$. This means that $u(\succ_1,{\rhd'})>0$ and thus $u(\succ_1,{\rhd'})/u(\succ_1,{\rhd})=\infty$.
\end{proof}

\begin{remark}
    In the proof of \Cref{prop:posscoring}, we aim to use as little ties as possible. As a consequence, for many positional scoring rules, the manipulation in our proof will only increase the utility of the manipulator from $0$ to $1$. While this suffices to prove our proposition, one may argue that this is not a severe manipulation. We note, however, that for any positional scoring function with strictly decreasing positional scoring function $p$, one can easily construct profiles where a voter can have a utility of $0$ in the truthful profile and of roughly $\frac{m^2}{4}$ after the manipulation. Indeed, consider the profile where every possible ranking is reported by one voter. In this profile, all alternatives have the same score and some voter gets necessarily a utility of $0$ from the output ranking. By the assumption that $p$ is strictly decreasing, this voter can severely increase her utility by changing her reported ranking from ${\succ_i}=x_1,\dots, x_m$ to ${\succ_i'}=x_2,\dots, x_{1+\lceil (m-1)/2\rceil}, x_1, x_{2+\lceil (m-1)/2\rceil},\dots x_m$. On the other hand, for example for the Plurality rule, which is induced by the score function $p$ with $p(1)=1$ and $p(x)=0$ for $x\in \{2,\dots, m\}$, it can be shown that a voter can increase her utility at most by an additive term of $(m-2)$, since she can only increase the score of a single alternative. 
\end{remark}

\subsection{Distance Scoring Rules}\label{subsec:distancescoring}

Next, we consider a generalization of the Kemeny rule, which we call distance scoring rules. To introduce these SWFs, we recall that the Kemeny rule chooses the lexicographically most preferred ranking $\rhd$ that minimizes $\sum_{i\in N}\Delta(\succ_i,\rhd)$. Now, the idea of distance scoring rules is to apply a scoring function to the voters' Kemeny distances and then choose the ranking that minimizes the total distance score. To make this more formal, we say that a function $d:\{0,\dots,{m\choose 2}\}\rightarrow\mathbb{R}$ is a \emph{distance scoring function} if it satisfies $d(x)>d(x-1)$ for all $x\in\{1,\dots,{m\choose 2}\}$ (strict monotonicity) and $d(x)-d(x-1)\geq d(x-1)-d(x-2)$ for all $x\in\{2,\dots,{m\choose 2}\}$ (convexity). Here, the first condition formalizes that a larger distance corresponds to a strictly larger score, which represents a larger penalty in our setting. Furthermore, the convexity condition demands that a reduction in the distance to a ranking with a large swap distance is more valuable than a reduction in the distance to a ranking with a small distance. Finally, an SWF $f$ is a \emph{distance scoring rule} if there is a distance scoring function $d$ such that, for every profile $R$, $f$ chooses the lexicographically most preferred ranking that minimizes $\sum_{i\in N}d(\Delta(\succ_i,\rhd))$. For example, the Kemeny rule is the distance scoring rule induced by $d(x)=x$, and the Squared Kemeny rule by \citet{LPW24a} corresponds to the distance scoring rule induced by $d(x)=x^2$. More generally, our class contains, e.g., all distance scoring rules induced by scoring functions of the form $d(x)=x^p$ with $p\geq 1$.

We next show that, while the incentive ratio of every distance scoring rule is upper-bounded by ${m\choose 2}$, all of them are severely manipulable. In more detail, we prove that, for every distance scoring rule other than the Kemeny rule, there exists a profile in which a voter with utility $1$ can obtain her maximal utility by manipulating. Hence, the incentive ratio of these SFWs is exactly ${m\choose 2}$. On the other hand, for the Kemeny rule we prove a slightly weaker lower bound of ${m\choose 2}-m$.

\begin{proposition}\label{prop:distancescoring}
The following claims are true.
\begin{enumerate}[leftmargin = *, label = (\arabic*)]
\item For all $m\geq 4$, the incentive ratio of the Kemeny rule $f_{\mathit{K}}$ satisfies ${m\choose 2}-m\leq \gamma_m(f_{\mathit{K}})\leq {m\choose 2}$.
\item For all $m\geq 3$, the incentive ratio of every distance scoring rule $f_{\mathit{dis}}$ other than $f_{\mathit{K}}$ is $\gamma_m(f_{\mathit{dis}})={m\choose 2}$.
\end{enumerate}
\end{proposition}

\begin{proof}
We split the proof of this proposition into three independent claims. In Claim~(1), we show that all distance scoring rules have an incentive ratio of at most ${m\choose 2}$. In Claim~(2), we infer a matching lower bound for every distance scoring rule other than the Kemeny rule. Finally, in Claim~(3), we derive the lower bound on the incentive ratio of the Kemeny rule.

\paragraph{Claim (1): $\gamma_m(f)\leq {m\choose 2}$ for every distance scoring rule $f$.} First, we show that every distance scoring rule has an incentive ratio of at most ${m\choose 2}$ if there are $m$ alternatives. Equivalently, this upper bound means that voters with utility $0$ cannot manipulate these SWFs. To prove this claim, we fix a distance scoring rule $f$ and its distance scoring function $d$, a profile $R$, a voter $i$ with ranking~$\succ_i$, and let ${\rhd}=f(R)$. Moreover, we suppose that $u(\succ_i,\rhd)=0$, which means that $\succ_i$ is inverse to $\rhd$ and $\Delta(\succ_i,\rhd)={m\choose 2}$. Lastly, let $R'$ denote a profile derived from $R$ by assigning an arbitrary ranking ${\succ_i'}\neq{\succ_i}$ to voter~$i$.

As a first step, we consider a second output ranking ${\rhd'}\in\mathcal{R}\setminus\{\rhd\}$ and show that
\[
d(\Delta({\succ_i},{\rhd}))-d(\Delta({\succ_i'},{\rhd}))
\geq
d(\Delta({\succ_i},{\rhd'}))-d(\Delta({\succ_i'},{\rhd'})).
\]
To this end, we first observe that $\Delta({\succ_i},{\rhd})-\Delta({\succ_i'},{\rhd})=\Delta({\succ_i},{\succ_i'})$ since $\rhd$ and $\succ_i$ are inverse to each other. Indeed, $\succ_i'$ and $\succ_i$ disagree on a pair if and only if $\succ_i'$ and $\rhd$ agree on this pair, which yields the equality. As a consequence, we know that $\Delta({\succ_i},{\rhd})-\Delta({\succ_i'},{\rhd})>0$ since ${\succ_i}\neq{\succ_i'}$. Now, if $\Delta(\succ_i,\rhd')\leq\Delta(\succ_i',\rhd')$, we can infer our inequality by the monotonicity of $d$. In more detail, we have that $d(\Delta(\succ_i,\rhd))>d(\Delta(\succ_i',\rhd))$ because $\Delta(\succ_i,\rhd)>\Delta(\succ_i',\rhd)$ and $d(\Delta(\succ_i,\rhd'))\leq d(\Delta(\succ_i',\rhd'))$ because $\Delta(\succ_i,\rhd')\leq\Delta(\succ_i',\rhd')$. Thus, it follows that
\[
d(\Delta({\succ_i},{\rhd}))-d(\Delta({\succ_i'},{\rhd}))
\geq 0
\geq d(\Delta({\succ_i},{\rhd'}))-d(\Delta({\succ_i'},{\rhd'})).
\]

Next, assume that $\Delta({\succ_i},{\rhd'})>\Delta({\succ_i'},{\rhd'})$. In this case, we observe that $\Delta({\succ_i},{\rhd'})-\Delta({\succ_i'},{\rhd'})\leq\Delta({\succ_i},{\succ_i'})$. Indeed, for each pair $(x,y)$ with $x\succ_i y$ and $x\succ_i' y$, it either holds that $(x,y)$ contributes to neither $\Delta(\succ_i,\rhd')$ nor $\Delta(\succ_i',\rhd')$, or to both. In both cases, the pair $(x,y)$ does not increase the difference between $\Delta(\succ_i,\rhd')$ and $\Delta(\succ_i',\rhd')$, so $\Delta(\succ_i,\rhd')-\Delta(\succ_i',\rhd')$ is upper-bounded by the number of pairs on which $\succ_i$ and $\succ_i'$ disagree. Next, by the convexity of $d$, it holds that $d(x)-d(x-1)\geq d(x-1)-d(x-2)$ for all $x\in\{2,\dots,{m\choose 2}\}$. By chaining these inequalities, we derive that $d(x)-d(x-1)\geq d(y)-d(y-1)$ for all $x,y\in\{1,\dots,{m\choose 2}\}$ with $x>y$. Finally, letting $z=\Delta(\succ_i,\rhd')-\Delta(\succ_i',\rhd')$, we again derive our inequality because
\begin{align*}
d(\Delta(\succ_i,\rhd)) - d(\Delta(\succ_i',\rhd))
&=d(\Delta(\succ_i,\rhd))-d(\Delta(\succ_i,\rhd)-\Delta(\succ_i,\succ_i'))\\
&\geq d(\Delta(\succ_i,\rhd))-d(\Delta(\succ_i,\rhd)-z)\\
&=\sum_{j=0}^{z-1}
\left(
d(\Delta(\succ_i,\rhd)-j)-d(\Delta(\succ_i,\rhd)-j-1)
\right)\\
&\geq \sum_{j=0}^{z-1}
\left(
d(\Delta(\succ_i,\rhd')-j)-d(\Delta(\succ_i,\rhd')-j-1)
\right)\\
&=d(\Delta(\succ_i,\rhd'))-d(\Delta(\succ_i,\rhd')-z)\\
&=d(\Delta(\succ_i,\rhd'))-d(\Delta(\succ_i',\rhd')).
\end{align*}

Here, the first inequality follows because $z\leq\Delta(\succ_i,\succ_i')$ and $d$ is monotonic, while the second inequality follows from convexity.

Based on our auxiliary inequality, we now show that voters with utility $0$ cannot affect the outcome of $f$ by manipulating. To this end, we observe that $\sum_{j\in N}d(\Delta({\succ_j},{\rhd}))\leq\sum_{j\in N}d(\Delta({\succ_j},{\rhd'}))$ for all ${\rhd'}\in\mathcal{R}$ since $f(R)={\rhd}$. Hence, we derive for all ${\rhd'}\in\mathcal{R}$ that
\begin{align*}
\sum_{j\in N} d(\Delta(\succ_j',\rhd))
&=
\sum_{j\in N} d(\Delta(\succ_j,\rhd))
-\left(d(\Delta(\succ_i,\rhd))-d(\Delta(\succ_i',\rhd))\right)\\
&\leq
\sum_{j\in N} d(\Delta(\succ_j,\rhd'))
-\left(d(\Delta(\succ_i,\rhd'))-d(\Delta(\succ_i',\rhd'))\right)\\
&=
\sum_{j\in N} d(\Delta(\succ_j',\rhd')).
\end{align*}

This proves that $\rhd$ minimizes the total score in $R'$. Further, if this inequality is tight for some ranking $\rhd'$, then $\sum_{j\in N}d(\Delta({\succ_j},{\rhd}))=\sum_{j\in N}d(\Delta({\succ_j},{\rhd'}))$. In this case, $f(R)={\rhd}$ implies that $\rhd$ is lexicographically preferred to $\rhd'$. Hence, $f(R')={\rhd}$, and voters with utility $0$ cannot manipulate~$f$.

\paragraph{Claim (2): $\gamma_m(f)\geq {m\choose 2}$ for every distance scoring rule $f$ other than $f_\mathit{K}$.} Next, we show that every distance scoring rule $f$ other than the Kemeny rule has an incentive ratio of at least ${m\choose 2}$ when there are $m$ alternatives. To prove this claim, we fix such an SWF $f\neq f_{\mathit K}$ and let $d$ denote its corresponding distance scoring function. Throughout this proof, we use the notation $d(R,\rhd)$ to denote the total distance score $\sum_{i\in N}d(\Delta(\succ_i,\rhd))$ of the ranking $\rhd$ in the profile $R$. Furthermore, we let $\rhd_1$ and $\rhd_2$ be two arbitrary rankings such that $\Delta(\rhd_1,\rhd_2)={m\choose 2}-1$, which means that there is a single pair of alternatives $x,y$ such that $x\rhd_1 y$ and $x\rhd_2 y$. Lastly, we choose $\blacktriangleleft_1$ and $\blacktriangleleft_2$ as the reverse rankings of $\rhd_1$ and $\rhd_2$, i.e., these satisfy $\Delta(\rhd_1,\blacktriangleleft_1)=\Delta(\rhd_2,\blacktriangleleft_2)={m\choose 2}$ and $\Delta(\rhd_1,\blacktriangleleft_2)=\Delta(\rhd_2,\blacktriangleleft_1)=1$.

As a first step, we construct a profile $R^*$ such that precisely $\rhd_1$ and $\rhd_2$ minimize the total distance score. For this, we first note that, by the convexity of $d$, for all $x,y\in\{1,\dots,{m\choose 2}\}$ with $x>y$, it holds that $d(x)-d(x-1)\geq d(x-1)-d(x-2)\geq\dots\geq d(y)-d(y-1)$. Put differently, this means that $d(x)+d(y-1)\geq d(x-1)+d(y)$ for all such $x$ and $y$. By choosing $x\in\{\lceil{m\choose 2}/2\rceil+1,\dots,{m\choose 2}\}$ and $y={m\choose 2}-(x-1)$, we derive from this that $d(x)+d({m\choose 2}-x)\geq d(x-1)+d({m\choose 2}-(x-1))$. Thus, it holds that $d({m\choose 2})+d(0)\geq d({m\choose 2}-1)+d(1)\geq\dots\geq d(\lceil {m\choose 2}/2\rceil)+d(\lfloor {m\choose 2}/2\rfloor)$.
Furthermore, since $f$ is not the Kemeny rule, $d$ cannot satisfy $d(x)-d(x-1)=d(x-1)-d(x-2)$ for all $x\in\{2,\dots,{m\choose 2}\}$. Convexity thus requires that at least one of these inequalities is strict. By chaining our inequalities, this shows that $d({m\choose 2})-d({m\choose 2}-1)>d(1)-d(0)$ and therefore $d({m\choose 2})+d(0)>d({m\choose 2}-1)+d(1)$.

Now, consider the profile $R^1$ where $\rhd_1$ and $\rhd_2$ are reported by one voter each and $\blacktriangleleft_1$ and $\blacktriangleleft_2$ are reported by two voters each. We next compute the total distance scores of all possible rankings.
\begin{itemize}[]
\item The total distance scores of $\rhd_1$ and $\rhd_2$ are given by $d(R^1,\rhd_1)=d(R^1,\rhd_2)=d(0)+2d(1)+d({m\choose 2}-1)+2d({m\choose 2})$. Indeed, $\rhd_1$ has Kemeny distance $0$ to $\rhd_1$, $1$ to $\blacktriangleleft_2$, ${m\choose 2}-1$ to $\rhd_2$, and ${m\choose 2}$ to $\blacktriangleleft_1$, and the situation for $\rhd_2$ is symmetric.
\item The total distance scores of $\blacktriangleleft_1$ and $\blacktriangleleft_2$ are $d(R^1,\blacktriangleleft_1)=d(R^1,\blacktriangleleft_2)=2d(0)+d(1)+2d({m\choose 2}-1)+d({m\choose 2})$. This follows since $\Delta(\blacktriangleleft_1,\blacktriangleleft_1)=0$, $\Delta(\blacktriangleleft_1,\blacktriangleleft_2)=\Delta(\rhd_1,\rhd_2)={m\choose 2}-1$, $\Delta({\blacktriangleleft_1},\rhd_1)={m\choose 2}$, and $\Delta(\blacktriangleleft_1,\rhd_2)={m\choose 2}-\Delta(\blacktriangleleft_1,\blacktriangleleft_2)=1$. The analysis for $\blacktriangleleft_2$ is again symmetric. Furthermore, the strict monotonicity of $d$ implies that $d(0)+2d(1)+d({m\choose 2}-1)+2d({m\choose 2})>2d(0)+d(1)+2d({m\choose 2}-1)+d({m\choose 2})$, and so $d(R^1,\rhd_1)=d(R^1,\rhd_2)>d(R^1,\blacktriangleleft_1)=d(R^1,\blacktriangleleft_2)$.
\item Next, we consider a ranking ${\rhd}\in\mathcal{R}\setminus\{\rhd_1,\rhd_2,\blacktriangleleft_1,\blacktriangleleft_2\}$ with $x\rhd y$ (where $x$ and $y$ are the alternatives such that $x\rhd_1 y$ and $x\rhd_2 y$). Since $\rhd_1$, $\rhd_2$, and $\rhd$ agree on the order of $x$ and $y$ and $\rhd\notin\{\rhd_1,\rhd_2,\blacktriangleleft_1,\blacktriangleleft_2\}$, there is an integer $z\in\{1,\dots,{m\choose 2}-2\}$ such that $\Delta(\rhd,\rhd_1)=z$ and $\Delta(\rhd,\rhd_2)={m\choose 2}-1-z$. Moreover, this implies that $\Delta(\rhd,\blacktriangleleft_1)={m\choose 2}-z$ and $\Delta(\rhd,\blacktriangleleft_2)=z+1$. Thus, the total distance score of $\rhd$ is
$d(R^1,\rhd)=d(z)+2d({m\choose 2}-z)+2d(z+1)+d({m\choose 2}-1-z)$.
Because $z\in\{1,\dots,{m\choose 2}-2\}$, our insights on the convexity of $d$ yield that $d(z)+d({m\choose 2}-z)\leq d(1)+d({m\choose 2}-1)$ and $d(z+1)+d({m\choose 2}-z-1)<d(0)+d({m\choose 2})$. Furthermore, the convexity of $d$ implies that $d({m\choose 2}-z)+d(z+1)\leq d({m\choose 2})+d(1)$. To see this, first assume that $z\leq ({m\choose 2}-1)/2$. In this case, convexity directly yields that
$d({m\choose 2}-z)+d(z+1)\leq d({m\choose 2}-z+1)+d(z)\leq\dots\leq d({m\choose 2})+d(1)$. On the other hand, if $z>({m\choose 2}-1)/2$, we let $z'={m\choose 2}-1-z$ and observe that $z'\leq ({m\choose 2}-1)/2$ and $d({m\choose 2}-z)+d(z+1)=d(z'+1)+d({m\choose 2}-z')$.
Hence, by repeating the previous argument with $z'$, we again get that $d({m\choose 2}-z)+d(z+1)\leq d({m\choose 2})+d(1)$. Finally, by combining our insights, it follows that
$d(R^1,\rhd)
=d(z)+2d({m\choose 2}-z)+2d(z+1)+d({m\choose 2}-1-z)
<d(0)+2d(1)+d({m\choose 2}-1)+2d({m\choose 2})$,
where the right-hand side is the score of $\rhd_1$ and $\rhd_2$.
\item Lastly, consider a ranking ${\rhd}\in\mathcal{R}\setminus\{\rhd_1,\rhd_2,\blacktriangleleft_1,\blacktriangleleft_2\}$ with $y\rhd x$. Similar to the previous case, there is now a parameter $z\in\{1,\dots,{m\choose 2}-2\}$ such that $\Delta(\rhd,\blacktriangleleft_1)=z$ and $\Delta(\rhd,\blacktriangleleft_2)={m\choose 2}-1-z$, which further implies that $\Delta(\rhd,\rhd_1)={m\choose 2}-z$ and $\Delta(\rhd,\rhd_2)=z+1$. Thus, the total distance score of $\rhd$ is
$d(R^1,\rhd)=2d(z)+d(z+1)+2d({m\choose 2}-z-1)+d({m\choose 2}-z)$.
By the strict monotonicity of $d$, it holds that
$d(R^1,\rhd)
<d(z)+2d(z+1)+d({m\choose 2}-z-1)+2d({m\choose 2}-z)$.
This is the same total score as in the previous case, so we again conclude that $d(R^1,\rhd)<d(R^1,\rhd_1)=d(R^1,\rhd_2)$.
\end{itemize}

In conclusion, this analysis shows that $\rhd_1$ and $\rhd_2$ are the unique maximizers of the total distance score in $R^1$. Now, let $R^2$ denote the profile where every possible input ranking is reported by two voters. For this profile, it holds by symmetry that $d(R^2,\rhd)=d(R^2,\rhd')$ for all rankings $\rhd,\rhd'\in\mathcal{R}$. To ease notation, we define this common score as $D$. Next, let $R^*$ denote the profile derived from $R^2$ by deleting the subprofile $R^1$, i.e., for each voter $i$ in $R^1$, we delete one voter in $R^2$ who reports the ranking $\succ_i^1$. For this profile, the total distance score of every ranking $\rhd$ is $d(R^*,\rhd)=D-d(R^1,\rhd)$. Thus, for $R^*$, the rankings $\rhd_1$ and $\rhd_2$ uniquely minimize this score.

We construct our final profile $R$ as follows. First, we choose an integer $\lambda$ such that $\lambda d(R^*,\rhd_1)+2(d({m\choose 2})-d(0))<\lambda d(R^*,\rhd)$ for every ranking ${\rhd}\in\mathcal{R}\setminus\{\rhd_1,\rhd_2\}$. Then, $R$ consists of $\lambda$ copies of $R^*$, one voter reporting $\rhd_1$, and one voter reporting $\rhd_2$. By the choice of $\lambda$, it hold for every ranking ${\rhd}\in\mathcal{R}\setminus\{\rhd_1,\rhd_2\}$ that $d(R,\rhd_1)
\leq \lambda d(R^*,\rhd_1)+2d({m\choose 2})
<\lambda d(R^*,\rhd)+2d(0)
\leq d(R,\rhd)$.
Further, the scores of $\rhd_1$ and $\rhd_2$ are equal, as
$d(R,\rhd_1)
=\lambda d(R^*,\rhd_1)+d(0)+d({m\choose 2}-1)
=\lambda d(R^*,\rhd_2)+d({m\choose 2}-1)+d(0)
=d(R,\rhd_2)$.
Hence, $f$ either chooses $\rhd_1$ or $\rhd_2$ for $R$. Without loss of generality, assume that $f$ chooses $\rhd_1$. Now, let $R'$ denote the profile derived from $R$ by changing the preference relation of the voter who reports $\rhd_2$ to $\blacktriangleleft_1$. An analogous argument as before shows that $d(R',\rhd_1)<d(R',\rhd)$ for all ${\rhd}\in\mathcal{R}\setminus\{\rhd_1,\rhd_2\}$. Thus, $f$ either chooses $\rhd_1$ or $\rhd_2$ for $R'$. Furthermore, it holds that
\[
\begin{aligned}
d(R',\rhd_1)-d(R',\rhd_2)
&=d(\Delta(\rhd_1,\rhd_1))+d(\Delta(\rhd_1,\blacktriangleleft_1))
-d(\Delta(\rhd_2,\rhd_1))-d(\Delta(\rhd_2,\blacktriangleleft_1))\\
&=d(0)+d({m\choose 2})-d({m\choose 2}-1)-d(1).
\end{aligned}
\]
Now, we have shown at the beginning that $d(0)+d({m\choose 2})>d(1)+d({m\choose 2}-1)$, so $d(R',\rhd_1)>d(R',\rhd_2)$ and $f(R')={\rhd_2}$. This means that $u(\rhd_2,f(R))=1$ and $u(\rhd_2,f(R'))={m\choose 2}$ and thus $\gamma_m(f)\geq {m\choose 2}$.

\paragraph{Claim (3): $\gamma_m(f_{\mathit{K}})\geq {m\choose 2}-m$.} Finally, we show that the incentive ratio of the Kemeny rule is at least ${m\choose 2}-m$. At a high level, we use a similar strategy as for Claim~(2) to prove this claim. In particular, we first find a profile $R^*$ where two rankings $\rhd_1$ and $\rhd_2$ with $\Delta(\rhd_1,\rhd_2)={m\choose 2}-m+1$ minimize the total Kemeny distance and then use this profile to construct a manipulation that implies our lower bound on the incentive ratio. However, to derive the profile $R^*$, we have to consider the geometry of the Kemeny rule in more detail. Specifically, in \Cref{lem:k2}, we characterize the pairs of rankings $\rhd_1,\rhd_2\in\mathcal{R}$ for which there is a profile $R$ such that $\sum_{i\in N}\Delta(\succ_i,\rhd_1)=\sum_{i\in N}\Delta(\succ_i,\rhd_2)<\sum_{i\in N}\Delta(\succ_i,\rhd)$ for all $\rhd\in\mathcal{R}\setminus\{\rhd_1,\rhd_2\}$. To state this result, we need to introduce the transposition graph $G_{\rhd_1,\rhd_2}=(V,E)$ for two rankings $\rhd_1$ and $\rhd_2$. In this graph, the vertices correspond to the pairs of alternatives on which $\rhd_1$ and $\rhd_2$ disagree, i.e., $V=\{\{x,y\}\subseteq A\colon x\rhd_1 y\iff y\rhd_2 x\}$. Moreover, for any two vertices $v_1,v_2\in V$, there is an edge $\{v_1,v_2\}\in E$ if and only if there are three distinct alternatives $x,y,z\in A$ such that $v_1=\{x,y\}$, $v_2=\{y,z\}$, and $\{x,z\}\notin V$. Put differently, there is an edge between two vertices $\{x,y\}$ and $\{y,z\}$ if and only if the triple on $x,y,z$ is not fully inverted. Now, in \Cref{app:approx}, we prove the following lemma.

\begin{restatable}{lemma}{kemenypairs}\label{lem:k2}
For every pair of rankings $\rhd_1,\rhd_2\in\mathcal{R}$, there is a profile $R$ such that $\{\rhd_1,\rhd_2\}=\arg\min_{\rhd\in\mathcal{R}}\sum_{i\in N}\Delta(\succ_i,\rhd)$ if and only if the transposition graph $G_{\rhd_1,\rhd_2}$ is connected.
\end{restatable}

Based on this lemma, we next aim to find two rankings $\rhd_1$ and $\rhd_2$ such that $G_{\rhd_1,\rhd_2}$ is connected and $\Delta(\rhd_1,\rhd_2)={m\choose 2}-m+1$. To derive these rankings, we use induction on the number of alternatives $m\geq 4$. First, for $m=4$, let $\rhd_1=x_1x_2x_3x_4$ and $\rhd_2=x_3x_1x_4x_2$. These rankings disagree on the three pairs $\{x_1,x_3\}$, $\{x_2,x_3\}$, and $\{x_2,x_4\}$. Hence, the transposition graph $G_{\rhd_1,\rhd_2}$, which is displayed in \Cref{fig:transposition-graphs}, is connected. Moreover, it holds that $\Delta(\rhd_1,\rhd_2)=3={4\choose 2}-4+1$. For $m=5$, let ${\rhd_1}=x_1x_2x_3x_4x_5$ and ${\rhd_2}=x_4x_2x_5x_1x_3$. In this case, the rankings disagree on the pairs $\{x_1,x_2\}$, $\{x_1,x_4\}$, $\{x_2,x_4\}$, $\{x_3,x_4\}$, $\{x_1,x_5\}$, and $\{x_3,x_5\}$. The transposition graph is again shown in \Cref{fig:transposition-graphs}, and it is easily seen to be connected. We also note that $\Delta(\rhd_1,\rhd_2)=6={5\choose 2}-5+1$, as required.

\begin{figure}
\begin{tikzpicture}[
vertex/.style={circle, minimum size=9mm, inner sep=1pt},
every edge/.style={draw, thick}
]

\begin{scope}[xshift=0cm]
\node[vertex] (13) at (0,0) {$\{x_1,x_3\}$};
\node[vertex] (23) at (2.5,0) {$\{x_2,x_3\}$};
\node[vertex] (24) at (5,0) {$\{x_2,x_4\}$};

\draw (13) -- (23);
\draw (23) -- (24);

\node at (2.5,-1.1) {$m=4$};
\end{scope}

\begin{scope}[xshift=7.5cm]
\node[vertex] (12) at (0,0) {$\{x_1,x_2\}$};
\node[vertex] (15) at (2,0) {$\{x_1,x_5\}$};
\node[vertex] (14) at (3.5,1.5) {$\{x_1,x_4\}$};
\node[vertex] (35) at (3.5,-1.5) {$\{x_3,x_5\}$};
\node[vertex] (34) at (5,0) {$\{x_3,x_4\}$};
\node[vertex] (24) at (7,0) {$\{x_2,x_4\}$};

\draw (12) -- (15);
\draw (15) -- (14);
\draw (15) -- (35);
\draw (14) -- (34);
\draw (35) -- (34);
\draw (34) -- (24);

\node at (3.5,-2.6) {$m=5$};
\end{scope}
\end{tikzpicture}
\caption{The transposition graphs for the rankings used in the proof of \Cref{prop:distancescoring} for $m=4$ and $m=5$.}
\label{fig:transposition-graphs}
\end{figure}

For $m\geq 6$, we construct our rankings $\rhd_1$ and $\rhd_2$ inductively. To this end, let $(\rhd_1',\rhd_2')$ denote the pair of rankings we constructed for $m-2$ alternatives. We inductively assume that $\Delta(\rhd_1',\rhd_2')={m-2\choose 2}-(m-2)+1$ and that $G_{\rhd_1',\rhd_2'}$ is connected. Now, the ranking $\rhd_1$ is given by ${\rhd_1}=x_1\dots x_m$. Furthermore, $\rhd_2$ is defined as follows: it agrees on all alternatives $\{x_1,\dots,x_{m-2}\}$ with $\rhd_2'$, it top-ranks $x_{m-1}$, and it places $x_m$ third. By the inductive construction and our base cases, this means that $x_{m-3}$ is ranked second by $\rhd_2$. For example, $\rhd_2$ is given by $x_5x_3x_6x_1x_4x_2$ for $m=6$, by $x_6x_4x_7x_2x_5x_1x_3$ when $m=7$, and by $x_7x_5x_8x_3x_6x_1x_4x_2$ when $m=8$. We note that $\rhd_2$ satisfies the following three properties by its inductive definition and the choice of our base cases: \emph{(I1)} $x_{m-2i-1}\rhd_2 x_{m-2i}$ for every $i\in\{0,\dots,\lfloor m/2\rfloor-1\}$, \emph{(I2)} $x_{m-3-2i}\rhd_2 x_{m-2i}$ for all $i\in\{0,\dots,\lfloor m/2\rfloor-2\}$, and \emph{(I3)} $x_1\rhd_2 x_3$ if $m$ is odd.

We next show that $\Delta(\rhd_1,\rhd_2)={m\choose 2}-m+1$ and that $G_{\rhd_1,\rhd_2}$ is connected. First, for the claim on the Kemeny distance, we note that $\rhd_1$ and $\rhd_2$ agree with $\rhd_1'$ and $\rhd_2'$, respectively, when restricted to $\{x_1,\dots,x_{m-2}\}$. Hence, for pairs that only involve these alternatives, $\rhd_1$ and $\rhd_2$ disagree on ${m-2\choose 2}-(m-2)+1$ pairs of alternatives. Next, $\rhd_1$ and $\rhd_2$ disagree on $m-2$ pairs of alternatives involving $x_{m-1}$ and on $m-3$ pairs involving $x_m$. Based on these insights, we compute that
\begin{align*}
\Delta(\rhd_1,\rhd_2)&={m-2\choose 2}-(m-2)+1+(m-2)+(m-3)\\
&=\sum_{i=1}^{m-3}i-m+(m-1)+(m-2)+1\\
&=\sum_{i=1}^{m-1}i-m+1\\
&={m\choose 2}-m+1.
\end{align*}

As a second step, we prove that $G_{\rhd_1,\rhd_2}$ is connected. To this end, we again observe that $\rhd_1$ and $\rhd_2$ agree with $\rhd_1'$ and $\rhd_2'$, respectively, when restricted to $\{x_1,\dots,x_{m-2}\}$. This means that $G_{\rhd_1,\rhd_2}$ agrees with $G_{\rhd_1',\rhd_2'}$ when restricted to the vertices $v$ with $v\subseteq\{x_1,\dots,x_{m-2}\}$. Thus, by the induction hypothesis, for every pair of vertices $v_1$ and $v_2$ with $v_1\subseteq\{x_1,\dots,x_{m-2}\}$ and $v_2\subseteq\{x_1,\dots,x_{m-2}\}$, there is a path from $v_1$ to $v_2$ in $G_{\rhd_1,\rhd_2}$. Next, we consider the pairs involving $x_{m-1}$. We first note that $\{x_i,x_{m-1}\}\in V$ for all $i\in\{1,\dots,m-2\}$ by construction. Now, based on (I1), it holds that $x_{m-1}\rhd_2 x_{m-2i-1}\rhd_2 x_{m-2i}$ and $x_{m-2i-1}\rhd_1 x_{m-2i}\rhd_1 x_{m-1}$ for all $i\in\{1,\dots,\lfloor m/2\rfloor-1\}$. This means that there is an edge from $\{x_{m-2i-1},x_{m-1}\}$ to $\{x_{m-2i},x_{m-1}\}$ in $G_{\rhd_1,\rhd_2}$. Further, by (I2), it holds for all $i\in\{1,\dots,\lfloor m/2\rfloor-2\}$ that $x_{m-1}\rhd_2 x_{m-3-2i}\rhd_2 x_{m-2i}$ and $x_{m-3-2i}\rhd_1 x_{m-2i}\rhd_1 x_{m-1}$. Consequently, there is an edge from $\{x_{m-3-2i},x_{m-1}\}$ to $\{x_{m-2i},x_{m-1}\}$ in $G_{\rhd_1,\rhd_2}$. If $m$ is even, this implies that all vertices involving $x_{m-1}$ are connected with each other: $\{x_{m-1},x_{m-3}\}$ has an edge to $\{x_{m-1},x_{m-2}\}$, $\{x_{m-1},x_{m-2}\}$ has an edge to $\{x_{m-1},x_{m-5}\}$, $\{x_{m-1},x_{m-5}\}$ has an edge to $\{x_{m-1},x_{m-4}\}$, and so on. Moreover, if $m$ is odd, this sequence remains valid except for $x_1$. In this case, we use (I3): it holds that $x_{m-1}\rhd_2 x_1\rhd_2 x_3$ and $x_1\rhd_1 x_3\rhd_1 x_{m-1}$, so there is an edge from $\{x_{m-1},x_3\}$ to $\{x_{m-1},x_1\}$. Next, we note that an analogous argument shows that all vertices containing $x_m$ are connected with each other. We further note that $x_{m-1}\rhd_2 x_m\rhd_2 x_{m-2}$ and $x_{m-2}\rhd_1 x_{m-1}\rhd_1 x_m$, so there is an edge from $\{x_m,x_{m-2}\}$ to $\{x_{m-1},x_{m-2}\}$. Hence, any pair of vertices that contain $x_m$ or $x_{m-1}$ can reach any other such vertex via the paths we identified. Lastly, we note that $x_{m-3}\rhd_2 x_m\rhd_2 x_{m-4}$ and $x_{m-4}\rhd_1 x_{m-3}\rhd_1 x_m$ since $m\geq 6$. Moreover, $\{x_{m-3},x_{m-4}\}$ is a vertex of $G_{\rhd_1,\rhd_2}$ because, in the construction for $m-2$ alternatives, $x_{m-3}$ is top-ranked by $\rhd_2'$ whereas $\rhd_1'$ ranks $x_{m-4}$ ahead of $x_{m-3}$. Hence, there is an edge from $\{x_m,x_{m-4}\}$ to $\{x_{m-3},x_{m-4}\}$. Since we have derived from the induction hypothesis that all vertices $v=\{x_i,x_j\}$ with $i,j\in\{1,\dots,m-2\}$ are connected with each other, this implies that $G_{\rhd_1,\rhd_2}$ is indeed connected.

By \Cref{lem:k2}, we now infer that, for every $m\geq 4$, there is a profile $R^*$ for which precisely our two rankings $\rhd_1$ and $\rhd_2$ minimize the total Kemeny distance. We subsequently assume that the tie-breaking prefers $\rhd_1$ to $\rhd_2$, so the Kemeny rule chooses $\rhd_1$ for $R^*$. However, if it were to choose $\rhd_2$ for this profile, the following construction can be adapted by exchanging the roles of the two rankings. Now, let ${\blacktriangleleft}=x_m\dots x_4x_2x_3x_1$ if $m$ is even and ${\blacktriangleleft}=x_m\dots x_3x_1x_2$ if $m$ is odd. We note that these rankings are chosen such that $\Delta(\blacktriangleleft,\rhd_1)={m\choose 2}-1$ and $\Delta(\blacktriangleleft,\rhd_2)={m\choose 2}-\Delta(\rhd_1,\rhd_2)+1=m$. Furthermore, we consider the profile $R$ derived from $R^*$ by adding a voter reporting $\blacktriangleleft$ and one voter reporting the reverse ranking $\blacktriangleleft^{-1}$. Since $\Delta(\blacktriangleleft,\rhd)+\Delta(\blacktriangleleft^{-1},\rhd)={m\choose 2}$ for every ranking $\rhd$, it follows that $\rhd_1$ and $\rhd_2$ still uniquely minimize the total Kemeny distance, and so $\rhd_1$ is chosen by our tie-breaking assumption. Next, let $R'$ be the profile derived from $R$ by changing the ranking of the voter reporting $\blacktriangleleft$ to the inverse ranking of $\rhd_1$, i.e., this voter swaps $x_2$ and $x_3$ if $m$ is even and $x_1$ and $x_2$ if $m$ is odd. This increases the Kemeny distance to $\rhd_1$ by $1$ and reduces the Kemeny distance to $\rhd_2$ by $1$. Hence, $\rhd_2$ now has a smaller total Kemeny distance than $\rhd_1$. Moreover, $\rhd_2$ still minimizes the total Kemeny distance to the input rankings. Indeed, before the deviation, every other ranking has total Kemeny distance at least one larger than that of $\rhd_2$, while changing the report by a single adjacent swap changes the distance to every ranking by at most $1$. Since the total distance of $\rhd_2$ decreases by $1$, it follows that after the deviation $\rhd_2$ has strictly smaller total distance than every other ranking. Hence, the Kemeny rule chooses $\rhd_2$ for $R'$. Finally, since $u(\blacktriangleleft,\rhd_1)={m\choose 2}-\Delta(\blacktriangleleft,\rhd_1)=1$ and $u(\blacktriangleleft,\rhd_2)={m\choose 2}-\Delta(\blacktriangleleft,\rhd_2)={m\choose 2}-m$, we derive that the incentive ratio of the Kemeny rule is at least ${m\choose 2}-m$.
\end{proof}

\begin{remark}
Distance scoring rules can also be formulated via the Kemeny utility function. To see this, we define the utility scoring function $t$ induced by a distance scoring function $d$ by $t(u(\succ,\rhd))=-d\left({m\choose 2}-u(\succ,\rhd)\right)=-d(\Delta(\succ,\rhd))$. Then, the ranking $\rhd$ that maximizes the total utility score $\sum_{i\in N}t(u(\succ_i,\rhd))$ is the same as the one that minimizes the total distance score $\sum_{i\in N}d(\Delta(\succ_i,\rhd))$. We moreover note that, if $d$ is strictly increasing and convex, then $t$ is strictly increasing and concave. In this formulation, the class of distance scoring rules closely resembles the class of Thiele rules that is frequently studied in approval-based committee voting \citep[see, e.g.,][]{LaSk23a}.
\end{remark}

\begin{remark}
In \Cref{lem:maxdistant} in \Cref{app:approx}, we show that, when $m\geq 5$, there is no pair of rankings $\rhd_1$ and $\rhd_2$ such that $\Delta(\rhd_1,\rhd_2)>{m\choose 2}-m+1$ and $\{\rhd_1,\rhd_2\}=\arg\min_{\rhd\in\mathcal{R}}\sum_{i\in N}\Delta(\succ_i,\rhd)$ for some profile $R$. This means that the rankings we identified for the lower bound on the incentive ratio of the Kemeny rule are maximally far apart subject to the condition that precisely these two rankings minimize the total Kemeny distance. Hence, we cannot improve the lower bound on the incentive ratio of the Kemeny rule by simply choosing a different pair of rankings in our proof. Moreover, based on this insight, it can be shown that the incentive ratio of the Kemeny rule is at most ${m\choose 2}-m+1$ when restricting attention to profiles for which there is a unique Kemeny ranking either before or after the manipulation.
\end{remark}

\subsection{Rank Vector Scoring Rules}\label{subsec:mincomp}

As our last class of SWFs, we consider rank vector scoring rules. Intuitively, these rules assign to each alternative a score based on how often the alternative appears at each rank. To make this more formal, we recall that $r(\succ,x)=1+|\{y\in A\setminus\{x\}\colon y\succ x\}|$ is the rank of alternative $x$ in a ranking~$\succ$. Then, the rank vector of an alternative $x$ in a profile $R$ is defined by $r^*(R,x)=(|\{i\in N\colon r(\succ_i,x)=k\}|)_{k\in\{1,\dots,m\}}$. Less formally, $r^*_k(R,x)$ states how often alternative $x$ is ranked at position $k$ in the profile $R$. Further, rank vector scoring rules are defined by a rank vector scoring function $s$, which maps each rank vector to a score in $\mathbb{R}$ and orders the alternatives in decreasing order of their scores. More formally, an SWF $f$ is a \emph{rank vector scoring rule} if there is a rank vector scoring function $s$ such that, for all profiles $R$ and alternatives $x,y\in A$, it holds for the ranking ${\rhd}=f(R)$ that $x\rhd y$ whenever $s(r^*(R,x))>s(r^*(R,y))$. As usual, we assume that ties between alternatives with the same score are broken by our external tie-breaking order $\mathbf{>}$.

To gain more intuition for the class of rank vector scoring rules, we first note that it generalizes positional scoring rules, as we can compute positional scores based on the rank vector. However, the class is much richer, since we are not restricted to simply summing up scores. Instead, we can, e.g., return the minimum, maximum, or median of the (transformed) ranks assigned by the voters to the alternative or, more generally, any ordered weighted average of the voters' (transformed) ranks. We can also use non-linear functions and, e.g., reward an alternative for appearing particularly often at the same position by squaring $r^*_k(R,x)$. That is, one can consider the rank vector scoring function given by $s(r^*(R,x))=\sum_{k=1}^m r^*_k(R,x)^2\cdot(m-k)$. Finally, a rank vector scoring rule that will turn out to be particularly useful for the incentive ratio is the \emph{Minimal Compromise rule} $f_{\min}$, which is the rule induced by the rank vector scoring function $s_{\min}(r^*(R,x))=m-\max\{k\colon r^*_k(R,x)>0\}$. Put differently, the Minimal Compromise rule assigns to each alternative a score capturing its worst rank in a voter's input ranking: $s_{\min}(r^*(R,x))=0$ if there is a voter who ranks $x$ last, $s_{\min}(r^*(R,x))=1$ if no voter ranks $x$ last but one ranks it second to last, and so on. We note that this rule has been considered before by \citet{CoMe12a} in an axiomatic context.

We next show that the Minimal Compromise rule attains an incentive ratio of $m-2$ when $m\geq 4$. While this improves the incentive ratio by a factor of roughly $m/2$ in comparison to distance scoring rules, we note that an incentive ratio of $m-2$ still allows for rather severe manipulation if there are many alternatives. However, we also show that, within the class of rank vector scoring rules, no unanimous SWF can have an incentive ratio lower than $m-2$. This means that, at least with the simple approach of assigning scores to the alternatives and then sorting them according to these scores, no unanimous SWF can be less manipulable than the Minimal Compromise rule in terms of the incentive ratio.



\begin{proposition}\label{prop:MC}
The following claims are true.
\begin{enumerate}[label=(\arabic*), leftmargin=*]
    \item For all $m\geq 4$, the incentive ratio of the Minimal Compromise rule $f_{\min}$ is $\gamma_m(f_{\min})=m-2$.
    \item For all $m\geq 4$, the incentive ratio of every unanimous rank vector scoring rule $f_{\mathit{RV}}$ satisfies $\gamma_m(f_{\mathit{RV}})\geq m-2$.
\end{enumerate}
\end{proposition}
\begin{proof}
We prove both claims of this proposition separately. Moreover, for the result on the Minimal Compromise rule, it suffices to show the upper bound, as a matching lower bound follows from Claim~(2).


    \paragraph{Claim (1): $\gamma_m(f_{\min})\leq m-2$.} We first show that the incentive ratio of the Minimal Compromise rule $f_{\min}$ is at most $m-2$ if $m\geq 4$. For this, we proceed in multiple steps. First, we prove that we can restrict our analysis to deviations such that the scores $s_{\min}(r^*(R,x))$ weakly decrease for all alternatives, because the maximal utility gain of a voter is attained by such a deviation. Second, we prove that $\sum_{x\in A}s_{\min}(r^*(R,x))$ is a lower bound on the Kemeny utility of every voter. Based on these insights, we prove this claim in the final step.\medskip

\emph{Step 1:} As a first step, we show that it suffices to focus on deviations that weakly reduce the scores $s_{\min}(r^*(R,x))$ of all alternatives. To prove this claim, let $R$ and $R'$ denote two profiles and $i$ a voter such that $R$ and $R'$ differ only in the ranking of voter $i$. We next show that there are two other profiles $\widehat R$ and $\widehat R'$ such that \emph{(i)} ${\widehat\succ_i}={\succ_i}$ and ${\widehat\succ_i'}={\succ_i'}$, \emph{(ii)} $\widehat R$ differs from $\widehat R'$ only in the ranking of voter~$i$, \emph{(iii)} $s_{\min}(r^*(\widehat R',x))\leq s_{\min}(r^*(\widehat R,x))$ for all $x\in A$, and \emph{(iv)} $f_{\min}(R)=f_{\min}(\widehat R)$ and $u({\succ_i},f_{\min}(\widehat R'))\geq u({\succ_i},f_{\min}(R'))$. This means that the incentive ratio of $f_{\min}$ is maximized when voters only decrease the scores $s_{\min}(r^*(R,x))$. Now, to prove this claim, we define $\widehat R$ and $\widehat R'$ as the profiles derived from $R$ and $R'$ by adding a new voter who reports $\succ_i$. Clearly, the resulting profiles $\widehat R$ and $\widehat R'$ satisfy conditions \emph{(i)} and \emph{(ii)} by construction. Further, it holds for all $x\in A$ that
\begin{align*}
s_{\min}(r^*(\widehat R',x))
&=\min\left(m-r(\succ_i,x),m-r({\succ_i'},x),\min_{j\in N\setminus\{i\}}m-r({\succ_j},x)\right)\\
&\leq \min\left(m-r({\succ_i},x),\min_{j\in N\setminus\{i\}}m-r({\succ_j},x)\right)\\
&=s_{\min}(r^*(\widehat R,x)).
\end{align*}
This proves condition \emph{(iii)}. Moreover, since cloning rankings does not affect the Minimal Compromise scores, we have $s_{\min}(r^*(R,x))=s_{\min}(r^*(\widehat R,x))$ for all $x\in A$ and thus $f_{\min}(R)=f_{\min}(\widehat R)$.

As a final point, we need to show that $u(\succ_i,f_{\min}(\widehat R'))\geq u(\succ_i,f_{\min}(R'))$. To ease notation, we set ${\widehat\rhd}=f_{\min}(\widehat R')$ and ${\rhd}=f_{\min}(R')$ for the rest of this step. Furthermore, for every alternative $x\in A$, we define the set $S(x)=\{y\in A\colon x\rhd y\land y\mathrel{\widehat\rhd}x\}$ and note that $\Delta(\rhd,\widehat\rhd)=\sum_{x\in A}|S(x)|$. Now, we fix an alternative $x\in A$ and analyze the set $S(x)$. To this end, we first observe that $s_{\min}(r^*(\widehat R',y))\leq s_{\min}(r^*(R',y))$ for all $y\in A$ because $\widehat R'$ arises from $R'$ by adding a new voter. If $s_{\min}(r^*(\widehat R',x))=s_{\min}(r^*(R',x))$, we claim that $S(x)=\emptyset$. In more detail, our assumption means that $s_{\min}(r^*(\widehat R',x))=s_{\min}(r^*(R',x))\geq s_{\min}(r^*(R',y))\geq s_{\min}(r^*(\widehat R',y))$ for all $y$ with $x\rhd y$. The second inequality here follows because $x\rhd y$ necessitates that $s_{\min}(r^*(R',x))\geq s_{\min}(r^*(R',y))$. If $s_{\min}(r^*(\widehat R',x))>s_{\min}(r^*(\widehat R',y))$, it immediately follows that $x\mathrel{\widehat\rhd}y$. On the other hand, if $s_{\min}(r^*(\widehat R',x))=s_{\min}(r^*(\widehat R',y))$, then also $s_{\min}(r^*(R',x))=s_{\min}(r^*(R',y))$. Hence, $x\rhd y$ means that $x$ is lexicographically preferred to $y$, which again implies that $x\mathrel{\widehat\rhd}y$. In summary, if $s_{\min}(r^*(\widehat R',x))=s_{\min}(r^*(R',x))$, then, for every $y$ with $x\rhd y$, we have $x\mathrel{\widehat\rhd}y$ and thus $S(x)=\emptyset$.

Next, suppose that $s_{\min}(r^*(\widehat R',x))<s_{\min}(r^*(R',x))$, which means that $s_{\min}(r^*(\widehat R',x))=m-r(\succ_i,x)$. By definition, it holds for all $y\in S(x)$ that $y\mathrel{\widehat\rhd}x$ and thus $s_{\min}(r^*(\widehat R',y))\geq s_{\min}(r^*(\widehat R',x))$. This implies that $y\succ_i x$ for all $y\in S(x)$. Indeed, if $x\succ_i y$, then $s_{\min}(r^*(\widehat R',y))\leq m-r(\succ_i,y)<m-r(\succ_i,x)=s_{\min}(r^*(\widehat R',x))$. By combining our two cases, we infer for all alternatives $x,y\in A$ with $x\rhd y$ and $y\mathrel{\widehat\rhd}x$ that $y\succ_i x$. This means that $\Delta(\succ_i,\widehat\rhd)=\Delta(\succ_i,\rhd)-\Delta(\rhd,\widehat\rhd)\leq\Delta(\succ_i,\rhd)$ and thus $u(\succ_i,\widehat\rhd)\geq u(\succ_i,\rhd)$.\medskip

\emph{Step 2:} Next, we show that $\sum_{x\in A}s_{\min}(r^*(R,x))$ is a lower bound on the utility of every voter in the considered profile, i.e., $u(\succ_i,f_{\min}(R))\geq\sum_{x\in A}s_{\min}(r^*(R,x))$ for all voters $i$ in $R$. To see this, fix a profile $R$ and a voter $i$ and let ${\rhd}=f_{\min}(R)$. We further define the utility of an alternative $x\in A$ by $u(x,\succ_i,\rhd)=|\{y\in A\colon x\succ_i y\land x\rhd y\}|$ and note that $u(\succ_i,\rhd)=\sum_{x\in A}u(x,\succ_i,\rhd)$. Finally, fix an alternative $x\in A$; we show that $u(x,\succ_i,\rhd)\geq s_{\min}(r^*(R,x))$. If $s_{\min}(r^*(R,x))=0$, this holds trivially, so suppose that $s_{\min}(r^*(R,x))>0$. In this case, let $Y=\{y\in A\colon m-r(\succ_i,y)<s_{\min}(r^*(R,x))\}$ and note that $|Y|=s_{\min}(r^*(R,x))$. Further, it holds for all $y\in Y$ that $x\succ_i y$ and $x\rhd y$ because $s_{\min}(r^*(R,y))\leq m-r(\succ_i,y)<s_{\min}(r^*(R,x))\leq m-r(\succ_i,x)$. This entails that $u(x,\succ_i,\rhd)\geq|Y|=s_{\min}(r^*(R,x))$, thus proving this step.\medskip

\emph{Step 3:} Lastly, we prove our upper bound. To this end, let $R$ and $R'$ denote two profiles on $m$ alternatives that differ only in the ranking of a single voter $i$. Further, let ${\rhd}=f_{\min}(R)$ and ${\rhd'}=f_{\min}(R')$. By Step~1, we may assume that $s_{\min}(r^*(R',x))\leq s_{\min}(r^*(R,x))$ for all $x\in A$. In turn, this implies that $\frac{u(\succ_i,\rhd')}{u(\succ_i,\rhd)}\neq\infty$. In particular, if $u(\succ_i,\rhd)=0$, then $\rhd$ has to top-rank the bottom-ranked alternative $x^*$ of $\succ_i$. However, it holds that $s_{\min}(r^*(R,x^*))=m-r(\succ_i,x^*)=0$ and thus also that $s_{\min}(r^*(R,x))=0$ for all $x\in A$ because $f_{\min}$ sorts the alternatives in decreasing order of their scores. Finally, because $s_{\min}(r^*(R',x))\leq s_{\min}(r^*(R,x))$ for all $x\in A$, we conclude that $s_{\min}(r^*(R',x))=0=s_{\min}(r^*(R,x))$ for all $x\in A$. This implies that ${\rhd}={\rhd'}$, so a voter with utility $0$ cannot manipulate.

More generally, our previous argument shows that, if $s_{\min}(r^*(R,x))=0$ for all $x\in A$, no voter can successfully manipulate. Hence, suppose that $s_{\min}(r^*(R,x))>0$ for some alternative $x$. Similar to Step~1, we define the sets
\begin{align*}
    S^+(x)&=\{y\in A\setminus\{x\}\colon x\rhd y\land y\rhd' x\land y\succ_i x\}\qquad\text{and}\\
    S^-(x)&=\{y\in A\setminus\{x\}\colon x\rhd y\land y\rhd' x\land x\succ_i y\}.
\end{align*}
We note for these sets that $\{(x,y)\in A^2\colon x\rhd y\land y\rhd' x\}=\{(x,y)\in A^2\colon y\in S^+(x)\}\cup\{(x,y)\in A^2\colon y\in S^-(x)\}$, so $\Delta(\rhd,\rhd')=\sum_{x\in A}(|S^+(x)|+|S^-(x)|)$. Further, the set $\{(x,y)\in A^2\colon y\in S^+(x)\}$ contains all pairs of alternatives on which $\rhd$ and $\rhd'$ disagree and that move $\rhd$ closer to $\succ_i$, whereas $\{(x,y)\in A^2\colon y\in S^-(x)\}$ contains the pairs of alternatives on which $\rhd$ and $\rhd'$ disagree and that move $\rhd$ further away from $\succ_i$. In particular, this means that $\Delta(\succ_i,\rhd')=\Delta(\succ_i,\rhd)+\sum_{x\in A}(|S^-(x)|-|S^+(x)|)$ and thus $u(\succ_i,\rhd')=u(\succ_i,\rhd)+\sum_{x\in A}(|S^+(x)|-|S^-(x)|)$.

We next aim to bound the value $|S^+(x)|-|S^-(x)|$ for every alternative $x$. To this end, we first note that if $s_{\min}(r^*(R,x))=s_{\min}(r^*(R',x))$ for some alternative $x$, then $S^+(x)=S^-(x)=\emptyset$. The reason for this is that if $s_{\min}(r^*(R,x))=s_{\min}(r^*(R',x))$, then $s_{\min}(r^*(R',x))=s_{\min}(r^*(R,x))\geq s_{\min}(r^*(R,y))\geq s_{\min}(r^*(R',y))$ for all $y$ with $x\rhd y$. This implies that $x\rhd' y$ for all such $y$, so there are no alternatives $x,y$ such that $x\rhd y$ and $y\rhd' x$. In particular, this means that $|S^+(x)|-|S^-(x)|=0$ for all $x\in A$ with $s_{\min}(r^*(R,x))=0$.

Next, we turn to alternatives $x$ with $s_{\min}(r^*(R,x))>0$. In this case, let $x^*$ denote the bottom-ranked alternative in $\succ_i$ and note that $x\rhd x^*$ as $s_{\min}(r^*(R,x))>0=s_{\min}(r^*(R,x^*))$. If $x^*\rhd' x$, it holds that $|S^+(x)|\leq m-2$ and $|S^-(x)|\geq 1$ and thus $|S^+(x)|-|S^-(x)|\leq m-3$. Indeed, $x\notin S^+(x)$ by definition of this set, $x^*\notin S^+(x)$ because this would require $x^*\succ_i x$, and $x^*\in S^-(x)$ since $x\rhd x^*$, $x^*\rhd' x$, and $x\succ_i x^*$. Similarly, if $x\rhd' x^*$ and there is another alternative $y$ with $x^*\rhd' y$, it holds that $|S^+(x)|-|S^-(x)|\leq m-3$. Indeed, we again have that $x\notin S^+(x)$ by definition, and $x^*,y\notin S^+(x)$ since $x\rhd' x^*\rhd' y$, whereas membership in $S^+(x)$ requires that $\rhd'$ ranks the considered alternative ahead of $x$. Hence, when letting $\ell$ denote the number of alternatives with $s_{\min}(r^*(R,x))>0$, we derive that
\begin{align*}
\frac{u(\succ_i,\rhd')}{u(\succ_i,\rhd)}
&=\frac{u(\succ_i,\rhd)+\sum_{x\in A}(|S^+(x)|-|S^-(x)|)}{u(\succ_i,\rhd)}\\
&\leq \frac{u(\succ_i,\rhd)+\ell(m-3)}{u(\succ_i,\rhd)}\\
&\leq \frac{\ell+\ell(m-3)}{\ell}\\
&=m-2.
\end{align*}
Here, the first equality uses the definition of $S^+(x)$ and $S^-(x)$, and the first inequality uses that $|S^+(x)|-|S^-(x)|=0$ if $s_{\min}(r^*(R,x))=0$ and $|S^+(x)|-|S^-(x)|\leq m-3$ otherwise. The second inequality holds because $\ell\leq\sum_{x\in A}s_{\min}(r^*(R,x))\leq u(\succ_i,\rhd)$, where the last inequality follows from Step~2. Hence, if $\rhd'$ does not bottom-rank $x^*$, our upper bound on the incentive ratio holds.

As the second case, assume that $x^*$ is bottom-ranked in $\rhd'$. Since $s_{\min}(r^*(R',x))\leq s_{\min}(r^*(R,x))$ for all $x\in A$, it holds that $\{x\in A\colon s_{\min}(r^*(R,x))=0\}\subseteq\{x\in A\colon s_{\min}(r^*(R',x))=0\}$. By our lexicographic tie-breaking, this means that $x^*$ is also bottom-ranked in $\rhd$, and so $S^+(x^*)=S^-(x^*)=\emptyset$. Further, let $y_i$ denote the $i$-th best alternative in $\rhd'$, i.e., $r(\rhd',y_i)=i$. By definition of the sets $S^+(x)$ and $S^-(x)$, it holds for alternatives $y_i$ with $i\in\{1,\dots,m-1\}$ that $|S^+(y_i)|-|S^-(y_i)|\leq|S^+(y_i)|\leq i-1$ because there are only $i-1$ alternatives $x$ with $x\rhd' y_i$. Lastly, since $x^*$ is bottom-ranked in both $\succ_i$ and $\rhd$, we conclude that $u(\succ_i,\rhd)\geq m-1$ and thus
\begin{align*}
    \frac{u(\succ_i,\rhd')}{u(\succ_i,\rhd)}
    &=\frac{u(\succ_i,\rhd)+\sum_{x\in A}(|S^+(x)|-|S^-(x)|)}{u(\succ_i,\rhd)}\\
    &\leq\frac{u(\succ_i,\rhd)+\sum_{i=1}^{m-1}(i-1)}{u(\succ_i,\rhd)}\\
    &\leq\frac{(m-1)+(m-2)(m-1)/2}{m-1}\\
    &=\frac{m}{2}\\
    &\leq m-2.
\end{align*}
This chain of (in)equalities follows analogously to the previous case, except for the last step, where we use that $m\geq 4$ implies $\frac{m}{2}\leq m-2$.\medskip

    \paragraph{Claim (2): $\gamma_m(f)\geq m-2$ for all unanimous rank vector scoring rules $f$.} To prove this claim, we fix a rank vector scoring rule $f$ for $m\geq 4$ alternatives and its rank vector scoring function $s$. Further, we assume for contradiction that $f$ has an incentive ratio of strictly less than $m-2$. We aim to show that this means that $f$ fails unanimity. To this end, we focus on the case that $m=n$ throughout the proof. Further, we label our alternatives $A=\{x_1,\dots,x_m\}$ and assume that the tie-breaking order is given by $x_1\mathbf{>}x_2\mathbf{>}\dots \mathbf{>}x_m$. Now to derive a contradiction, we proceed in four steps. First, we show that, whenever all alternatives have the same score at a profile, the score of an alternative is invariant under arbitrary changes to its rank vector that leave the numbers of first- and last-place appearances unchanged. Based on this insight, we show that $s(1,m-1,0,\dots,0)=s(1,\dots, 1)=s(0,\dots,0,m-1,1)$. Finally, these equalities allow us to construct a profile at which $f$ violates unanimity. 
\medskip

    \emph{Step 1:} First, we establish an invariance condition on the rank vector scoring function $s$. To this end, let $R$ be a profile such that $s(r^*(R,x))=s(r^*(R,y))$ for all alternatives $x,y\in A$, and let $v^0$ be a rank vector that occurs in $R$, i.e., $r^*(R,x^*)=v^0$ for some $x^*\in A$. We will show that $s(v)=s(v^0)$ for every rank vector $v$ with $v_1=v^0_1$, $v_m=v^0_m$, and $\sum_{j=1}^m v_j=\sum_{j=1}^m v_j^0=n$. Moreover, for every such $v$, there exists a profile $\hat R$ in which one alternative has rank vector $v$ and all alternatives receive the same score. To prove this claim, fix such a vector $v$ and let $d=\frac12\sum_{j=1}^m |v_j-v^0_j|$. We will next transform $v^0$ into $v$ by a sequence of small changes. Specifically, we construct vectors $v^0,v^1,\dots,v^d=v$ such that, for every $k\in\{1,\dots,d\}$, there are indices $a_k,b_k\in\{2,\dots,m-1\}$ satisfying \emph{(i)} $v^{k-1}_{a_k}<v_{a_k}$ and $v^k_{a_k}=v^{k-1}_{a_k}+1$, \emph{(ii)} $v^{k-1}_{b_k}>v_{b_k}$ and $v^k_{b_k}=v^{k-1}_{b_k}-1$, and \emph{(iii)} $v^k_j=v^{k-1}_j$ for all $j\notin\{a_k,b_k\}$. Such indices exist whenever $v^{k-1}\neq v$, because $\sum_{j=1}^m v_j=\sum_{j=1}^m v_j^{k-1}=n$, $v_1=v^{k-1}_1$ and $v_m=v^{k-1}_m$. Furthermore, each such step decreases the $\ell_1$-distance to $v$ by $2$, so this process indeed terminates at $v$ after $d$ steps.

We now prove inductively that, for every $k\in\{0,\dots,d\}$, there exists a profile $\hat R^k$ and an alternative $x^k$ such that $r^*(\hat R^k,x^k)=v^k$
and $s(r^*(\hat R^k,x))=s(v^0)$ for all alternatives $x\in A$. For $k=0$, we can simply take $\hat R^0=R$ and $x^0=x^*$. Now fix $k\in\{1,\dots,d\}$ and suppose that such a profile $\hat R^{k-1}$ exists. Let $a=a_k$ and $b=b_k$, and choose a voter $i$ in $\hat R^{k-1}$ who ranks $x^{k-1}$ in position $b$. Such a voter exists because $v^{k-1}_b>v_b\geq 0$. Furthermore, we denote voter $i$'s ranking by $y_m y_{m-1}\cdots y_1$. We next rename the alternatives in $\hat R^{k-1}$ so that $y_t$ is renamed as $x_t$ for every $t\in\{1,\dots,m\}$, and denote the resulting profile by $\bar R^{k-1}$. Hence, voter $i$ reports $x_m x_{m-1}\cdots x_1$ in $\bar R^{k-1}$. Renaming the alternatives only permutes the rank vectors among the alternatives. Consequently, all alternatives still receive the same score in $\bar R^{k-1}$. By our tie-breaking assumption, this implies that $f(\bar R^{k-1})=x_1\cdots x_m$. Thus, the outcome is precisely inverse to voter $i$'s ranking, so voter $i$ obtains utility $0$ in $\bar R^{k-1}$.

The alternative corresponding to $x^{k-1}$ is ranked by voter $i$ in position $b$ and therefore has rank vector $v^{k-1}$. Let $z$ denote this alternative, and let $z'$ denote the alternative that voter $i$ ranks in position $a$. We derive a new profile $\hat R^k$ from $\bar R^{k-1}$ by letting voter $i$ exchange $z$ and $z'$ in her reported ranking. Since $a,b\in\{2,\dots,m-1\}$, this deviation does not change the positions of $x_m$ and $x_1$ in voter $i$'s ranking. By construction, the rank vector of $z$ changes from $v^{k-1}$ to $v^k$: one fewer voter ranks $z$ in position $b$, and one additional voter ranks it in position $a$. Hence, $r^*(\hat R^k,z)=v^k$. 

Moreover, voter $i$ had utility $0$ at $\bar R^{k-1}$. Since $f$ has a finite incentive ratio, no deviation of voter $i$ can result in positive utility. The only ranking that gives voter $i$ utility $0$ is the inverse of her true ranking, so we obtain that $f(\hat R^k)=x_1\cdots x_m$. Furthermore, we observe that neither $x_1$ nor $x_m$ changes position when moving from $\bar R^{k-1}$ to $\hat R^k$. Hence, by the induction hypothesis,
\[
s(r^*(\hat R^k,x_1))
=
s(r^*(\bar R^{k-1},x_1))
=
s(r^*(\bar R^{k-1},x_m))
=
s(r^*(\hat R^k,x_m)).
\]
Since $f(\hat R^k)=x_1\cdots x_m$, every alternative $x$ must have a score between those of $x_1$ and $x_m$. This implies that $s(r^*(\hat R^k,x))=s(r^*(\hat R^k,x_1))=s(r^*(\bar R^{k-1},x_1))$ for every $x\in A$. Thus, all alternatives again receive the same score in $\hat R^k$. Moreover, this also implies that 
this implies that 
\[s(v^k)=s(r^*(\hat R^k,z))=s(r^*(\hat R^k, x_1)=s(r^*(\bar R^{k-1},z))=s(v^{k-1}).
\]
This completes the induction and thus proves that $s(v)=s(v^0)$, as desired.\medskip

\emph{Step 2:} We will next show that $s(1,m-1,0,\dots,0)=s(1,\dots,1)$. By Step 1, it is enough to find a profile $R^*$ such that $s(r^*(R^*,x))=s(r^*(R^*,y))$ for all $x,y\in A$ and there is an alternative $x^*$ such that $r^*_1(R^*,x^*)=1$ and $r^*_m(R^*,x^*)=0$. Indeed, we can then apply Step 1 to the rank vector $r^*(R^*,x^*)$ and the vector $(1,m-1,0,\dots,0)$.

As a first step, consider the profile $R$ where, for every $k\in\{1,\dots,m\}$, there is one voter reporting the ranking
$x_{k-1}\dots x_1x_m\dots x_k$.
In this profile, every alternative appears at every rank exactly once, so $r^*(R,x)=(1,\dots,1)$ for every $x\in A$. Consequently, all alternatives have the same score, and therefore $f(R)=x_1\dots x_m$ by our tie-breaking order. Next, we observe that there is a voter $i$ in $R$ who reports $x_m\dots x_1$ and thus obtains utility $0$ from the outcome $x_1\dots x_m$. Since $f$ has a bounded incentive ratio, no deviation of this voter can result in positive utility. As $x_1\dots x_m$ is the unique ranking that gives this voter utility $0$, she cannot affect the outcome by deviating. Now, let $R^1$ denote the profile derived from $R$ by letting voter $i$ deviate to $x_m\dots x_3x_1x_2$. It follows that $f(R^1)=x_1\dots x_m$ and thus $s(r^*(R^1,x_1))\geq s(r^*(R^1,x_2))\geq s(r^*(R^1,x_3))$. Moreover, $r^*(R^1,x_1)=(1,\dots,1,2,0)$, $r^*(R^1,x_2)=(1,\dots,1,0,2)$, and $r^*(R^1,x_3)=(1,\dots,1)$. Hence,
\[
s(1,\dots,1,2,0)\geq s(1,\dots,1,0,2)\geq s(1,\dots, 1).
\]

For the next step, consider the voter $i$ in $R$ who reports $x_{m-1}\dots x_1x_m$, and let $R^2$ denote the profile in which this voter instead reports ${\succ_i^2}=x_{m-1}\dots x_2x_mx_1$. We have $r^*(R^2,x_m)=(1,\dots,1,2,0)$, $r^*(R^2,x_1)=(1,\dots,1,0,2)$, and $r^*(R^2,x)=(1,\dots,1)$ for all $x\in A\setminus\{x_1,x_m\}$. Suppose that $s(1,\dots,1,2,0)>s(1,\dots,1,0,2)$. By our analysis so far, this implies that $f(R^2)=x_mx_1\dots x_{m-1}$. If $\succ_i^2$ is the true ranking of voter $i$, then  $u(\succ_i^2,f(R^2))=1$. On the other hand, if voter $i$ reports $x_{m-1}\dots x_1x_m$, we return to the profile $R$, whose outcome is $x_1\dots x_m$. This ranking agrees with $\succ_i^2$ on the $m-2$ comparisons $x_j\succ_i^2 x_m$ for $j\in\{2,\dots,m-1\}$. Hence, $u(\succ_i^2,f(R))=m-2$. The incentive ratio of $f$ would therefore be at least $m-2$, contradicting our assumption. Hence, we conclude that
\[
s(1,\dots,1,2,0)=s(1,\dots,1,0,2).
\]

Finally, consider the profiles $R^3$ and $R^4$ defined as follows. In $R^3$, the rankings $x_3x_m\dots x_4x_1x_2$, $x_2x_3x_m\dots x_4x_1$, and $x_1x_2x_3x_m\dots x_4$ are reported by one voter each and, for every $k\in\{4,\dots,m\}$, there is one voter who reports $x_k\dots x_4x_1x_2x_3x_m\dots x_{k+1}$. Again, every alternative appears at every rank exactly once, so $r^*(R^3,x)=(1,\dots,1)$ for all $x\in A$ and $f(R^3)=x_1\dots x_m$. Furthermore, $R^4$ arises from $R^3$ by letting the voter who reports $x_m\dots x_4x_1x_2x_3$ deviate to ${\succ}^4=x_m\dots x_4x_1x_3x_2$. Hence, $r^*(R^4,x_3)=(1,\dots,1,2,0)$, $r^*(R^4,x_2)=(1,\dots,1,0,2)$, and $r^*(R^4,x)=(1,\dots,1)$ for all $x\in A\setminus\{x_2,x_3\}$.
Suppose now that $s(1,\dots,1,2,0)=s(1,\dots,1,0,2)>s(1,\dots, 1)$. This implies that $f(R^4)=x_2x_3x_1x_4\dots x_m$ by our tie-breaking order. This ranking is precisely inverse to $\succ^4$, so $u(\succ^4,f(R^4))=0$. On the other hand, if the voter deviates to $R^3$, the outcome is $x_1\dots x_m$, which gives her strictly positive utility. Thus, $f$ would have an unbounded incentive ratio, a contradiction. Together with our previous two inequalities, we therefore obtain
\[
s(1,\dots,1,2,0)
=
s(1,\dots,1,0,2)
=
s(1,\dots,1)
\]

Finally, this means that we can, e.g., choose $R^*=R^1$ and $x^*=x_1$ to prove the claim of this step with the insights of Step 1.\medskip

\emph{Step 3:} We will next show that $s(0,\dots,0,m-1,1)=s(1,\dots,1)$. For this, we will essentially mirror the argument of Step 2. In more detail, by Step 1, it is enough to find a profile $R^*$ such that $s(r^*(R^*,x))=s(r^*(R^*,y))$ for all $x,y\in A$ and there is an alternative $x^*$ such that $r^*_1(R^*,x^*)=0$ and $r^*_m(R^*,x^*)=1$. Indeed, we can then apply Step 1 to the rank vector $r^*(R^*,x^*)$ and the vector $(0,\dots,0,m-1,1)$.

To find the profile $R^*$, we return to the cyclic profile $R$ from Step 2, where $r^*(R,x)=(1,\dots,1)$ for every alternative $x\in A$ and thus $f(R)=x_1\dots x_m$. Moreover, there is a voter $i$ who reports $x_m\dots x_1$ and thus obtains utility $0$ from the outcome. Let $R^1$ denote the profile derived from $R$ by letting this voter deviate to $x_{m-1}x_mx_{m-2}\dots x_1$. Since $f$ has a bounded incentive ratio, this deviation cannot change the outcome, and thus $f(R^1)=x_1\dots x_m$. We further have $r^*(R^1,x_m)=(0,2,1,\dots,1)$, $r^*(R^1,x_{m-1})=(2,0,1,\dots,1)$, and $r^*(R^1,x_{m-2})=(1,\dots,1)$. Hence, $s(1,\dots, 1)\geq s(2,0,1,\dots,1)\geq s(0,2,1,\dots,1)$.

Next, consider the voter $i$ in $R$ who reports $x_1x_mx_{m-1}\dots x_2$, and let $R^2$ denote the profile in which this voter instead reports ${\succ_i^2}=x_mx_1x_{m-1}\dots x_2$. We have $r^*(R^2,x_m)=(2,0,1,\dots,1)$, $r^*(R^2,x_1)=(0,2,1,\dots,1)$, and $r^*(R^2,x)=(1,\dots,1)$ for all $x\in A\setminus\{x_1,x_m\}$. Suppose that $s(2,0,1,\dots,1)>s(0,2,1,\dots,1)$. Since $s(1,\dots, 1)\geq s(2,0,1,\dots,1)$, it follows that $f(R^2)=x_2\dots x_mx_1$. This implies that $u(\succ_i^2,f(R^2))=1$. On the other hand, in $R$, it holds that $u(\succ_i^2,f(R))=m-2$. Hence, if $s(2,0,1,\dots,1)>s(0,2,1,\dots,1)$, the incentive ratio of $f$ would be at least $m-2$, contradicting our assumption. Thus, $s(2,0,1,\dots,1)=s(0,2,1,\dots,1)$.

It remains to show that this score is equal to $s(1,\dots, 1)$. For this, consider a cyclic profile $R^3$ that contains the ranking $x_2x_3x_m\dots x_4x_1$, i.e., $R^3$ consists of this ranking and its $m-1$ cyclic shifts. Every alternative appears at every rank exactly once, so $r^*(R^3,x)=(1,\dots,1)$ for all $x\in A$ and $f(R^3)=x_1\dots x_m$. Let $R^4$ arise from $R^3$ by letting the voter who reports $x_2x_3x_m\dots x_4x_1$ deviate to $\succ=x_3x_2x_m\dots x_4x_1$. Hence, $r^*(R^4,x_2)=(0,2,1,\dots,1)$, $r^*(R^4,x_3)=(2,0,1,\dots,1)$, while $r^*(R^4,x)=(1,\dots,1)$ for all $x\in A\setminus\{x_2,x_3\}$.
Suppose now that $s(2,0,1,\dots,1)=s(0,2,1,\dots,1)<s(1,\dots, 1)$. This, implies that $f(R^4)=x_1x_4\dots x_mx_2x_3$. This ranking is precisely inverse to $\succ$, so $u(\succ,f(R^4))=0$. On the other hand, if the voter deviates back to $R^3$, $x_1\dots x_m$ is chosen and she obtains positive utility. Thus, $f$ would have an unbounded incentive ratio, a contradiction. Consequently, $s(2,0,1,\dots,1)=s(0,2,1,\dots,1)=s(1,\dots,1)$.
Finally, we can now choose $R^*=R^1$ and $x^*=x_{m}$ and apply the insights of Step 1 to finish the proof of this step.\medskip

\emph{Step 4:} Finally, we will show that $f$ fails unanimity. To this end, consider the profile $R^*$ in which one voter reports
$x_m\dots x_1$
and every other voter reports
$x_{m-1}x_mx_{m-2}\dots x_3x_1x_2$.
We first note that the rank vector of $x_m$ is $r^*(R^*,x_m)=(1,m-1,0,\dots,0)$, since $x_m$ is ranked first by one voter and second by all remaining $m-1$ voters. Moreover, the rank vector of $x_1$ is $r^*(R^*,x_1)=(0,\dots,0,m-1,1)$, since $x_1$ is ranked last by one voter and second to last by all remaining voters.
By Steps 2 and 3, it therefore holds that $s(r^*(R^*,x_m))
=
s(1,m-1,0,\dots,0)
=
s(1,\dots,1)
=
s(0,\dots,0,m-1,1)
=
s(r^*(R^*,x_1))$.
That is, $x_m$ and $x_1$ receive the same score. Since ties are broken according to the fixed tie-breaking order $x_1>\dots>x_m$, it follows that $x_1\rhd x_m$ for $\rhd=f(R^*)$. However, every voter in $R^*$ ranks $x_m$ ahead of $x_1$. Unanimity therefore requires that $x_m\rhd x_1$, which contradicts the conclusion above. Thus, we finally deduce that no unanimous rank vector scoring rules attains an incentive ratio of less than $m-2$. 
\end{proof}


\begin{remark}
When the number of voters $n$ is substantially smaller than the number of alternatives $m$, the incentive ratio of the Minimal Compromise rule improves. The reason for this is that, if $n$ is much smaller than $m$, every voter obtains substantial utility from the output ranking of the Minimal Compromise rule. To make this more formal, fix a voter $i$ with ranking ${\succ_i}=x_1\dots x_m$ and let $\rhd$ be the output ranking of the Minimal Compromise rule for a profile $R$ on $n$ voters. There are at most $n$ alternatives that obtain rank $m$ in $R$, so there are at least $m-n$ alternatives that both $\succ_i$ and $\rhd$ rank ahead of $x_m$. Similarly, there are at most $2n$ alternatives that obtain rank at least $m-1$ in $R$, so there are at least $m-2n$ alternatives that both $\succ_i$ and $\rhd$ rank ahead of $x_{m-1}$. More generally, for each rank $r$ with $n(m+1-r)<m$, there are at least $m-n(m+1-r)$ alternatives that both $\succ_i$ and $\rhd$ rank ahead of $x_r$. Hence, a lower bound on the utility of voter $i$ for $\rhd$ is $\sum_{j=1}^{\lfloor m/n\rfloor}(m-j\cdot n)
=m\cdot\lfloor m/n\rfloor-\frac{n}{2}\cdot\lfloor m/n\rfloor\cdot(\lfloor m/n\rfloor+1)
\geq \frac{(m-n)^2}{2n}$.
By contrast, the maximal utility of a voter is $\frac{m(m-1)}{2}$, so this yields an upper bound of $
\frac{nm(m-1)}{(m-n)^2}$
on the incentive ratio, which is approximately $n$ when $n$ is small relative to $m$.

We moreover note that this argument shows that the Minimal Compromise rule is approximately proportional, as every voter obtains a utility of roughly ${m\choose 2}/n$. Such proportional SWFs were recently studied by \citet{Lede25b}. In more detail, this author suggests the Proportional Sequential Borda rule, which guarantees every voter a utility of at least $\lfloor {m\choose 2}/n\rfloor$. Again, if $n<m-2$, this yields better bounds on the incentive ratio than those discussed in our propositions. Furthermore, for many applications of rank aggregation, we are precisely in the situation where a few voters have to rank many alternatives. This suggests that such proportional SWFs may also be appealing from a strategic point of view.
\end{remark}

\section{Conclusion}\label{sec:conclusion}

In this paper, we study social welfare functions (SWFs) with respect to (Kemeny-)strategyproofness, which requires that voters cannot obtain a ranking that is closer to their true ranking in terms of the Kemeny distance by voting strategically. As our main result, we show a sweeping impossibility theorem, demonstrating that no anonymous and unanimous SWF satisfies strategyproofness if there are $m=4$ alternatives and the number of voters $n$ is a multiple of $4$, or if there are $m\geq 5$ alternatives and $n$ is even. This result is obtained via SAT solving, a computer-aided theorem-proving technique, and verified via Isabelle, a highly trustworthy interactive theorem prover. Moreover, we prove by hand that no majority-consistent SWF is strategyproof when $m\geq 4$. Lastly, we analyze the incentive ratio for several natural classes of SWFs, which quantifies the maximal multiplicative utility gain of a voter when manipulating. Unfortunately, it turns out that many SWFs have a very high incentive ratio. Specifically, the Kemeny rule and generalizations thereof have an incentive ratio of roughly ${m\choose 2}$, whereas the incentive ratio of positional scoring rules is even unbounded. We moreover identify an SWF called the Minimal Compromise rule, which attains an improved incentive ratio of $m-2$, and show that this is optimal within the broad class of rank vector scoring rules.

On a more intuitive level, our results have a clear interpretation. First, our impossibility theorems show that every SWF that one would consider in practice can be strategically manipulated. These results can therefore be seen as a counterpart to the Gibbard-Satterthwaite theorem \citep{Gibb73a,Satt75a} for rank aggregation. Second, our results on the incentive ratio demonstrate that, for natural classes of SWFs, the considered voting rules can, in the worst case, be severely manipulated. For instance, our results on the incentive ratio of the Kemeny rule show that a voter who agrees with only a single pairwise comparison of the output ranking can obtain almost her ideal output ranking by manipulating.

Perhaps the most natural follow-up question in light of our negative results is how to circumvent them. There are a multitude of different ways to approach this question. First, a natural avenue would be to consider strategyproofness notions based on different assumptions on how voters compare rankings. Indeed, once we move away from the Kemeny distance, it may be possible to obtain more positive results. Second, similar to single-winner voting, it might be worthwhile to consider restricted domains of input rankings or randomized mechanisms. For example, when assuming that the input rankings belong to a single-crossing domain (i.e., the possible input rankings can be arranged in a sequence $\succ_0,\dots,\succ_{{m\choose 2}}$ such that $\succ_0$ and $\succ_{{m\choose 2}}$ are reverse to each other and each consecutive pair of rankings differs only in a single pair or alternatives), the Kemeny distance becomes single-peaked over the input rankings. This suggests that variants of Moulin's phantom median rules \citep{Moul80a} may be strategyproof on this domain, and it seems worthwhile to further explore this question. In a similar vein, when randomization is allowed, the rule that picks a voter uniformly at random and returns her input ranking as the output ranking is easily seen to be strategyproof, unanimous, and anonymous. A natural question is thus to better understand which other methods satisfy these properties under randomization.

\section*{Funding and Personal Acknowledgements}

At the University of Amsterdam, Patrick Lederer is supported by the European Research Council (ERC) via the ERC Synergy Grant on ``Advancing Digital Democratic Innovation'' (\url{doi.org/10.3030/101166894}). Most of this research was carried out while Patrick Lederer was at UNSW Sydney, where he was supported by the NSF-CSIRO project on ``Fair Sequential Collective Decision Making'' (RG230833). Manuel Eberl, a coauthor of the conference version of this paper, was responsible for the Isabelle verification of the results as well as the SAT solving for the case that $m=4$ and $n=5$. In this context, he used the LEO HPC infrastructure of the University of Innsbruck and additional computing resources of the Formal Methods and Tools group at the University of Twente. After leaving academia, Manuel asked to not be listed as a coauthor of this manuscript. 
The author also thanks Dominik Peters for helpful discussions.

Funded by the European Union. Views and opinions expressed are, however, those of the authors only and do not necessarily reflect those of the European Union or the European Research Council Executive Agency. Neither the European Union nor the granting authority can be held responsible for them.

\section*{AI Acknowledgements}

We used AI, in particular ChatGPT-5.6 Sol and Claude Opus-5, to assist with the writing of this document. Specifically, this includes improving the language and the presentation of proofs. All proofs and definitions are the authors' own work and were derived without the assistance of AI.

\clearpage
\appendix

\section{Majority-consistency and Strategyproofness for $m=3$}\label{app:m3}

As mentioned in \Cref{rem:Condm3}, when $m=3$, the Kemeny rule satisfies both majority-consistency and strategyproofness. We will show next that, under mild additional assumptions, these two axioms even require us to choose a ranking minimizing the total Kemeny distance to the input rankings. To prove this claim, we extend the domain of SWFs from profiles for a fixed electorate $\mathcal{R}^N$ to the set of all profiles $\mathcal{R}^*=\bigcup_{N\subseteq\mathcal N\colon \text{$N$ is non-empty and finite}}\mathcal{R}^N$ that are defined for any non-empty and finite electorate. Furthermore, we define the \emph{majority margin} between two alternatives $x,y\in A$ in a profile $R$ by $g_{xy}(R)=|\{i\in N\colon x\succ_i y\}|-|\{i\in N\colon y\succ_i x\}|$, i.e., the majority margin between two alternatives counts how many more voters prefer $x$ to $y$ than vice versa. Lastly, we say an SWF $f:\mathcal{R}^*\rightarrow\mathcal{R}$ satisfies
\begin{itemize}[]
\item \emph{quasi-neutrality} if $f(\tau(R))=\tau(f(R))$ for all profiles $R$ and permutations $\tau:A\rightarrow A$ such that \emph{(i)} $g_{xy}(R)\neq g_{vw}(R)$ for all $v,w,x,y\in A$ with $v\neq w$, $x\neq y$, and $\{x,y\}\neq\{v,w\}$, and \emph{(ii)} the alternatives can be labeled such that $x\succsim_R y\succsim_R z\succsim_R x$.
\item \emph{cancellation} if $f(R)=f(R')$ for all profiles $R$ and $R'$ such that $R'$ arises from $R$ by adding two voters with inverse preferences.
\end{itemize}
Less formally, quasi-neutrality enforces a mild degree of neutrality for profiles where the majority relation is cyclic and all majority margins are unique. On the other hand, cancellation requires that adding pairs of voters with inverse preferences does not affect the outcome. As we show next, every SWF on $\mathcal{R}^*$ that satisfies strategyproofness, majority-consistency, anonymity, quasi-neutrality, and cancellation must always choose a Kemeny ranking, i.e., a ranking from the set $K(R)=\arg\min_{\rhd\in\mathcal{R}}\sum_{i\in N_R}\Delta(\succ_i,\rhd)$, where $N_R$ denotes the voters in $R$.

\begin{restatable}{proposition}{condm}\label{thm:Condm3}
Assume $m=3$. If an SWF $f$ on $\mathcal{R}^*$ satisfies strategyproofness, majority-consistency, anonymity, quasi-neutrality, and cancellation, then $f(R)\in K(R)$ for all profiles $R\in\mathcal{R}^*$.
\end{restatable}
\begin{proof}
Let $m=3$ and $f$ be a SWF on $\mathcal{R}^*$ that satisfies all given axioms. We will show in multiple steps that $f(R)\in K(R)$ for all profiles $R\in\mathcal{R}^*$.\medskip

\textbf{Step 1:} First, we show that $f$ only depends on the majority margins: it holds that $f(R)=f(R')$ for all profiles $R,R'\in\mathcal{R}^*$ such that $g_{xy}(R)=g_{xy}(R')$ for all alternatives $x,y\in A$. To prove this claim, we denote the six possible input rankings by
\begin{center}
\begin{tabular}{llllll}
     ${\succ_1}=abc$\hspace{0.5cm} & ${\succ_2}=bca$\hspace{0.5cm} & ${\succ_3}=cab$ \hspace{0.5cm}& 
     ${\succ_4}=cba$ \hspace{0.5cm} & ${\succ_5}=acb$ \hspace{0.5cm} & ${\succ_6}=bac$.
\end{tabular}
\end{center}
Furthermore, given a profile $R$, we denote by $n_i^R$ for $i\in\{1,\dots,6\}$ the number of voters who report the ranking $\succ_i$. By anonymity, we can compute $f$ based only on these six numbers. We further observe that $\succ_i$ is inverse to $\succ_{i+3}$ for $i\in\{1,2,3\}$. We hence define $\delta_i^R=n_i^R-n_{i+3}^R$ for all $i\in\{1,2,3\}$ to count how many more voters report $\succ_i$ than $\succ_{i+3}$ (or vice versa). By cancellation, these values are sufficient to determine $f$, as we can add or remove pairs of voters with inverse preferences without affecting the outcome. To prove this step, we next show that if $g_{xy}(R)=g_{xy}(R')$ for all $x,y\in A$, then also $\delta_i^R=\delta_i^{R'}$. This implies that $f(R)=f(R')$ by our previous insights, so $f$ can indeed be computed based only on the pairwise majority margins. To prove our claim, we observe that
\begin{align*}
    g_{ab}(R)=\delta_1^R-\delta_2^R+\delta_3^R,\qquad
    g_{bc}(R)=\delta_1^R+\delta_2^R-\delta_3^R,\qquad
    g_{ca}(R)=-\delta_1^R+\delta_2^R+\delta_3^R.
\end{align*}
By rearranging these equations, we infer that
\begin{align*}
\delta_1^R=\frac{g_{ab}(R)+g_{bc}(R)}{2},\qquad
\delta_2^R=\frac{g_{bc}(R)+g_{ca}(R)}{2},\qquad
\delta_3^R=\frac{g_{ca}(R)+g_{ab}(R)}{2}.
\end{align*}
Hence, the majority margins fully determine the values of $\delta_i^R$. In particular, if $g_{xy}(R)=g_{xy}(R')$ for all $x,y\in A$, then it also holds that $\delta_i^R=\delta_i^{R'}$ for all $i\in\{1,2,3\}$. This completes the proof of this step.\medskip

 \textbf{Step 2:} We next show that $f(R)\in K(R)$ for all profiles $R\in\mathcal{R}^*$. To this end, we first note that if the majority relation $\succsim_R$ of a profile $R$ is a ranking, then $K(R)=\{\succsim_R\}$ and our inclusion holds by majority-consistency. Furthermore, if $g_{xy}(R)=0$ for all $x,y\in A$, then $K(R)=\mathcal{R}$ and it holds trivially that $f(R)\in K(R)$. We next proceed with a case distinction regarding the structure of the profile.\smallskip

\emph{Case 2.1:} Let $R$ denote a profile such that $g_{ab}(R)=g_{bc}(R)=g_{ca}(R)>0$. For this profile, it holds that $K(R)=\{abc,bca,cab\}$, and we assume for contradiction that $f(R)\notin\{abc,bca,cab\}$. The three possible cases are analogous, so suppose that $f(R)=cba$. Using cancellation, we next add $g_{ab}(R)$ pairs of voters such that one voter reports $abc$ and the other voter reports $cba$. This leads to a new profile $R'$ with $f(R')=cba$ due to cancellation. Finally, we let the newly added voters who report $abc$ deviate one after another to $bac$. Since these voters completely disagree with $f(R')$, strategyproofness requires that the outcome does not change at any step. However, in the resulting profile $R''$, we have $g_{ba}(R'')=g_{ab}(R')>0$, $g_{bc}(R'')>0$, and $g_{ca}(R'')>0$, i.e., the majority relation is transitive. Hence, majority-consistency requires that the ranking $bca$ is chosen for $R''$, which contradicts strategyproofness. This proves that the assumption that $f(R)\notin\{abc,bca,cab\}$ is false.\smallskip

\emph{Case 2.2:} We next turn to profiles such that the majority relation is cyclic and all majority margins are unique and non-zero. More specifically, we analyze profiles $R$ such that $g_{ab}(R)>0$, $g_{bc}(R)>0$, $g_{ca}(R)>0$, and $g_{xy}(R)\neq g_{yz}(R)$ for all distinct $x,y,z\in A$. We start by again considering a profile $R^*$ such that $g_{ab}(R^*)=g_{bc}(R^*)=g_{ca}(R^*)=\lambda$ for some positive integer $\lambda$. By Case~2.1, it holds that $f(R^*)\in\{abc,bca,cab\}$. The three cases are analogous, so suppose that $f(R^*)=abc$. Using cancellation, we next add pairs of voters to $R^*$ such that one voter reports $abc$ and the other reports $cba$. By strategyproofness, the voters who report $cba$ cannot change the outcome by deviating, as any other outcome is better for them than $abc$. In particular, this means that, when swapping $b$ and $c$ or $a$ and $b$, the outcome does not change. By this argument and Step~1, it follows that $f(R')=abc$ for all profiles $R'$ such that $g_{ab}(R')\geq g_{ab}(R^*)$, $g_{bc}(R')\geq g_{bc}(R^*)$, and $g_{ca}(R')=\lambda$. In particular, this holds when $g_{ab}(R')\neq g_{bc}(R')$, $g_{ab}(R')>g_{ca}(R')$, and $g_{bc}(R')>g_{ca}(R')$, i.e., when all majority margins are different and $g_{ca}(R')$ has the smallest value. Put differently, for all such profiles, we ``cut'' the edge in the majority graph with the smallest weight, which means that $f(R')\in K(R')=\{abc\}$. Further, since all majority margins are unique in $R'$, we can generalize this insight using quasi-neutrality to all profiles whose majority margins can be obtained from those of $R'$ by permuting the alternatives. This shows that $f(R)\in K(R)$ for all profiles such that the smallest majority margin has value $\lambda$ and all majority margins are different. Finally, we note that we can apply this argument for every positive integer $\lambda$. Hence, $f(R)\in K(R)$ for all profiles with unique non-zero majority margins and a cyclic majority relation.\smallskip

\emph{Case 2.3:} Third, we consider the case where the majority relation is cyclic and the majority margins have two different non-zero values. To this end, let $R$ denote a profile and $\lambda_1,\lambda_2\in\mathbb{N}$ denote two non-zero integers such that $g_{ab}(R)=\lambda_1$, $g_{bc}(R)=\lambda_1$, and $g_{ca}(R)=\lambda_2$. First, suppose that $\lambda_1>\lambda_2$, in which case $K(R)=\{abc\}$. Suppose for contradiction that $f(R)\neq abc$. We first note that it can be shown analogously to Case~2.1 that $f(R)\notin\{cba,bac,acb\}$, as we can otherwise deviate to a profile with a transitive majority relation. It thus holds that $f(R)\in\{bca,cab\}$. Suppose that $f(R)=bca$. In this case, we add a pair of voters with rankings $bca$ and $acb$, and cancellation shows that the outcome remains $bca$. Next, let the voter reporting $acb$ swap $b$ and $c$, resulting in a profile $R'$ with $g_{ab}(R')=\lambda_1$, $g_{bc}(R')=\lambda_1+2$, and $g_{ca}(R')=\lambda_2$. Moreover, strategyproofness requires that $f(R')=bca$, as any other outcome constitutes a manipulation. However, this contradicts Case~2.2, as all majority margins are unique and non-zero and $g_{ca}(R')$ is minimal. On the other hand, if $f(R)=cab$, we can apply an analogous argument by adding a pair of voters with preferences $cab$ and $bac$ and letting the latter voter reinforce $a$ against $b$.

As the second subcase, suppose that $\lambda_1<\lambda_2$, which means that $K(R)=\{bca,cab\}$. Using again the argument of Case~2.1, it follows that $f(R)\in\{abc,bca,cab\}$. So, we assume for contradiction that $f(R)=abc$. In this case, we can add pairs of voters with rankings $abc$ and $cba$, which does not affect the outcome due to cancellation. Moreover, strategyproofness requires that the voters reporting $cba$ cannot change the outcome by deviating. By letting these voters swap $b$ and $c$, we can increase the majority margin $g_{bc}(R)$ arbitrarily, so we infer that $f(R')=abc$ for all profiles $R'$ with $g_{ab}(R')=\lambda_1$, $g_{bc}(R')\geq\lambda_1$, and $g_{ca}(R')=\lambda_2$. However, once $g_{bc}(R')>\lambda_2$, this conflicts with Case~2.2: in this case, the smallest majority margin is $g_{ab}(R')$, so $f(R')$ must be $bca$.\smallskip

\emph{Case 2.4:} In our fourth case, we assume that exactly one majority margin has value $0$. Without loss of generality, we suppose that $g_{ac}(R)=0$ and consider three subcases. Moreover, we assume that every ranking is reported by at least one voter; this is without loss of generality due to cancellation. As the first subcase, suppose that $g_{bc}(R)>0$ and $g_{ba}(R)>0$, which implies that $K(R)=\{bca,bac\}$. Assume for contradiction that $f(R)\notin K(R)$. If $f(R)=acb$, let $R'$ denote the profile derived from $R$ by letting a voter change her ranking from $bca$ to $bac$. Strategyproofness requires that the outcome does not change. However, the majority relation now agrees with the ranking $bac$, so majority-consistency requires that $f(R')=bac$, a contradiction. If $f(R)=abc$, let $R'$ denote the profile derived from $R$ by letting a voter change her ranking from $cba$ to $bac$. Strategyproofness requires that $f(R')=abc$, but majority-consistency requires that $f(R')=bac$, hence yielding a contradiction. Finally, the cases $f(R)=cab$ and $f(R)=cba$ are symmetric to $f(R)=acb$ and $f(R)=abc$, respectively.

As the second subcase, suppose that $g_{bc}(R)<0$ and $g_{ba}(R)<0$, so $K(R)=\{acb,cab\}$. If $f(R)=bca$, a voter can benefit by manipulating from the ranking $acb$ to $cab$. In particular, after this deviation, the majority relation agrees with the ranking $cab$, so this ranking must be chosen, but $\Delta(acb,bca)=3>1=\Delta(acb,cab)$. On the other hand, if $f(R)=cba$, a voter can manipulate by deviating from the ranking $abc$ to $cab$. After this step, the outcome must be $cab$ by majority-consistency, which again decreases the Kemeny distance of the manipulator. The cases $f(R)=bac$ and $f(R)=abc$ are again symmetric, so it follows that $f(R)\in K(R)$.

As the third subcase, we suppose that $g_{ab}(R)>0$ and $g_{bc}(R)>0$ and note that the case $g_{ab}(R)<0$ and $g_{bc}(R)<0$ is symmetric. Let $R^*$ denote a profile with $g_{xy}(R^*)=0$ for all $x,y\in A$ and, after suitably labeling the alternatives, suppose that $f(R^*)=bca$. Using the same argument as in Case~2.2, it holds for every profile $R'$ with $g_{bc}(R')\geq0$, $g_{ca}(R')\geq0$, and $g_{ab}(R')=0$ that $f(R')=bca$ as well. Based on quasi-neutrality, this insight generalizes to all profiles such that one majority margin is $0$, all three majority margins are distinct, and the majority relation is cyclic. Hence, if $g_{ab}(R)\neq g_{bc}(R)$, it follows for our profile $R$ that $f(R)=abc$ and thus $f(R)\in K(R)$. So, suppose that $g_{ab}(R)=g_{bc}(R)$ and assume that $f(R)\neq abc$. If $f(R)=cba$, we repeatedly let voters with true ranking $abc$ change their reported ranking to $acb$. (Note that we may use cancellation to add these voters.) This eventually results in a profile $R'$ with $g_{ab}(R')>0$, $g_{cb}(R')>0$, and $g_{ac}(R')=0$, so $f(R')\in\{acb,cab\}$ by the second subcase. However, this contradicts strategyproofness, as $cba$ is the reverse of the manipulators' true ranking $abc$. Next, if $f(R)=bac$, we let a voter with ranking $cab$ swap $a$ and $c$. For the resulting profile $R'$, majority-consistency requires that $f(R')=abc$, which constitutes a manipulation. A similar construction works if $f(R)=acb$. Lastly, if $f(R)=bca$ or $f(R)=cab$, we can infer a manipulation analogously to Case~2.3 by ensuring that all majority margins are unique.\smallskip

\emph{Case 2.5:} As the last case, suppose that exactly two majority margins are zero. More specifically, we consider the case that $g_{ab}(R)>0$ and $g_{bc}(R)=g_{ca}(R)=0$, so $K(R)=\{abc,acb,cab\}$. The other cases are analogous. Moreover, using cancellation, we may assume that every ranking is reported by at least one voter. Suppose for contradiction that $f(R)\notin K(R)$. If $f(R)=bac$, let $R'$ denote the profile derived from $R$ by letting a voter deviate from $cab$ to $acb$. Consequently, $g_{ab}(R')>0$, $g_{ac}(R')>0$, and $g_{bc}(R')=0$, so $f(R')\in\{abc,acb\}$ by Case~2.4. Regardless of the exact outcome, this constitutes a manipulation. If $f(R)=bca$, we can use an analogous construction to infer a contradiction. Finally, if $f(R)=cba$, we let a voter deviate from $abc$ to $acb$. This results in a profile $R'$ with $g_{ab}(R')>0$, $g_{cb}(R')>0$, and $g_{ac}(R')=0$. Hence, $f(R')\in\{acb,cab\}$ by Case~2.4, which again constitutes a beneficial manipulation. Hence, it holds that $f(R)\in K(R)$ in all cases.
\end{proof}

\section{Proof of \Cref{lem:k2}}\label{app:approx}

In this appendix, we prove the auxiliary lemma characterizing when two rankings can jointly minimize the total Kemeny distance for some profile. We recall that the transposition graph $G_{\rhd_1,\rhd_2}$ for two rankings $\rhd_1$ and $\rhd_2$ is the undirected graph $(V,E)$, where $V=\{\{x,y\}\subseteq A\colon x\rhd_1 y\iff y\rhd_2 x\}$ and $\{v_1,v_2\}\in E$ if and only if there are three distinct alternatives $x,y,z\in A$ such that $v_1=\{x,y\}$, $v_2=\{y,z\}$, and $\{x,z\}\notin V$. Moreover, in this appendix, we prove that, if $m\geq 5$, there is no pair of rankings $\rhd_1$ and $\rhd_2$ with $\Delta(\rhd_1,\rhd_2)\geq {m\choose 2}-(m-2)$ that jointly minimize the total Kemeny distance for some profile. This means that the rankings we consider for the lower bound on the incentive ratio of the Kemeny rule indeed maximize their Kemeny distance subject to the condition that they can jointly minimize the total Kemeny distance for some profile.

\kemenypairs*
\begin{proof}
Fix two rankings $\rhd_1$ and $\rhd_2$. We prove this result in three steps. First, we show that the Kemeny rule can be seen as an optimization problem over the linear ordering polytope. Based on this insight, we then show that two rankings $\rhd_1$ and $\rhd_2$ can be the unique minimizers of the total Kemeny distance in a profile if they correspond to neighboring vertices of the linear ordering polytope. Finally, by applying Young's (\citeyear{Youn78a}) characterization of neighboring vertices of this polytope, we obtain the result.\medskip

\emph{Step 1:} We start by reformulating the Kemeny rule in terms of the linear ordering polytope. To this end, for each ranking ${\rhd}\in\mathcal{R}$, we define the matrix $M^{\rhd}\in\{0,1\}^{m\times m}$ by $M^{\rhd}_{xy}=1$ if $x\rhd y$ and $M^{\rhd}_{xy}=0$ otherwise. Then, the linear ordering polytope $\mathcal{M}$ is the convex hull of the set of all ranking matrices, i.e., $\mathcal{M}=\mathit{conv}(\{M^{\rhd}\colon {\rhd}\in\mathcal{R}\})$. In this step, we show that, for every linear optimization problem $\max_{M\in\mathcal{M}}\sum_{x,y\in A}C_{xy}\cdot M_{xy}$ with $C\in\mathbb{Q}^{m\times m}$, there is a profile $R$ such that a ranking $\rhd$ minimizes the total Kemeny distance in $R$ if and only if $M^{\rhd}$ maximizes $\sum_{x,y\in A}C_{xy}\cdot M_{xy}$ subject to $M\in\mathcal{M}$.

To this end, we first note that the set $\arg\max_{M\in\mathcal{M}}\sum_{x,y\in A}C_{xy}\cdot M_{xy}$ is invariant under adding the same constant $t$ to $C_{xy}$ and $C_{yx}$ for a single pair of distinct alternatives $x,y\in A$. In more detail, fix a pair of distinct alternatives $x,y\in A$ and a constant $t$, and define the matrix $C'$ by $C'_{xy}=C_{xy}+t$, $C'_{yx}=C_{yx}+t$, and $C'_{x'y'}=C_{x'y'}$ for all other pairs of alternatives $x',y'\in A$. Now, for each point $M\in\mathcal{M}$, it holds that $M_{xy}+M_{yx}=1$, so $\sum_{x',y'\in A}C'_{x'y'}\cdot M_{x'y'}=t+\sum_{x',y'\in A}C_{x'y'}\cdot M_{x'y'}$. This immediately implies that $\arg\max_{M\in\mathcal{M}}\sum_{x,y\in A}C_{xy}\cdot M_{xy}=\arg\max_{M\in\mathcal{M}}\sum_{x,y\in A}C'_{xy}\cdot M_{xy}$. Now, by shifting each pair $x,y$ by $-(C_{xy}+C_{yx})/2$, we derive a matrix $C'$ such that $C'_{xy}=-C'_{yx}$ for all distinct $x,y\in A$ and the set of maximizers is unchanged. Furthermore, since $C$ is rational, so is $C'$. Thus, there is a positive scalar $\lambda$ such that $C''=\lambda C'\in\mathbb{Z}^{m\times m}$ contains only even entries and remains skew-symmetric. Moreover, it still holds that $\arg\max_{M\in\mathcal{M}}\sum_{x,y\in A}C''_{xy}\cdot M_{xy}=\arg\max_{M\in\mathcal{M}}\sum_{x,y\in A}C_{xy}\cdot M_{xy}$, as scaling the objective by a positive constant does not affect the set of maximizers.

Next, we define the majority margin for an alternative $x$ against another alternative $y$ in a profile $R$ by $g_{xy}(R)=|\{i\in N\colon x\succ_i y\}|-|\{i\in N\colon y\succ_i x\}|$. Since $C''$ is skew-symmetric and all its entries are even, we can use McGarvey's construction \citep{McGa53a} to find a profile $R$ such that $g_{xy}(R)=C''_{xy}$ for all distinct $x,y\in A$. Finally, we claim that the matrix $M^{\rhd}$ maximizes $\sum_{x,y\in A}C''_{xy}\cdot M_{xy}$ subject to $M\in\mathcal{M}$ if and only if $\rhd$ minimizes the total Kemeny distance in $R$. To prove this, we first note that, for every linear optimization problem over a polytope, there is a vertex of the polytope that optimizes the objective. Since the vertices of $\mathcal{M}$ are the matrices $M^{\rhd}$, this means that there is a ranking $\rhd$ such that $M^{\rhd}$ solves $\max_{M\in\mathcal{M}}\sum_{x,y\in A}C''_{xy}\cdot M_{xy}$. Next, we observe for each ranking ${\rhd}\in\mathcal{R}$ that

\begin{align*}
\sum_{x,y\in A} C''_{xy}\cdot M_{xy}^{\rhd}
&=\sum_{x,y\in A\colon x\rhd y} g_{xy}(R)\\
&=\sum_{x,y\in A\colon x\rhd y}
\left(|\{i\in N\colon x\succ_i y\}|-|\{i\in N\colon y\succ_i x\}|\right)\\
&=\sum_{i\in N} |\{(x,y)\in A^2\colon x\succ_i y\land x\rhd y\}|-|\{(x,y)\in A^2\colon y\succ_i x\land x\rhd y\}|\\
&=\sum_{i\in N}\left(u(\succ_i,\rhd)-\Delta(\succ_i,\rhd)\right)\\
&={m\choose 2}\cdot |N|-2\sum_{i\in N}\Delta(\succ_i,\rhd).
\end{align*}
Here, the first equality uses that $C''_{xy}=g_{xy}(R)$ for all distinct $x,y\in A$ and that $M_{xy}^{\rhd}=1$ if $x\rhd y$ and $M_{xy}^{\rhd}=0$ otherwise. The second equality substitutes the definition of $g_{xy}(R)$. The third equality changes the summation from pairs of alternatives to voters, and the fourth equality uses the definitions of the Kemeny utility and distance. For the final step, we use that $u(\succ_i,\rhd)={m\choose 2}-\Delta({\succ_i},{\rhd})$ for all $i\in N$. This equation proves that if $\rhd$ maximizes $\sum_{x,y\in A}C''_{xy}\cdot M_{xy}^{\rhd}$, then it minimizes $\sum_{i\in N}\Delta({\succ_i},{\rhd})$ because the two objectives are affine transformations of each other. Moreover, if $\{\rhd_1,\rhd_2\}=\arg\min_{\rhd\in\mathcal{R}}\sum_{i\in N}\Delta(\succ_i,\rhd)$ for some profile $R$, the same equation implies that $M^{\rhd_1}$ and $M^{\rhd_2}$ are the only vertices of $\mathcal{M}$ that maximize $\sum_{x,y\in A}g_{xy}(R)\cdot M_{xy}$, which proves the converse direction.\medskip

\emph{Step 2:} Based on the previous step, it suffices to find a rational objective $C$ such that precisely the matrices $M^{\rhd_1}$ and $M^{\rhd_2}$ are the only vertices of $\mathcal{M}$ that maximize $\sum_{x,y\in A}C_{xy}\cdot M_{xy}$ subject to $M\in\mathcal{M}$. Since the set of points that maximize a linear objective over a polytope is the convex hull of some of its vertices, this means that we are looking for a rational objective $C$ such that $\arg\max_{M\in\mathcal{M}}\sum_{x,y\in A}C_{xy}\cdot M_{xy}=\mathit{conv}(\{M^{\rhd_1},M^{\rhd_2}\})$. Such an objective exists precisely if $\mathit{conv}(\{M^{\rhd_1},M^{\rhd_2}\})$ is a face of $\mathcal{M}$. Indeed, a face of a polytope is defined as the set of points in the polytope that maximize some real linear objective. Hence, if $\mathit{conv}(\{M^{\rhd_1},M^{\rhd_2}\})$ is not a face of $\mathcal{M}$, then, by Step~1, there is no profile for which precisely $\rhd_1$ and $\rhd_2$ minimize the total Kemeny distance.

On the other hand, if $\mathit{conv}(\{M^{\rhd_1},M^{\rhd_2}\})$ is a face of $\mathcal{M}$, we claim that it is also exposed by a rational objective, since $\mathcal{M}$ has only rational vertices. To see this, assume there is an objective $C\in\mathbb{R}^{m\times m}$ such that $\arg\max_{M\in\mathcal{M}}\sum_{x,y\in A}C_{xy}\cdot M_{xy}=\mathit{conv}(\{M^{\rhd_1},M^{\rhd_2}\})$. Equivalently, this means that $\sum_{x,y\in A}C_{xy}\cdot M_{xy}^{\rhd_1}=\sum_{x,y\in A}C_{xy}\cdot M_{xy}^{\rhd_2}$ and $\sum_{x,y\in A}C_{xy}\cdot M_{xy}^{\rhd_1}>\sum_{x,y\in A}C_{xy}\cdot M_{xy}^{\rhd}$ for all ${\rhd}\in\mathcal{R}\setminus\{\rhd_1,\rhd_2\}$. Now, if $C$ is real, we can find a rational approximation $C^\epsilon$ of $C$ that satisfies all these conditions as well. To see this, fix a pair of alternatives $x^*,y^*$ such that $x^*\rhd_1 y^*$ and $y^*\rhd_2 x^*$, which implies that $M^{\rhd_1}_{x^*y^*}-M^{\rhd_2}_{x^*y^*}=1$. Further, for every $\epsilon>0$, we define a matrix $C^\epsilon$ as follows. First, for each ordered pair $(x,y)\neq(x^*,y^*)$, we set $C^\epsilon_{xy}$ to an arbitrary rational value in $[C_{xy}-\epsilon,C_{xy}+\epsilon]$. Second, let $C^\epsilon_{x^*y^*}=-\sum_{(x,y)\in A^2\setminus\{(x^*,y^*)\}}C^\epsilon_{xy}\cdot\left(M^{\rhd_1}_{xy}-M^{\rhd_2}_{xy}\right)$. First, $C^\epsilon$ is rational by definition. Moreover, by the choice of $C^\epsilon_{x^*y^*}$, it holds that $\sum_{x,y\in A}C^\epsilon_{xy}\cdot\left(M^{\rhd_1}_{xy}-M^{\rhd_2}_{xy}\right)=0$. Finally, by choosing $\epsilon$ sufficiently small, all strict inequalities with rankings $\rhd\in\mathcal{R}\setminus\{\rhd_1,\rhd_2\}$ are preserved. Thus, $C^\epsilon$ is a rational objective that exposes precisely $\mathit{conv}(\{M^{\rhd_1},M^{\rhd_2}\})$.

Finally, we note that two vertices $v_1$ and $v_2$ of a polytope are called neighboring if and only if $\mathit{conv}(\{v_1,v_2\})$ is a face of the polytope. Hence, there is a profile $R$ such that $\rhd_1$ and $\rhd_2$ exclusively minimize the total Kemeny distance if and only if $M^{\rhd_1}$ and $M^{\rhd_2}$ are neighboring vertices of $\mathcal{M}$.\medskip

\emph{Step 3:} Based on Steps 1 and 2, the lemma follows by showing that the matrices $M^{\rhd_1}$ and $M^{\rhd_2}$ induced by our two rankings $\rhd_1$ and $\rhd_2$ are neighbors in the linear ordering polytope if and only if their transposition graph $G_{\rhd_1,\rhd_2}$ is connected. This follows directly from a characterization by \citet{GiHo64a} and \citet{Youn78a}. In particular, these authors have shown that $M^{\rhd_1}$ and $M^{\rhd_2}$ are neighbors in $\mathcal{M}$ if and only if $G_{\rhd_1,\rhd_2}$ is connected. We also note that, e.g., the characterizations of the Kemeny rule by \citet{YoLe78a} and of bivariate scoring rules by \citet{Lede23a} exactly builds on this insight. By combining our three steps, the lemma follows now.
\end{proof}
Next, we show that, if $m\geq 5$, there is no pair of rankings with a Kemeny distance larger than ${m\choose 2}-(m-1)$ that exclusively minimize the total Kemeny distance in some profile. This means that the choice of rankings for the lower bound on the incentive ratio of the Kemeny rule is optimal for our proof strategy.

\begin{lemma}\label{lem:maxdistant}
When $m\geq 5$, there is no pair of rankings $\rhd_1$ and $\rhd_2$ such that $\Delta(\rhd_1,\rhd_2)\geq {m\choose 2}-(m-2)$ and $\{\rhd_1,\rhd_2\}=\arg\min_{{\rhd}\in\mathcal{R}}\sum_{i\in N}\Delta(\succ_i,\rhd)$ for some profile $R$.
\end{lemma}

\begin{proof}
To prove this claim, we fix two rankings $\rhd_1$ and $\rhd_2$ with $\Delta(\rhd_1,\rhd_2)\geq {m\choose 2}-(m-2)$. By \Cref{lem:k2}, our result follows by showing that the transposition graph $G_{\rhd_1,\rhd_2}=(V,E)$ of our two rankings is disconnected. To prove this claim, we define another graph: $H_{\rhd_1,\rhd_2}=(V_H,E_H)$ is the undirected graph on the alternatives (i.e., $V_H=A$) such that $\{x,y\}\in E_H$ if and only if $x\rhd_1 y\iff x\rhd_2 y$. If $H_{\rhd_1,\rhd_2}$ is connected, it must have at least $m-1$ edges, which means that there are at least $m-1$ pairs of alternatives on which $\rhd_1$ and $\rhd_2$ agree. This contradicts $\Delta(\rhd_1,\rhd_2)\geq {m\choose 2}-(m-2)$, so we assume that $H_{\rhd_1,\rhd_2}$ is disconnected. We next show that, if $G_{\rhd_1,\rhd_2}$ is connected and $H_{\rhd_1,\rhd_2}$ is disconnected, then $G_{\rhd_1,\rhd_2}$ has fewer than ${m\choose 2}-(m-2)$ vertices. This again contradicts $\Delta(\rhd_1,\rhd_2)\geq {m\choose 2}-(m-2)$. Hence, we derive that if $\Delta(\rhd_1,\rhd_2)\geq {m\choose 2}-(m-2)$, then $G_{\rhd_1,\rhd_2}$ must be disconnected.

Now, assume for contradiction that $G_{\rhd_1,\rhd_2}$ is connected. To derive a contradiction, for every alternative $x\in A$, let $K(x)$ denote a label for the connected component of $H_{\rhd_1,\rhd_2}$ containing $x$. That is, $K$ is a function such that $K(x)=K(y)$ if there is a path from $x$ to $y$ in $H_{\rhd_1,\rhd_2}$ and $K(x)\neq K(y)$ if no such path exists. Now, let $v_1=\{x,y\}$ and $v_2=\{y,z\}$ denote two vertices of the transposition graph $G_{\rhd_1,\rhd_2}$ such that $\{v_1,v_2\}\in E$. This requires that $\rhd_1$ and $\rhd_2$ agree on the order of $x$ and $z$, so these alternatives are connected in $H_{\rhd_1,\rhd_2}$. Thus, $K(x)=K(z)$ and $\{K(x),K(y)\}=\{K(y),K(z)\}$. Consequently, all neighboring vertices $v_1=\{x_1,y_1\}$ and $v_2=\{x_2,y_2\}$ in $G_{\rhd_1,\rhd_2}$ are labeled by the same pair $\{K(x_1),K(y_1)\}=\{K(x_2),K(y_2)\}$. Moreover, if $G_{\rhd_1,\rhd_2}$ is connected, we can repeatedly apply this argument to infer that every vertex in $G_{\rhd_1,\rhd_2}$ has the same pair of labels.

Now, we proceed with a case distinction regarding the connected components of $H_{\rhd_1,\rhd_2}$. First, assume that $H_{\rhd_1,\rhd_2}$ has at least three connected components $X_1,X_2,X_3$. This means that, for all distinct $i,j\in\{1,2,3\}$, $\rhd_1$ and $\rhd_2$ disagree on all pairs of alternatives $x_i\in X_i$ and $x_j\in X_j$. Thus, we can fix three alternatives $x_1\in X_1$, $x_2\in X_2$, and $x_3\in X_3$ and note that $\{x_1,x_2\}$ as well as $\{x_1,x_3\}$ are vertices in $G_{\rhd_1,\rhd_2}$. Moreover, it holds that $\{K(x_1),K(x_2)\}\neq\{K(x_1),K(x_3)\}$ by the definition of $K$, since there is no path from $x_2$ to $x_3$ in $H_{\rhd_1,\rhd_2}$. However, if $G_{\rhd_1,\rhd_2}$ is connected, all vertices must have the same pair of labels, a contradiction.

As a second case, suppose that there are two connected components $X_1$ and $X_2$ in $H_{\rhd_1,\rhd_2}$, and denote their labels by $\ell_1$ and $\ell_2$. We note that $H_{\rhd_1,\rhd_2}$ cannot have a single connected component since it is disconnected. By definition, for all $x_1\in X_1$ and $x_2\in X_2$, the vertex $\{x_1,x_2\}$ belongs to $G_{\rhd_1,\rhd_2}$ and is labeled $\{\ell_1,\ell_2\}$. Now, if there is a pair of alternatives $x,y\in X_1$ or $x,y\in X_2$ such that $\rhd_1$ and $\rhd_2$ disagree on the order of $x$ and $y$, then $\{x,y\}$ is a vertex of $G_{\rhd_1,\rhd_2}$. Moreover, this vertex is labeled $\{\ell_1,\ell_1\}$ or $\{\ell_2,\ell_2\}$ and thus has a different label from the vertices containing one alternative from $X_1$ and one alternative from $X_2$. Since this is not possible if $G_{\rhd_1,\rhd_2}$ is connected, it follows that no such vertex exists. Hence, $G_{\rhd_1,\rhd_2}$ has precisely $|X_1|\cdot|X_2|$ vertices. Because $|X_1|+|X_2|=m$, this term is upper-bounded by $\lfloor\frac{m}{2}\rfloor\cdot\lceil\frac{m}{2}\rceil$. Now, if $m\geq 5$ is even, it holds that $\lfloor\frac{m}{2}\rfloor\cdot\lceil\frac{m}{2}\rceil=\frac{m^2}{4}<\frac{m^2-3m+4}{2}={m\choose 2}-(m-2)$. Similarly, if $m\geq 5$ is odd, we have that $\lfloor\frac{m}{2}\rfloor\cdot\lceil\frac{m}{2}\rceil=\frac{m^2-1}{4}<\frac{m^2-3m+4}{2}={m\choose 2}-(m-2)$. Hence, in both cases, we derive that $G_{\rhd_1,\rhd_2}$ has fewer than ${m\choose 2}-(m-2)$ vertices, which contradicts $\Delta(\rhd_1,\rhd_2)\geq {m\choose 2}-(m-2)$. This completes the proof of our lemma.
\end{proof}

\newgeometry{left=1.5cm,right=1.5cm,top=2cm,bottom=2cm}
\section{Proof of Proposition 1}\label{app:mainproof}

In this section, we will prove one of the base cases of our main impossibility theorem: if there are $n=2$ voters and $m=5$ alternatives, no anonymous SWF satisfies both unanimity and strategyproofness. We will prove this statement by contradiction and hence assume that there is an SWF $f$ for the given numbers of voters and alternatives that satisfies our axioms. To derive a contradiction, we will subsequently reason about numerous profiles and show that, regardless of which outcomes we choose at certain profiles, strategyproofness must be violated. We start by considering the profile $R^*$ shown below. 
\begin{center}
\begin{tabular}{lll}
    $R^*$: & $abcde$ & $acbed$
    \end{tabular}
\end{center}

We first note that for this profile, it holds that $f(R^*)\in \{abcde, abced, acbde, acbed\}$ because of unanimity. Moreover, an analogous statement holds for all profiles $R'=\pi(R^*)$ that are derived by permuting the alternatives in $R^*$ according to a bijection $\pi:A\rightarrow A$. In more detail, it is easy to see that, for each permutation $\pi:A\rightarrow A$, it holds that $f(\pi(R^*))\in \{\pi(abcde), \pi(abced), \pi(acbde), \pi(acbed)\}$ because unanimity does not depend on the names of the alternatives.

As the first step of our proof, we will show that there is a permutation $\pi:A\rightarrow A$ such that $f(\pi(R^*))\in \{\pi(abcde), \pi(acbed)\}$. Assume for contradiction that this is not true, so $f(\pi(R^*))\in \{\pi(abced), \pi(acbde)\}$ for all permutations $\pi:A\rightarrow A$. In particular, this means that $f(R^*)\in \{abced, acbde\}$. We further note that $b$ and $c$ as well as $c$ and $d$ are symmetric in $R^*$, i.e., it holds that $\pi^*(R^*)=R^*$ for the permutation $\pi^*$ given by $\pi^*(a)=a$, $\pi^*(b)=c$, $\pi^*(c)=b$, $\pi^*(d)=e$, and $\pi^*(e)=d$. Consequently, we can assume that $f(R^*)=acbde$; the case that $f(R^*)=abced$ follows by simply permuting all profiles and rankings in the subsequent argument with respect to $\pi^*$. 

In the following table, we show that our assumptions are in conflict with each other as no feasible outcome remains for the profile $R^4$. Note that we display this simple derivation in the highly compressed form that is used throughout this section. In particular, we will present all proofs in tabular forms, where each row consists of a profile marker (left most column), the two rankings that make up the profile (second and third column), and all feasible outcomes. We note that the feasible outcomes are either determined by our assumptions (as, e.g., for $R^*$ in the following table) or correspond to the set of rankings that satisfy unanimity for the given profile (as, e.g., for $R^2$ in the following table). Moreover, profiles may appear multiple times in our derivations as we may infer additional information about the possible outcomes.

Based on the assumptions, we show that strategyproofness rules out specific rankings at given profiles. For instance, in the following derivation, the fact that $f(R^*)=acbde$ entails that $f(R^2)\neq acbed$ as the first voter can otherwise manipulate by deviating from $R^2$ to $R^*$. Hence, we know that $acbde$ must be chosen for $R^2$, and we will use this fact in further derivations. More generally, all grayed out rankings in the right most column satisfy unanimity, but violate strategyproofness due to the outcome that has been inferred for the profile indicated in brackets. We note here that strategyproofness may apply in either directions (i.e., either a voter manipulated from the considered profile to the one in the brackets or vice versa). Moreover, it is possible that multiple rankings are feasible outcomes for the profile indicated in brackets; in this case, strategyproofness rules out that the indicated ranking is chosen, regardless of the exact outcome for a given profile. Note that strategyproofness may be applied in either of the two directions for each of the possible outcomes at the manipulated profile. The second possibility that a ranking is grayed out is that we explicitly assume that this outcome is not chosen (e.g., for $R^4$, our assumption that $f(\pi(R^*)\in \{\pi(abced), \pi(acbde)\}$ for all permutations $\pi:A\rightarrow A$ rules out that $f(R^4)=abced$ or $f(R^4)=acbde$). Lastly, all our proofs end in a profile where no valid outcome remains, thereby showing that our axioms the assumptions are incompatible with each other.

{\setlength{\tabcolsep}{3pt}
\noindent\begin{longtable}{c|cc|llllllll}
$R^{*}$ & $abcde$ & $acbed$ & $acbde$ (A)\\
$R^{2}$ & $acbde$ & $acbed$ & $acbde$ & \textcolor{gray}{$acbed\,(R^{*})$} \\
$R^{3}$ & $abcde$ & $acbde$ & \textcolor{gray}{$abcde\,(R^{*})$} & $acbde$ \\
$R^{4}$ & $abced$ & $acbde$ & \textcolor{gray}{$abcde\,(R^{3})$} & \textcolor{gray}{$acbed\,(R^{2})$} & \textcolor{gray}{$abced$ (A)} & \textcolor{gray}{$acbde$ (A)}\\
\end{longtable}
}

For this simple derivation, it is straightforward to translate our tabular form into natural language. We assume that $f(R^*)=acbde$. This implies that $f(R^2)=f(R^3)=acbde$ as these profiles are derived by letting one of the voters in $R^*$ deviate to $acbde$. Lastly, for $R^4$, we have by assumption that $f(R^4)\neq abced$ and $f(R^4)\neq acbde$. However, if we choose $f(R^4)=abcde$, voter $1$ can manipulate by deviating from $R^3$ to $R^4$. Similarly, if $f(R^4)=acbed$, voter $2$ can manipulate by deviating from $R^2$ to $R^4$. Hence, no ranking that satisfies unanimity remains for $R^4$, thereby showing that our assumptions are in conflict. 

By this derivation, we know that $f(\pi(R^*))\in \{\pi(abcde), \pi(acbed)\}$ for some permutation $\pi:A\rightarrow A$. We will subsequently assume that $\pi$ is given by the identity, i.e., that $f(R^*)\in \{abcde, acbed\}$. Our proof applies to any other permutation $\pi$ by simply renaming the alternatives in all proofs and outcomes accordingly as both unanimity and strategyproofness are independent of the names of alternatives. Further, by the fact that $b$ and $c$, as well as $d$ and $e$ are symmetric in $R^*$, we can assume without loss of generality that $f(R^*)=abcde$; if $f(R^*)=acbed$, we can just exchange the roles of $b$ and $c$ as well as $d$ and $e$ in the subsequent proofs. 

Now, from here on, we proceed with a case distinction with respect to the profile $R^+$ shown below. 
\begin{center}
\begin{tabular}{lll}
    $R^+$: & $eabcd$ & $ecabd$
    \end{tabular}
\end{center}

We note that only three outcomes satisfy unanimity for this profile: $f(R^+)=eabcd$, $f(R^+)=eacbd$, or $f(R^+)=ecabd$. We will next show that all of these three cases result in a contradiction. To give further structure to our proof, we will discuss each of these cases as a separate lemma. In particular, we show in the next three lemmas that none of these outcomes is compatible with the fact that $f$ is anonymous, strategyproof, unanimous, and satisfies that $f(R^*)=abcde$. Note that each case itself breaks down in several subcases and steps. Since no valid outcome remains for $R^+$, we conclude that our basic assumptions are in conflict, so no SWF satisfies anonymity, unanimity, and strategyproofness if $m=5$ and $n=2$. 

\begin{lemma}
    $f(R^+)\neq ecabd$.
\end{lemma}
\begin{proof}
To prove her lemma, we assume for contradiction that $f(R^+)=ecabd$. We then proceed in five steps to specify the outcomes for further profiles, which ultimately results in a contradiction. We note that, except for $R^*$ and $R^+$, we will reset the profile markers for each step as the corresponding derivations are fully self-contained. 
\medskip

\textbf{Step 1:} Our first goal is to show that $f(R)=aecbd$ for the profile $R$ where one voter reports $aebcd$ and the other $aecbd$. We hence assume for contradiction that $f(R)\neq aecbd$. The subsequent derivation shows that this assumption is invalid as no feasible outcome remains for profile $R^{28}$.  In this table, the profile $R$ appears in Steps 1 and 6 (i.e., $R=R^1$ and $R=R^6$).
{\footnotesize
{\setlength{\tabcolsep}{3pt}
\noindent\begin{longtable}{c|cc|llllllll}
$R^{*}$ & $abcde$ & $acbed$ & $abcde$ (A)\\
$R^{+}$ & $eabcd$ & $ecabd$ & $ecabd$ (A)\\
$R^{1}$ & $aebdc$ & $aecbd$ & $aebcd$ & $aebdc$ & \textcolor{gray}{$aecbd\,(A)$} \\
$R^{2}$ & $eabdc$ & $ecabd$ & \textcolor{gray}{$eabcd\,(R^{+})$} & \textcolor{gray}{$eabdc\,(R^{+})$} & \textcolor{gray}{$eacbd\,(R^{+})$} & $ecabd$ \\
$R^{3}$ & $aebdc$ & $ecabd$ & \textcolor{gray}{$aebcd\,(R^{+})$} & \textcolor{gray}{$aebdc\,(R^{2})$} & \textcolor{gray}{$aecbd\,(R^{1})$} & \textcolor{gray}{$eabcd\,(R^{+})$} & \textcolor{gray}{$eabdc\,(R^{+})$} & \textcolor{gray}{$eacbd\,(R^{+})$} & $ecabd$ \\
$R^{4}$ & $aebdc$ & $caebd$ & $acebd$ & \textcolor{gray}{$aebcd\,(R^{3})$} & \textcolor{gray}{$aebdc\,(R^{3})$} & \textcolor{gray}{$aecbd\,(R^{1})$} & $caebd$ \\
$R^{5}$ & $acebd$ & $aebdc$ & $acebd$ & \textcolor{gray}{$aebcd\,(R^{4})$} & \textcolor{gray}{$aebdc\,(R^{4})$} & \textcolor{gray}{$aecbd\,(R^{1})$} \\
$R^{6}$ & $aebdc$ & $aecbd$ & $aebcd$ & \textcolor{gray}{$aebdc\,(R^{5})$} & \textcolor{gray}{$aecbd\,(A)$} \\
$R^{7}$ & $abecd$ & $aecbd$ & $abecd$ & $aebcd$ & \textcolor{gray}{$aecbd\,(R^{6})$} \\
$R^{8}$ & $abecd$ & $ecabd$ & \textcolor{gray}{$abecd\,(R^{3})$} & \textcolor{gray}{$aebcd\,(R^{+})$} & \textcolor{gray}{$aecbd\,(R^{3})$} & \textcolor{gray}{$eabcd\,(R^{+})$} & \textcolor{gray}{$eacbd\,(R^{+})$} & $ecabd$ \\
$R^{9}$ & $abecd$ & $eacbd$ & \textcolor{gray}{$abecd\,(R^{8})$} & \textcolor{gray}{$aebcd\,(R^{8})$} & \textcolor{gray}{$aecbd\,(R^{7})$} & $eabcd$ & $eacbd$ \\
$R^{10}$ & $abecd$ & $eabcd$ & \textcolor{gray}{$abecd\,(R^{9})$} & $aebcd$ & $eabcd$ \\
$R^{11}$ & $eabdc$ & $eacbd$ & $eabcd$ & \textcolor{gray}{$eabdc\,(R^{2})$} & $eacbd$ \\
$R^{12}$ & $abedc$ & $aecbd$ & $abecd$ & $abedc$ & $aebcd$ & \textcolor{gray}{$aebdc\,(R^{6})$} & \textcolor{gray}{$aecbd\,(R^{1})$} \\
$R^{13}$ & $abcde$ & $acebd$ & $abcde$ & \textcolor{gray}{$abced\,(R^{*})$} & \textcolor{gray}{$acbde\,(R^{*})$} & \textcolor{gray}{$acbed\,(R^{*})$} & \textcolor{gray}{$acebd\,(R^{*})$} \\
$R^{14}$ & $abcde$ & $aecbd$ & $abcde$ & \textcolor{gray}{$abced\,(R^{*})$} & $abecd$ & \textcolor{gray}{$acbde\,(R^{*})$} & \textcolor{gray}{$acbed\,(R^{*})$} & \textcolor{gray}{$acebd\,(R^{*})$} & \textcolor{gray}{$aebcd\,(R^{13})$} & \textcolor{gray}{$aecbd\,(R^{13})$} \\
$R^{15}$ & $abdec$ & $aecbd$ & $abdec$ & $abecd$ & $abedc$ & \textcolor{gray}{$aebcd\,(R^{14})$} & \textcolor{gray}{$aebdc\,(R^{6})$} & \textcolor{gray}{$aecbd\,(R^{1})$} \\
$R^{16}$ & $abedc$ & $aecbd$ & $abecd$ & $abedc$ & \textcolor{gray}{$aebcd\,(R^{15})$} & \textcolor{gray}{$aebdc\,(R^{6})$} & \textcolor{gray}{$aecbd\,(R^{1})$} \\
$R^{17}$ & $abedc$ & $eacbd$ & \textcolor{gray}{$abecd\,(R^{9})$} & \textcolor{gray}{$abedc\,(R^{9})$} & \textcolor{gray}{$aebcd\,(R^{9})$} & $aebdc$ & \textcolor{gray}{$aecbd\,(R^{12})$} & $eabcd$ & \textcolor{gray}{$eabdc\,(R^{11})$} & \textcolor{gray}{$eacbd\,(R^{16})$} \\
$R^{18}$ & $abecd$ & $eacbd$ & \textcolor{gray}{$abecd\,(R^{8})$} & \textcolor{gray}{$aebcd\,(R^{8})$} & \textcolor{gray}{$aecbd\,(R^{7})$} & $eabcd$ & \textcolor{gray}{$eacbd\,(R^{17})$} \\
$R^{19}$ & $abecd$ & $eabcd$ & \textcolor{gray}{$abecd\,(R^{9})$} & \textcolor{gray}{$aebcd\,(R^{18})$} & $eabcd$ \\
$R^{20}$ & $abedc$ & $aecbd$ & $abecd$ & \textcolor{gray}{$abedc\,(R^{17})$} & \textcolor{gray}{$aebcd\,(R^{15})$} & \textcolor{gray}{$aebdc\,(R^{6})$} & \textcolor{gray}{$aecbd\,(R^{1})$} \\
$R^{21}$ & $abecd$ & $aecbd$ & $abecd$ & \textcolor{gray}{$aebcd\,(R^{20})$} & \textcolor{gray}{$aecbd\,(R^{6})$} \\
$R^{22}$ & $abced$ & $eabcd$ & \textcolor{gray}{$abced\,(R^{19})$} & \textcolor{gray}{$abecd\,(R^{10})$} & \textcolor{gray}{$aebcd\,(R^{19})$} & $eabcd$ \\
$R^{23}$ & $abecd$ & $aecdb$ & $abecd$ & \textcolor{gray}{$aebcd\,(R^{21})$} & \textcolor{gray}{$aecbd\,(R^{7})$} & \textcolor{gray}{$aecdb\,(R^{21})$} \\
$R^{24}$ & $acebd$ & $eabcd$ & \textcolor{gray}{$acebd\,(R^{22})$} & \textcolor{gray}{$aebcd\,(R^{5})$} & \textcolor{gray}{$aecbd\,(R^{5})$} & \textcolor{gray}{$eabcd\,(R^{5})$} & $eacbd$ \\
$R^{25}$ & $abecd$ & $eacdb$ & \textcolor{gray}{$abecd\,(R^{8})$} & \textcolor{gray}{$aebcd\,(R^{8})$} & \textcolor{gray}{$aecbd\,(R^{7})$} & \textcolor{gray}{$aecdb\,(R^{21})$} & $eabcd$ & \textcolor{gray}{$eacbd\,(R^{18})$} & \textcolor{gray}{$eacdb\,(R^{23})$} \\
$R^{26}$ & $eabcd$ & $eacdb$ & \textcolor{gray}{$eabcd\,(R^{24})$} & \textcolor{gray}{$eacbd\,(R^{25})$} & \textcolor{gray}{$eacdb\,(R^{25})$} \\
\end{longtable}
}
}

\textbf{Step 2:} Next, we will show that $f(R)=aedcb$ for the profile $R$ where one voters reports $aecbd$ and the other $aedcb$. We hence assume that $f(R)\neq aedcb$ and infer a contradiction as shown below. The profile $R$ corresponds to $R^2$ and $R^{13}$ in the subsequent derivation.
{\footnotesize
{\setlength{\tabcolsep}{3pt}
\noindent\begin{longtable}{c|cc|llllllll}
$R^{*}$ & $abcde$ & $acbed$ & $abcde$ (A)\\
$R^{1}$ & $aebdc$ & $aecbd$ & \multicolumn{2}{l}{$aecbd$ (A, Step 1)}\\
$R^{2}$ & $aecbd$ & $aedcb$ & $aecbd$ & $aecdb$ & \textcolor{gray}{$aedcb\,(A)$} \\
$R^{3}$ & $abecd$ & $aecbd$ & $abecd$ & \textcolor{gray}{$aebcd\,(R^{1})$} & $aecbd$ \\
$R^{4}$ & $aebcd$ & $aecbd$ & \textcolor{gray}{$aebcd\,(R^{1})$} & $aecbd$ \\
$R^{5}$ & $abcde$ & $acebd$ & $abcde$ & \textcolor{gray}{$abced\,(R^{*})$} & \textcolor{gray}{$acbde\,(R^{*})$} & \textcolor{gray}{$acbed\,(R^{*})$} & \textcolor{gray}{$acebd\,(R^{*})$} \\
$R^{6}$ & $abcde$ & $aecbd$ & $abcde$ & \textcolor{gray}{$abced\,(R^{*})$} & $abecd$ & \textcolor{gray}{$acbde\,(R^{*})$} & \textcolor{gray}{$acbed\,(R^{*})$} & \textcolor{gray}{$acebd\,(R^{*})$} & \textcolor{gray}{$aebcd\,(R^{1})$} & \textcolor{gray}{$aecbd\,(R^{5})$} \\
$R^{7}$ & $abdec$ & $aecbd$ & $abdec$ & $abecd$ & \textcolor{gray}{$abedc\,(R^{1})$} & \textcolor{gray}{$aebcd\,(R^{1})$} & \textcolor{gray}{$aebdc\,(R^{1})$} & \textcolor{gray}{$aecbd\,(R^{6})$} \\
$R^{8}$ & $abedc$ & $aecbd$ & $abecd$ & \textcolor{gray}{$abedc\,(R^{1})$} & \textcolor{gray}{$aebcd\,(R^{1})$} & \textcolor{gray}{$aebdc\,(R^{1})$} & \textcolor{gray}{$aecbd\,(R^{7})$} \\
$R^{9}$ & $abecd$ & $aecbd$ & $abecd$ & \textcolor{gray}{$aebcd\,(R^{1})$} & \textcolor{gray}{$aecbd\,(R^{8})$} \\
$R^{10}$ & $abecd$ & $aecdb$ & $abecd$ & \textcolor{gray}{$aebcd\,(R^{3})$} & \textcolor{gray}{$aecbd\,(R^{9})$} & \textcolor{gray}{$aecdb\,(R^{9})$} \\
$R^{11}$ & $aebcd$ & $aecdb$ & \textcolor{gray}{$aebcd\,(R^{4})$} & $aecbd$ & \textcolor{gray}{$aecdb\,(R^{10})$} \\
$R^{12}$ & $aecbd$ & $aecdb$ & $aecbd$ & \textcolor{gray}{$aecdb\,(R^{11})$} \\
$R^{13}$ & $aecbd$ & $aedcb$ & $aecbd$ & \textcolor{gray}{$aecdb\,(R^{12})$} & \textcolor{gray}{$aedcb\,(A)$} \\
$R^{14}$ & $adecb$ & $aecbd$ & \textcolor{gray}{$adecb\,(R^{13})$} & $aecbd$ & \textcolor{gray}{$aecdb\,(R^{12})$} & \textcolor{gray}{$aedcb\,(R^{2})$} \\
$R^{15}$ & $adebc$ & $aecbd$ & \textcolor{gray}{$adebc\,(R^{14})$} & \textcolor{gray}{$adecb\,(R^{13})$} & \textcolor{gray}{$aebcd\,(R^{1})$} & \textcolor{gray}{$aebdc\,(R^{1})$} & $aecbd$ & \textcolor{gray}{$aecdb\,(R^{12})$} & \textcolor{gray}{$aedbc\,(R^{1})$} & \textcolor{gray}{$aedcb\,(R^{2})$} \\
$R^{16}$ & $adebc$ & $aecdb$ & \textcolor{gray}{$adebc\,(R^{15})$} & \textcolor{gray}{$adecb\,(R^{15})$} & $aecdb$ & \textcolor{gray}{$aedbc\,(R^{15})$} & $aedcb$ \\
$R^{17}$ & $adbec$ & $aecdb$ & \textcolor{gray}{$adbec\,(R^{16})$} & \textcolor{gray}{$adebc\,(R^{16})$} & \textcolor{gray}{$adecb\,(R^{16})$} & \textcolor{gray}{$aecdb\,(R^{10})$} & \textcolor{gray}{$aedbc\,(R^{16})$} & $aedcb$ \\
$R^{18}$ & $adbec$ & $aedcb$ & \textcolor{gray}{$adbec\,(R^{17})$} & \textcolor{gray}{$adebc\,(R^{17})$} & \textcolor{gray}{$adecb\,(R^{17})$} & \textcolor{gray}{$aedbc\,(R^{17})$} & $aedcb$ \\
$R^{19}$ & $aebdc$ & $aedcb$ & $aebdc$ & \textcolor{gray}{$aedbc\,(R^{18})$} & $aedcb$ \\
$R^{20}$ & $adceb$ & $aedcb$ & $adceb$ & \textcolor{gray}{$adecb\,(R^{18})$} & $aedcb$ \\
$R^{21}$ & $adebc$ & $aebcd$ & \textcolor{gray}{$adebc\,(R^{15})$} & $aebcd$ & $aebdc$ & \textcolor{gray}{$aedbc\,(R^{15})$} \\
$R^{22}$ & $adebc$ & $aebdc$ & \textcolor{gray}{$adebc\,(R^{21})$} & $aebdc$ & $aedbc$ \\
$R^{23}$ & $abecd$ & $aedcb$ & $abecd$ & \textcolor{gray}{$abedc\,(R^{18})$} & \textcolor{gray}{$aebcd\,(R^{3})$} & $aebdc$ & \textcolor{gray}{$aecbd\,(R^{9})$} & \textcolor{gray}{$aecdb\,(R^{9})$} & \textcolor{gray}{$aedbc\,(R^{10})$} & \textcolor{gray}{$aedcb\,(R^{2})$} \\
$R^{24}$ & $abedc$ & $aedcb$ & \textcolor{gray}{$abedc\,(R^{18})$} & $aebdc$ & \textcolor{gray}{$aedbc\,(R^{18})$} & \textcolor{gray}{$aedcb\,(R^{23})$} \\
$R^{25}$ & $aebdc$ & $aedcb$ & $aebdc$ & \textcolor{gray}{$aedbc\,(R^{18})$} & \textcolor{gray}{$aedcb\,(R^{24})$} \\
$R^{26}$ & $adecb$ & $aebdc$ & \textcolor{gray}{$adebc\,(R^{22})$} & \textcolor{gray}{$adecb\,(R^{14})$} & $aebdc$ & \textcolor{gray}{$aedbc\,(R^{19})$} & \textcolor{gray}{$aedcb\,(R^{25})$} \\
$R^{27}$ & $adceb$ & $aebdc$ & \textcolor{gray}{$adceb\,(R^{26})$} & \textcolor{gray}{$adebc\,(R^{22})$} & \textcolor{gray}{$adecb\,(R^{20})$} & $aebdc$ & \textcolor{gray}{$aedbc\,(R^{19})$} & \textcolor{gray}{$aedcb\,(R^{25})$} \\
$R^{28}$ & $acedb$ & $aebdc$ & \textcolor{gray}{$acebd\,(R^{27})$} & \textcolor{gray}{$acedb\,(R^{27})$} & \textcolor{gray}{$aebcd\,(R^{1})$} & \textcolor{gray}{$aebdc\,(R^{1})$} & $aecbd$ & \textcolor{gray}{$aecdb\,(R^{25})$} & \textcolor{gray}{$aedbc\,(R^{1})$} & \textcolor{gray}{$aedcb\,(R^{25})$} \\
$R^{29}$ & $acedb$ & $aecbd$ & \textcolor{gray}{$acebd\,(R^{28})$} & \textcolor{gray}{$acedb\,(R^{28})$} & $aecbd$ & \textcolor{gray}{$aecdb\,(R^{12})$} \\
$R^{30}$ & $adceb$ & $aecbd$ & \textcolor{gray}{$acdeb\,(R^{14})$} & \textcolor{gray}{$acebd\,(R^{29})$} & \textcolor{gray}{$acedb\,(R^{29})$} & \textcolor{gray}{$adceb\,(R^{14})$} & \textcolor{gray}{$adecb\,(R^{13})$} & $aecbd$ & \textcolor{gray}{$aecdb\,(R^{12})$} & \textcolor{gray}{$aedcb\,(R^{2})$} \\
$R^{31}$ & $acebd$ & $adceb$ & \textcolor{gray}{$acdeb\,(R^{30})$} & $acebd$ & $acedb$ & \textcolor{gray}{$adceb\,(R^{30})$} \\
$R^{32}$ & $acedb$ & $adceb$ & $acdeb$ & $acedb$ & \textcolor{gray}{$adceb\,(R^{31})$} \\
$R^{33}$ & $adbce$ & $aedcb$ & \textcolor{gray}{$adbce\,(R^{18})$} & \textcolor{gray}{$adbec\,(R^{18})$} & \textcolor{gray}{$adcbe\,(R^{18})$} & $adceb$ & \textcolor{gray}{$adebc\,(R^{18})$} & \textcolor{gray}{$adecb\,(R^{18})$} & \textcolor{gray}{$aedbc\,(R^{18})$} & $aedcb$ \\
$R^{34}$ & $abcde$ & $acedb$ & $abcde$ & \textcolor{gray}{$abced\,(R^{*})$} & \textcolor{gray}{$acbde\,(R^{*})$} & \textcolor{gray}{$acbed\,(R^{*})$} & $acdbe$ & \textcolor{gray}{$acdeb\,(R^{5})$} & \textcolor{gray}{$acebd\,(R^{*})$} & \textcolor{gray}{$acedb\,(R^{5})$} \\
$R^{35}$ & $acedb$ & $adbce$ & $acdbe$ & \textcolor{gray}{$acdeb\,(R^{34})$} & \textcolor{gray}{$acedb\,(R^{34})$} & \textcolor{gray}{$adbce\,(R^{33})$} & \textcolor{gray}{$adcbe\,(R^{33})$} & \textcolor{gray}{$adceb\,(R^{32})$} \\
$R^{36}$ & $acdbe$ & $acedb$ & $acdbe$ & \textcolor{gray}{$acdeb\,(R^{35})$} & \textcolor{gray}{$acedb\,(R^{35})$} \\
$R^{37}$ & $acdbe$ & $aecdb$ & $acdbe$ & \textcolor{gray}{$acdeb\,(R^{36})$} & \textcolor{gray}{$acedb\,(R^{36})$} & \textcolor{gray}{$aecdb\,(R^{36})$} \\
$R^{38}$ & $acdbe$ & $aedcb$ & $acdbe$ & \textcolor{gray}{$acdeb\,(R^{36})$} & \textcolor{gray}{$acedb\,(R^{36})$} & \textcolor{gray}{$adcbe\,(R^{18})$} & $adceb$ & \textcolor{gray}{$adecb\,(R^{18})$} & \textcolor{gray}{$aecdb\,(R^{36})$} & \textcolor{gray}{$aedcb\,(R^{37})$} \\
$R^{39}$ & $adbce$ & $aedcb$ & \textcolor{gray}{$adbce\,(R^{18})$} & \textcolor{gray}{$adbec\,(R^{18})$} & \textcolor{gray}{$adcbe\,(R^{18})$} & $adceb$ & \textcolor{gray}{$adebc\,(R^{18})$} & \textcolor{gray}{$adecb\,(R^{18})$} & \textcolor{gray}{$aedbc\,(R^{18})$} & \textcolor{gray}{$aedcb\,(R^{38})$} \\
$R^{40}$ & $adceb$ & $aedcb$ & $adceb$ & \textcolor{gray}{$adecb\,(R^{18})$} & \textcolor{gray}{$aedcb\,(R^{39})$} \\
$R^{41}$ & $acdbe$ & $adbce$ & $acdbe$ & \textcolor{gray}{$adbce\,(R^{35})$} & \textcolor{gray}{$adcbe\,(R^{35})$} \\
$R^{42}$ & $adceb$ & $aedbc$ & \textcolor{gray}{$adceb\,(R^{27})$} & $adebc$ & \textcolor{gray}{$adecb\,(R^{20})$} & \textcolor{gray}{$aedbc\,(R^{40})$} & \textcolor{gray}{$aedcb\,(R^{40})$} \\
$R^{43}$ & $adceb$ & $adebc$ & \textcolor{gray}{$adceb\,(R^{42})$} & $adebc$ & \textcolor{gray}{$adecb\,(R^{40})$} \\
$R^{44}$ & $acdbe$ & $adbec$ & $acdbe$ & \textcolor{gray}{$adbce\,(R^{41})$} & \textcolor{gray}{$adbec\,(R^{41})$} & \textcolor{gray}{$adcbe\,(R^{41})$} \\
$R^{45}$ & $acdeb$ & $adebc$ & \textcolor{gray}{$acdeb\,(R^{43})$} & \textcolor{gray}{$adceb\,(R^{43})$} & $adebc$ & \textcolor{gray}{$adecb\,(R^{43})$} \\
$R^{46}$ & $acdeb$ & $adbec$ & \textcolor{gray}{$acdbe\,(R^{45})$} & \textcolor{gray}{$acdeb\,(R^{45})$} & \textcolor{gray}{$adbce\,(R^{44})$} & \textcolor{gray}{$adbec\,(R^{17})$} & \textcolor{gray}{$adcbe\,(R^{44})$} & \textcolor{gray}{$adceb\,(R^{45})$} & \textcolor{gray}{$adebc\,(R^{44})$} & \textcolor{gray}{$adecb\,(R^{44})$} \\
\end{longtable}
}
}

\textbf{Step 3:} As our third step, we will show that $f(R)=ecabd$ for the profile $R$ where one voter reports $abecd$ and the other reports $ecabd$. As usual, we assume that $f(R)$ is not our desired outcome, i.e., $f(R)\neq ecabd$, and derive a contradiction. We use our assumption on $R$ at profile $R^3$ and $R^{19}$.
{\footnotesize
{\setlength{\tabcolsep}{3pt}
\noindent\begin{longtable}{c|cc|llllllll}
$R^{*}$ & $abcde$ & $acbed$ & $abcde$ (A)\\
$R^{+}$ & $eabcd$ & $ecabd$ & $ecabd$ (A)\\
$R^{1}$ & $aebdc$ & $aecbd$ & \multicolumn{2}{l}{$aecbd$ (A, Step 1)}\\
$R^{2}$ & $aecbd$ & $aedcb$ & \multicolumn{2}{l}{$aedcb$ (A, Step 2)}\\
$R^{3}$ & $abecd$ & $ecabd$ & $abecd$ & \textcolor{gray}{$aebcd\,(R^{+})$} & $aecbd$ & \textcolor{gray}{$eabcd\,(R^{+})$} & \textcolor{gray}{$eacbd\,(R^{+})$} & \textcolor{gray}{$ecabd\,(A)$} \\
$R^{4}$ & $eacbd$ & $ecabd$ & \textcolor{gray}{$eacbd\,(R^{+})$} & $ecabd$ \\
$R^{5}$ & $aecbd$ & $ecabd$ & $aecbd$ & \textcolor{gray}{$eacbd\,(R^{+})$} & $ecabd$ \\
$R^{6}$ & $aebdc$ & $ecabd$ & \textcolor{gray}{$aebcd\,(R^{1})$} & \textcolor{gray}{$aebdc\,(R^{1})$} & $aecbd$ & \textcolor{gray}{$eabcd\,(R^{+})$} & \textcolor{gray}{$eabdc\,(R^{1})$} & \textcolor{gray}{$eacbd\,(R^{+})$} & \textcolor{gray}{$ecabd\,(R^{3})$} \\
$R^{7}$ & $aecbd$ & $ecabd$ & $aecbd$ & \textcolor{gray}{$eacbd\,(R^{+})$} & \textcolor{gray}{$ecabd\,(R^{6})$} \\
$R^{8}$ & $aecbd$ & $ecbad$ & $aecbd$ & \textcolor{gray}{$eacbd\,(R^{5})$} & \textcolor{gray}{$ecabd\,(R^{7})$} & \textcolor{gray}{$ecbad\,(R^{7})$} \\
$R^{9}$ & $eacbd$ & $ecbad$ & \textcolor{gray}{$eacbd\,(R^{4})$} & $ecabd$ & \textcolor{gray}{$ecbad\,(R^{8})$} \\
$R^{10}$ & $ecabd$ & $ecbad$ & $ecabd$ & \textcolor{gray}{$ecbad\,(R^{9})$} \\
$R^{11}$ & $aebcd$ & $aecbd$ & \textcolor{gray}{$aebcd\,(R^{1})$} & $aecbd$ \\
$R^{12}$ & $abecd$ & $aecbd$ & $abecd$ & \textcolor{gray}{$aebcd\,(R^{1})$} & $aecbd$ \\
$R^{13}$ & $abcde$ & $acebd$ & $abcde$ & \textcolor{gray}{$abced\,(R^{*})$} & \textcolor{gray}{$acbde\,(R^{*})$} & \textcolor{gray}{$acbed\,(R^{*})$} & \textcolor{gray}{$acebd\,(R^{*})$} \\
$R^{14}$ & $abcde$ & $aecbd$ & $abcde$ & \textcolor{gray}{$abced\,(R^{*})$} & $abecd$ & \textcolor{gray}{$acbde\,(R^{*})$} & \textcolor{gray}{$acbed\,(R^{*})$} & \textcolor{gray}{$acebd\,(R^{*})$} & \textcolor{gray}{$aebcd\,(R^{1})$} & \textcolor{gray}{$aecbd\,(R^{13})$} \\
$R^{15}$ & $abdec$ & $aecbd$ & $abdec$ & $abecd$ & \textcolor{gray}{$abedc\,(R^{1})$} & \textcolor{gray}{$aebcd\,(R^{1})$} & \textcolor{gray}{$aebdc\,(R^{1})$} & \textcolor{gray}{$aecbd\,(R^{14})$} \\
$R^{16}$ & $abedc$ & $aecbd$ & $abecd$ & \textcolor{gray}{$abedc\,(R^{1})$} & \textcolor{gray}{$aebcd\,(R^{1})$} & \textcolor{gray}{$aebdc\,(R^{1})$} & \textcolor{gray}{$aecbd\,(R^{15})$} \\
$R^{17}$ & $abecd$ & $aecbd$ & $abecd$ & \textcolor{gray}{$aebcd\,(R^{1})$} & \textcolor{gray}{$aecbd\,(R^{16})$} \\
$R^{18}$ & $eabdc$ & $ecabd$ & \textcolor{gray}{$eabcd\,(R^{+})$} & \textcolor{gray}{$eabdc\,(R^{+})$} & \textcolor{gray}{$eacbd\,(R^{+})$} & $ecabd$ \\
$R^{19}$ & $abecd$ & $ecabd$ & $abecd$ & \textcolor{gray}{$aebcd\,(R^{+})$} & \textcolor{gray}{$aecbd\,(R^{17})$} & \textcolor{gray}{$eabcd\,(R^{+})$} & \textcolor{gray}{$eacbd\,(R^{+})$} & \textcolor{gray}{$ecabd\,(A)$} \\
$R^{20}$ & $beacd$ & $ecabd$ & $beacd$ & $becad$ & \textcolor{gray}{$eabcd\,(R^{+})$} & \textcolor{gray}{$eacbd\,(R^{+})$} & \textcolor{gray}{$ebacd\,(R^{+})$} & $ebcad$ & \textcolor{gray}{$ecabd\,(R^{19})$} & \textcolor{gray}{$ecbad\,(R^{10})$} \\
$R^{21}$ & $ebadc$ & $ecabd$ & \textcolor{gray}{$eabcd\,(R^{+})$} & \textcolor{gray}{$eabdc\,(R^{+})$} & \textcolor{gray}{$eacbd\,(R^{+})$} & \textcolor{gray}{$ebacd\,(R^{+})$} & \textcolor{gray}{$ebadc\,(R^{18})$} & $ebcad$ & \textcolor{gray}{$ecabd\,(R^{20})$} & \textcolor{gray}{$ecbad\,(R^{10})$} \\
$R^{22}$ & $ebcad$ & $ecabd$ & $ebcad$ & \textcolor{gray}{$ecabd\,(R^{21})$} & \textcolor{gray}{$ecbad\,(R^{10})$} \\
$R^{23}$ & $ebadc$ & $ebcad$ & \textcolor{gray}{$ebacd\,(R^{21})$} & \textcolor{gray}{$ebadc\,(R^{21})$} & $ebcad$ \\
$R^{24}$ & $aecbd$ & $ebcad$ & \textcolor{gray}{$aebcd\,(R^{1})$} & $aecbd$ & $eabcd$ & \textcolor{gray}{$eacbd\,(R^{5})$} & \textcolor{gray}{$ebacd\,(R^{8})$} & \textcolor{gray}{$ebcad\,(R^{8})$} & \textcolor{gray}{$ecabd\,(R^{7})$} & \textcolor{gray}{$ecbad\,(R^{7})$} \\
$R^{25}$ & $eacdb$ & $ebcad$ & $eabcd$ & \textcolor{gray}{$eacbd\,(R^{22})$} & $eacdb$ & \textcolor{gray}{$ebacd\,(R^{24})$} & \textcolor{gray}{$ebcad\,(R^{24})$} & \textcolor{gray}{$ecabd\,(R^{22})$} & \textcolor{gray}{$ecadb\,(R^{22})$} & \textcolor{gray}{$ecbad\,(R^{24})$} \\
$R^{26}$ & $eadbc$ & $ebcad$ & $eabcd$ & \textcolor{gray}{$eabdc\,(R^{23})$} & $eadbc$ & \textcolor{gray}{$ebacd\,(R^{23})$} & \textcolor{gray}{$ebadc\,(R^{23})$} & \textcolor{gray}{$ebcad\,(R^{25})$} \\
$R^{27}$ & $eabcd$ & $eadbc$ & $eabcd$ & \textcolor{gray}{$eabdc\,(R^{26})$} & $eadbc$ \\
$R^{28}$ & $acedb$ & $aedcb$ & $acedb$ & \textcolor{gray}{$aecdb\,(R^{2})$} & $aedcb$ \\
$R^{29}$ & $abecd$ & $aecdb$ & $abecd$ & \textcolor{gray}{$aebcd\,(R^{12})$} & \textcolor{gray}{$aecbd\,(R^{17})$} & \textcolor{gray}{$aecdb\,(R^{17})$} \\
$R^{30}$ & $aebcd$ & $aecdb$ & \textcolor{gray}{$aebcd\,(R^{11})$} & $aecbd$ & \textcolor{gray}{$aecdb\,(R^{29})$} \\
$R^{31}$ & $aecbd$ & $aecdb$ & $aecbd$ & \textcolor{gray}{$aecdb\,(R^{30})$} \\
$R^{32}$ & $abecd$ & $aedcb$ & $abecd$ & $abedc$ & \textcolor{gray}{$aebcd\,(R^{2})$} & $aebdc$ & \textcolor{gray}{$aecbd\,(R^{2})$} & \textcolor{gray}{$aecdb\,(R^{2})$} & \textcolor{gray}{$aedbc\,(R^{29})$} & \textcolor{gray}{$aedcb\,(R^{29})$} \\
$R^{33}$ & $abedc$ & $aedcb$ & $abedc$ & $aebdc$ & \textcolor{gray}{$aedbc\,(R^{32})$} & \textcolor{gray}{$aedcb\,(R^{32})$} \\
$R^{34}$ & $aebdc$ & $aedcb$ & $aebdc$ & $aedbc$ & \textcolor{gray}{$aedcb\,(R^{33})$} \\
$R^{35}$ & $abecd$ & $ebcad$ & $abecd$ & \textcolor{gray}{$aebcd\,(R^{17})$} & $baecd$ & $beacd$ & \textcolor{gray}{$becad\,(R^{24})$} & \textcolor{gray}{$eabcd\,(R^{19})$} & \textcolor{gray}{$ebacd\,(R^{24})$} & \textcolor{gray}{$ebcad\,(R^{24})$} \\
$R^{36}$ & $baecd$ & $ebcad$ & $baecd$ & $beacd$ & \textcolor{gray}{$becad\,(R^{35})$} & \textcolor{gray}{$ebacd\,(R^{23})$} & \textcolor{gray}{$ebcad\,(R^{35})$} \\
$R^{37}$ & $beacd$ & $ebcad$ & $beacd$ & $becad$ & \textcolor{gray}{$ebacd\,(R^{23})$} & \textcolor{gray}{$ebcad\,(R^{36})$} \\
$R^{38}$ & $abecd$ & $eacdb$ & $abecd$ & \textcolor{gray}{$aebcd\,(R^{12})$} & \textcolor{gray}{$aecbd\,(R^{17})$} & \textcolor{gray}{$aecdb\,(R^{17})$} & \textcolor{gray}{$eabcd\,(R^{19})$} & \textcolor{gray}{$eacbd\,(R^{3})$} & \textcolor{gray}{$eacdb\,(R^{19})$} \\
$R^{39}$ & $beacd$ & $eacdb$ & $beacd$ & $eabcd$ & \textcolor{gray}{$eacbd\,(R^{38})$} & \textcolor{gray}{$eacdb\,(R^{38})$} & \textcolor{gray}{$ebacd\,(R^{25})$} \\
$R^{40}$ & $beacd$ & $ecabd$ & $beacd$ & $becad$ & \textcolor{gray}{$eabcd\,(R^{+})$} & \textcolor{gray}{$eacbd\,(R^{+})$} & \textcolor{gray}{$ebacd\,(R^{+})$} & \textcolor{gray}{$ebcad\,(R^{37})$} & \textcolor{gray}{$ecabd\,(R^{19})$} & \textcolor{gray}{$ecbad\,(R^{10})$} \\
$R^{41}$ & $beacd$ & $eacdb$ & $beacd$ & \textcolor{gray}{$eabcd\,(R^{40})$} & \textcolor{gray}{$eacbd\,(R^{38})$} & \textcolor{gray}{$eacdb\,(R^{38})$} & \textcolor{gray}{$ebacd\,(R^{25})$} \\
$R^{42}$ & $beacd$ & $eadbc$ & $beacd$ & $beadc$ & \textcolor{gray}{$eabcd\,(R^{40})$} & \textcolor{gray}{$eabdc\,(R^{26})$} & \textcolor{gray}{$eadbc\,(R^{41})$} & \textcolor{gray}{$ebacd\,(R^{39})$} & $ebadc$ \\
$R^{43}$ & $eadbc$ & $ebcad$ & $eabcd$ & \textcolor{gray}{$eabdc\,(R^{23})$} & \textcolor{gray}{$eadbc\,(R^{42})$} & \textcolor{gray}{$ebacd\,(R^{23})$} & \textcolor{gray}{$ebadc\,(R^{23})$} & \textcolor{gray}{$ebcad\,(R^{25})$} \\
$R^{44}$ & $eabcd$ & $eadbc$ & $eabcd$ & \textcolor{gray}{$eabdc\,(R^{26})$} & \textcolor{gray}{$eadbc\,(R^{43})$} \\
$R^{45}$ & $eabcd$ & $edabc$ & $eabcd$ & \textcolor{gray}{$eabdc\,(R^{27})$} & \textcolor{gray}{$eadbc\,(R^{44})$} & \textcolor{gray}{$edabc\,(R^{44})$} \\
$R^{46}$ & $deabc$ & $eabcd$ & \textcolor{gray}{$deabc\,(R^{45})$} & $eabcd$ & \textcolor{gray}{$eabdc\,(R^{27})$} & \textcolor{gray}{$eadbc\,(R^{44})$} & \textcolor{gray}{$edabc\,(R^{44})$} \\
$R^{47}$ & $adebc$ & $aecbd$ & $adebc$ & $adecb$ & \textcolor{gray}{$aebcd\,(R^{2})$} & \textcolor{gray}{$aebdc\,(R^{1})$} & \textcolor{gray}{$aecbd\,(R^{2})$} & \textcolor{gray}{$aecdb\,(R^{2})$} & \textcolor{gray}{$aedbc\,(R^{1})$} & $aedcb$ \\
$R^{48}$ & $adebc$ & $eabcd$ & \textcolor{gray}{$adebc\,(R^{46})$} & \textcolor{gray}{$aebcd\,(R^{47})$} & $aebdc$ & \textcolor{gray}{$aedbc\,(R^{44})$} & $eabcd$ & \textcolor{gray}{$eabdc\,(R^{27})$} & \textcolor{gray}{$eadbc\,(R^{44})$} \\
$R^{49}$ & $adebc$ & $aebcd$ & \textcolor{gray}{$adebc\,(R^{48})$} & \textcolor{gray}{$aebcd\,(R^{47})$} & $aebdc$ & \textcolor{gray}{$aedbc\,(R^{48})$} \\
$R^{50}$ & $adebc$ & $aebdc$ & \textcolor{gray}{$adebc\,(R^{49})$} & $aebdc$ & \textcolor{gray}{$aedbc\,(R^{49})$} \\
$R^{51}$ & $adecb$ & $aebdc$ & \textcolor{gray}{$adebc\,(R^{50})$} & \textcolor{gray}{$adecb\,(R^{50})$} & $aebdc$ & \textcolor{gray}{$aedbc\,(R^{50})$} & \textcolor{gray}{$aedcb\,(R^{34})$} \\
$R^{52}$ & $adceb$ & $aebdc$ & \textcolor{gray}{$adceb\,(R^{51})$} & \textcolor{gray}{$adebc\,(R^{50})$} & \textcolor{gray}{$adecb\,(R^{50})$} & $aebdc$ & \textcolor{gray}{$aedbc\,(R^{50})$} & \textcolor{gray}{$aedcb\,(R^{34})$} \\
$R^{53}$ & $acedb$ & $aebdc$ & \textcolor{gray}{$acebd\,(R^{52})$} & \textcolor{gray}{$acedb\,(R^{52})$} & \textcolor{gray}{$aebcd\,(R^{1})$} & \textcolor{gray}{$aebdc\,(R^{1})$} & $aecbd$ & \textcolor{gray}{$aecdb\,(R^{28})$} & \textcolor{gray}{$aedbc\,(R^{1})$} & \textcolor{gray}{$aedcb\,(R^{34})$} \\
$R^{54}$ & $acedb$ & $aecbd$ & \textcolor{gray}{$acebd\,(R^{53})$} & \textcolor{gray}{$acedb\,(R^{53})$} & $aecbd$ & \textcolor{gray}{$aecdb\,(R^{31})$} \\
$R^{55}$ & $adebc$ & $aecbd$ & \textcolor{gray}{$adebc\,(R^{48})$} & \textcolor{gray}{$adecb\,(R^{48})$} & \textcolor{gray}{$aebcd\,(R^{2})$} & \textcolor{gray}{$aebdc\,(R^{1})$} & \textcolor{gray}{$aecbd\,(R^{2})$} & \textcolor{gray}{$aecdb\,(R^{2})$} & \textcolor{gray}{$aedbc\,(R^{1})$} & $aedcb$ \\
$R^{56}$ & $adebc$ & $aecdb$ & \textcolor{gray}{$adebc\,(R^{55})$} & \textcolor{gray}{$adecb\,(R^{55})$} & \textcolor{gray}{$aecdb\,(R^{47})$} & \textcolor{gray}{$aedbc\,(R^{55})$} & $aedcb$ \\
$R^{57}$ & $adbec$ & $aecdb$ & \textcolor{gray}{$adbec\,(R^{56})$} & \textcolor{gray}{$adebc\,(R^{56})$} & \textcolor{gray}{$adecb\,(R^{56})$} & \textcolor{gray}{$aecdb\,(R^{29})$} & \textcolor{gray}{$aedbc\,(R^{56})$} & $aedcb$ \\
$R^{58}$ & $adbec$ & $aedcb$ & \textcolor{gray}{$adbec\,(R^{57})$} & \textcolor{gray}{$adebc\,(R^{57})$} & \textcolor{gray}{$adecb\,(R^{57})$} & \textcolor{gray}{$aedbc\,(R^{57})$} & $aedcb$ \\
$R^{59}$ & $adceb$ & $aecbd$ & \textcolor{gray}{$acdeb\,(R^{52})$} & \textcolor{gray}{$acebd\,(R^{2})$} & \textcolor{gray}{$acedb\,(R^{54})$} & \textcolor{gray}{$adceb\,(R^{52})$} & \textcolor{gray}{$adecb\,(R^{55})$} & \textcolor{gray}{$aecbd\,(R^{2})$} & \textcolor{gray}{$aecdb\,(R^{2})$} & $aedcb$ \\
$R^{60}$ & $adceb$ & $aedcb$ & \textcolor{gray}{$adceb\,(R^{59})$} & \textcolor{gray}{$adecb\,(R^{58})$} & $aedcb$ \\
$R^{61}$ & $adbce$ & $aedcb$ & \textcolor{gray}{$adbce\,(R^{58})$} & \textcolor{gray}{$adbec\,(R^{58})$} & \textcolor{gray}{$adcbe\,(R^{58})$} & \textcolor{gray}{$adceb\,(R^{60})$} & \textcolor{gray}{$adebc\,(R^{58})$} & \textcolor{gray}{$adecb\,(R^{58})$} & \textcolor{gray}{$aedbc\,(R^{58})$} & $aedcb$ \\
$R^{62}$ & $acdbe$ & $aedcb$ & \textcolor{gray}{$acdbe\,(R^{61})$} & \textcolor{gray}{$acdeb\,(R^{60})$} & $acedb$ & \textcolor{gray}{$adcbe\,(R^{58})$} & \textcolor{gray}{$adceb\,(R^{60})$} & \textcolor{gray}{$adecb\,(R^{58})$} & \textcolor{gray}{$aecdb\,(R^{2})$} & $aedcb$ \\
$R^{63}$ & $acdbe$ & $aecdb$ & \textcolor{gray}{$acdbe\,(R^{62})$} & \textcolor{gray}{$acdeb\,(R^{62})$} & $acedb$ & $aecdb$ \\
$R^{64}$ & $acdbe$ & $acedb$ & \textcolor{gray}{$acdbe\,(R^{63})$} & $acdeb$ & $acedb$ \\
$R^{65}$ & $abcde$ & $acedb$ & $abcde$ & \textcolor{gray}{$abced\,(R^{*})$} & \textcolor{gray}{$acbde\,(R^{*})$} & \textcolor{gray}{$acbed\,(R^{*})$} & $acdbe$ & \textcolor{gray}{$acdeb\,(R^{13})$} & \textcolor{gray}{$acebd\,(R^{*})$} & \textcolor{gray}{$acedb\,(R^{13})$} \\
$R^{66}$ & $acedb$ & $adbce$ & \textcolor{gray}{$acdbe\,(R^{64})$} & \textcolor{gray}{$acdeb\,(R^{65})$} & \textcolor{gray}{$acedb\,(R^{65})$} & \textcolor{gray}{$adbce\,(R^{61})$} & \textcolor{gray}{$adcbe\,(R^{61})$} & $adceb$ \\
$R^{67}$ & $adbce$ & $adceb$ & \textcolor{gray}{$adbce\,(R^{66})$} & \textcolor{gray}{$adcbe\,(R^{66})$} & $adceb$ \\
$R^{68}$ & $acedb$ & $adceb$ & \textcolor{gray}{$acdeb\,(R^{66})$} & \textcolor{gray}{$acedb\,(R^{66})$} & $adceb$ \\
$R^{69}$ & $acbed$ & $aedcb$ & $acbed$ & \textcolor{gray}{$acebd\,(R^{2})$} & $acedb$ & \textcolor{gray}{$aecbd\,(R^{2})$} & \textcolor{gray}{$aecdb\,(R^{2})$} & $aedcb$ \\
$R^{70}$ & $acebd$ & $adceb$ & \textcolor{gray}{$acdeb\,(R^{68})$} & \textcolor{gray}{$acebd\,(R^{59})$} & \textcolor{gray}{$acedb\,(R^{68})$} & $adceb$ \\
$R^{71}$ & $abcde$ & $acedb$ & $abcde$ & \textcolor{gray}{$abced\,(R^{*})$} & \textcolor{gray}{$acbde\,(R^{*})$} & \textcolor{gray}{$acbed\,(R^{*})$} & \textcolor{gray}{$acdbe\,(R^{64})$} & \textcolor{gray}{$acdeb\,(R^{13})$} & \textcolor{gray}{$acebd\,(R^{*})$} & \textcolor{gray}{$acedb\,(R^{13})$} \\
$R^{72}$ & $abcde$ & $adceb$ & $abcde$ & \textcolor{gray}{$abdce\,(R^{67})$} & \textcolor{gray}{$acbde\,(R^{*})$} & \textcolor{gray}{$acdbe\,(R^{71})$} & \textcolor{gray}{$acdeb\,(R^{13})$} & \textcolor{gray}{$adbce\,(R^{67})$} & \textcolor{gray}{$adcbe\,(R^{71})$} & \textcolor{gray}{$adceb\,(R^{71})$} \\
$R^{73}$ & $acbed$ & $adceb$ & \textcolor{gray}{$acbde\,(R^{70})$} & \textcolor{gray}{$acbed\,(R^{70})$} & $acdbe$ & \textcolor{gray}{$acdeb\,(R^{68})$} & \textcolor{gray}{$acebd\,(R^{69})$} & \textcolor{gray}{$acedb\,(R^{68})$} & \textcolor{gray}{$adcbe\,(R^{67})$} & \textcolor{gray}{$adceb\,(R^{72})$} \\
$R^{74}$ & $acdbe$ & $adceb$ & $acdbe$ & \textcolor{gray}{$acdeb\,(R^{68})$} & \textcolor{gray}{$adcbe\,(R^{67})$} & \textcolor{gray}{$adceb\,(R^{73})$} \\
$R^{75}$ & $acdbe$ & $acedb$ & \textcolor{gray}{$acdbe\,(R^{63})$} & \textcolor{gray}{$acdeb\,(R^{74})$} & $acedb$ \\
$R^{76}$ & $acbed$ & $acdbe$ & \textcolor{gray}{$acbde\,(R^{73})$} & \textcolor{gray}{$acbed\,(R^{73})$} & $acdbe$ \\
$R^{77}$ & $acbde$ & $acedb$ & \textcolor{gray}{$acbde\,(R^{75})$} & $acbed$ & \textcolor{gray}{$acdbe\,(R^{64})$} & \textcolor{gray}{$acdeb\,(R^{65})$} & \textcolor{gray}{$acebd\,(R^{65})$} & \textcolor{gray}{$acedb\,(R^{65})$} \\
$R^{78}$ & $acbed$ & $acedb$ & $acbed$ & \textcolor{gray}{$acebd\,(R^{77})$} & \textcolor{gray}{$acedb\,(R^{77})$} \\
$R^{79}$ & $acdeb$ & $aecbd$ & \textcolor{gray}{$acdeb\,(R^{54})$} & \textcolor{gray}{$acebd\,(R^{54})$} & \textcolor{gray}{$acedb\,(R^{54})$} & $aecbd$ & \textcolor{gray}{$aecdb\,(R^{31})$} \\
$R^{80}$ & $acdbe$ & $aecbd$ & \textcolor{gray}{$acbde\,(R^{79})$} & \textcolor{gray}{$acbed\,(R^{76})$} & \textcolor{gray}{$acdbe\,(R^{79})$} & \textcolor{gray}{$acdeb\,(R^{54})$} & \textcolor{gray}{$acebd\,(R^{54})$} & \textcolor{gray}{$acedb\,(R^{54})$} & $aecbd$ & \textcolor{gray}{$aecdb\,(R^{31})$} \\
$R^{81}$ & $acbed$ & $aecbd$ & \textcolor{gray}{$acbed\,(R^{80})$} & \textcolor{gray}{$acebd\,(R^{54})$} & $aecbd$ \\
$R^{82}$ & $acdbe$ & $adecb$ & $acdbe$ & \textcolor{gray}{$acdeb\,(R^{74})$} & \textcolor{gray}{$adcbe\,(R^{74})$} & \textcolor{gray}{$adceb\,(R^{74})$} & \textcolor{gray}{$adecb\,(R^{74})$} \\
$R^{83}$ & $acdbe$ & $aedcb$ & \textcolor{gray}{$acdbe\,(R^{61})$} & \textcolor{gray}{$acdeb\,(R^{60})$} & $acedb$ & \textcolor{gray}{$adcbe\,(R^{58})$} & \textcolor{gray}{$adceb\,(R^{60})$} & \textcolor{gray}{$adecb\,(R^{58})$} & \textcolor{gray}{$aecdb\,(R^{2})$} & \textcolor{gray}{$aedcb\,(R^{82})$} \\
$R^{84}$ & $acbed$ & $aedcb$ & \textcolor{gray}{$acbed\,(R^{81})$} & \textcolor{gray}{$acebd\,(R^{2})$} & \textcolor{gray}{$acedb\,(R^{78})$} & \textcolor{gray}{$aecbd\,(R^{2})$} & \textcolor{gray}{$aecdb\,(R^{2})$} & \textcolor{gray}{$aedcb\,(R^{83})$} \\
\end{longtable}
}
}

\textbf{Step 4:} Fourthly, we will show that $f(R)=adcbe$. Once again, we assume this to be wrong, which means that $f(R)\neq adcbe$. The subsequent derivation shows that this is impossible. We note that our assumption that $f(R)\neq adcbe$ is only used at profile $R^{40}$, i.e., $R=R^{40}$. 
{\footnotesize
{\setlength{\tabcolsep}{3pt}
\noindent\begin{longtable}{c|cc|llllllll}
$R^{*}$ & $abcde$ & $acbed$ & $abcde$ (A)\\
$R^{1}$ & $aebdc$ & $aecbd$ & \multicolumn{2}{l}{$aecbd$ (A, Step 1)}\\
$R^{2}$ & $aecbd$ & $aedcb$ & \multicolumn{2}{l}{$aedcb$ (A, Step 2)}\\
$R^{3}$ & $abecd$ & $ecabd$ & \multicolumn{2}{l}{$ecabd$ (A, Step 3)}\\
$R^{4}$ & $abecd$ & $aecbd$ & $abecd$ & \textcolor{gray}{$aebcd\,(R^{1})$} & $aecbd$ \\
$R^{5}$ & $aebcd$ & $aecbd$ & \textcolor{gray}{$aebcd\,(R^{1})$} & $aecbd$ \\
$R^{6}$ & $abcde$ & $acebd$ & $abcde$ & \textcolor{gray}{$abced\,(R^{*})$} & \textcolor{gray}{$acbde\,(R^{*})$} & \textcolor{gray}{$acbed\,(R^{*})$} & \textcolor{gray}{$acebd\,(R^{*})$} \\
$R^{7}$ & $abcde$ & $aecbd$ & $abcde$ & \textcolor{gray}{$abced\,(R^{*})$} & $abecd$ & \textcolor{gray}{$acbde\,(R^{*})$} & \textcolor{gray}{$acbed\,(R^{*})$} & \textcolor{gray}{$acebd\,(R^{*})$} & \textcolor{gray}{$aebcd\,(R^{1})$} & \textcolor{gray}{$aecbd\,(R^{6})$} \\
$R^{8}$ & $abdec$ & $aecbd$ & $abdec$ & $abecd$ & \textcolor{gray}{$abedc\,(R^{1})$} & \textcolor{gray}{$aebcd\,(R^{1})$} & \textcolor{gray}{$aebdc\,(R^{1})$} & \textcolor{gray}{$aecbd\,(R^{7})$} \\
$R^{9}$ & $abedc$ & $aecbd$ & $abecd$ & \textcolor{gray}{$abedc\,(R^{1})$} & \textcolor{gray}{$aebcd\,(R^{1})$} & \textcolor{gray}{$aebdc\,(R^{1})$} & \textcolor{gray}{$aecbd\,(R^{8})$} \\
$R^{10}$ & $abecd$ & $aecbd$ & $abecd$ & \textcolor{gray}{$aebcd\,(R^{1})$} & \textcolor{gray}{$aecbd\,(R^{9})$} \\
$R^{11}$ & $abecd$ & $aecdb$ & $abecd$ & \textcolor{gray}{$aebcd\,(R^{4})$} & \textcolor{gray}{$aecbd\,(R^{10})$} & \textcolor{gray}{$aecdb\,(R^{10})$} \\
$R^{12}$ & $abecd$ & $caebd$ & \textcolor{gray}{$abced\,(R^{3})$} & \textcolor{gray}{$abecd\,(R^{3})$} & $acbed$ & \textcolor{gray}{$acebd\,(R^{10})$} & \textcolor{gray}{$aebcd\,(R^{3})$} & \textcolor{gray}{$aecbd\,(R^{10})$} & $cabed$ & $caebd$ \\
$R^{13}$ & $abecd$ & $acedb$ & $abced$ & \textcolor{gray}{$abecd\,(R^{12})$} & $acbed$ & \textcolor{gray}{$acebd\,(R^{10})$} & \textcolor{gray}{$acedb\,(R^{11})$} & \textcolor{gray}{$aebcd\,(R^{4})$} & \textcolor{gray}{$aecbd\,(R^{10})$} & \textcolor{gray}{$aecdb\,(R^{10})$} \\
$R^{14}$ & $acbed$ & $acedb$ & $acbed$ & $acebd$ & \textcolor{gray}{$acedb\,(R^{13})$} \\
$R^{15}$ & $abcde$ & $abecd$ & $abcde$ & \textcolor{gray}{$abced\,(R^{*})$} & $abecd$ \\
$R^{16}$ & $abecd$ & $aedcb$ & $abecd$ & $abedc$ & \textcolor{gray}{$aebcd\,(R^{2})$} & $aebdc$ & \textcolor{gray}{$aecbd\,(R^{2})$} & \textcolor{gray}{$aecdb\,(R^{2})$} & \textcolor{gray}{$aedbc\,(R^{11})$} & \textcolor{gray}{$aedcb\,(R^{11})$} \\
$R^{17}$ & $abedc$ & $aedcb$ & $abedc$ & $aebdc$ & \textcolor{gray}{$aedbc\,(R^{16})$} & \textcolor{gray}{$aedcb\,(R^{16})$} \\
$R^{18}$ & $aebdc$ & $aedcb$ & $aebdc$ & $aedbc$ & \textcolor{gray}{$aedcb\,(R^{17})$} \\
$R^{19}$ & $acedb$ & $aedcb$ & $acedb$ & \textcolor{gray}{$aecdb\,(R^{2})$} & $aedcb$ \\
$R^{20}$ & $aebcd$ & $aecdb$ & \textcolor{gray}{$aebcd\,(R^{5})$} & $aecbd$ & \textcolor{gray}{$aecdb\,(R^{11})$} \\
$R^{21}$ & $aecbd$ & $aecdb$ & $aecbd$ & \textcolor{gray}{$aecdb\,(R^{20})$} \\
$R^{22}$ & $abecd$ & $eacbd$ & \textcolor{gray}{$abecd\,(R^{3})$} & \textcolor{gray}{$aebcd\,(R^{3})$} & \textcolor{gray}{$aecbd\,(R^{10})$} & $eabcd$ & \textcolor{gray}{$eacbd\,(R^{10})$} \\
$R^{23}$ & $abecd$ & $eabcd$ & \textcolor{gray}{$abecd\,(R^{22})$} & \textcolor{gray}{$aebcd\,(R^{10})$} & $eabcd$ \\
$R^{24}$ & $abecd$ & $eacdb$ & \textcolor{gray}{$abecd\,(R^{3})$} & \textcolor{gray}{$aebcd\,(R^{3})$} & \textcolor{gray}{$aecbd\,(R^{10})$} & \textcolor{gray}{$aecdb\,(R^{10})$} & $eabcd$ & \textcolor{gray}{$eacbd\,(R^{10})$} & \textcolor{gray}{$eacdb\,(R^{11})$} \\
$R^{25}$ & $eabcd$ & $eacdb$ & $eabcd$ & \textcolor{gray}{$eacbd\,(R^{24})$} & \textcolor{gray}{$eacdb\,(R^{24})$} \\
$R^{26}$ & $acebd$ & $aebdc$ & $acebd$ & \textcolor{gray}{$aebcd\,(R^{1})$} & \textcolor{gray}{$aebdc\,(R^{1})$} & $aecbd$ \\
$R^{27}$ & $abced$ & $eabcd$ & \textcolor{gray}{$abced\,(R^{23})$} & \textcolor{gray}{$abecd\,(R^{23})$} & \textcolor{gray}{$aebcd\,(R^{23})$} & $eabcd$ \\
$R^{28}$ & $acbed$ & $eabcd$ & \textcolor{gray}{$abced\,(R^{23})$} & \textcolor{gray}{$abecd\,(R^{23})$} & \textcolor{gray}{$acbed\,(R^{27})$} & \textcolor{gray}{$acebd\,(R^{27})$} & \textcolor{gray}{$aebcd\,(R^{23})$} & $aecbd$ & $eabcd$ & \textcolor{gray}{$eacbd\,(R^{25})$} \\
$R^{29}$ & $acbed$ & $eacdb$ & \textcolor{gray}{$acbed\,(R^{28})$} & \textcolor{gray}{$acebd\,(R^{28})$} & \textcolor{gray}{$acedb\,(R^{14})$} & $aecbd$ & $aecdb$ & $eacbd$ & $eacdb$ \\
$R^{30}$ & $acbed$ & $aedcb$ & \textcolor{gray}{$acbed\,(R^{29})$} & \textcolor{gray}{$acebd\,(R^{2})$} & \textcolor{gray}{$acedb\,(R^{14})$} & \textcolor{gray}{$aecbd\,(R^{2})$} & \textcolor{gray}{$aecdb\,(R^{2})$} & $aedcb$ \\
$R^{31}$ & $acedb$ & $aedcb$ & \textcolor{gray}{$acedb\,(R^{30})$} & \textcolor{gray}{$aecdb\,(R^{2})$} & $aedcb$ \\
$R^{32}$ & $acebd$ & $eabcd$ & \textcolor{gray}{$acebd\,(R^{27})$} & \textcolor{gray}{$aebcd\,(R^{26})$} & $aecbd$ & $eabcd$ & \textcolor{gray}{$eacbd\,(R^{25})$} \\
$R^{33}$ & $acebd$ & $aebdc$ & \textcolor{gray}{$acebd\,(R^{32})$} & \textcolor{gray}{$aebcd\,(R^{1})$} & \textcolor{gray}{$aebdc\,(R^{1})$} & $aecbd$ \\
$R^{34}$ & $acedb$ & $aebdc$ & \textcolor{gray}{$acebd\,(R^{33})$} & $acedb$ & \textcolor{gray}{$aebcd\,(R^{1})$} & \textcolor{gray}{$aebdc\,(R^{1})$} & $aecbd$ & \textcolor{gray}{$aecdb\,(R^{19})$} & \textcolor{gray}{$aedbc\,(R^{1})$} & \textcolor{gray}{$aedcb\,(R^{18})$} \\
$R^{35}$ & $acedb$ & $aecbd$ & \textcolor{gray}{$acebd\,(R^{34})$} & $acedb$ & $aecbd$ & \textcolor{gray}{$aecdb\,(R^{21})$} \\
$R^{36}$ & $acedb$ & $aebdc$ & \textcolor{gray}{$acebd\,(R^{33})$} & \textcolor{gray}{$acedb\,(R^{31})$} & \textcolor{gray}{$aebcd\,(R^{1})$} & \textcolor{gray}{$aebdc\,(R^{1})$} & $aecbd$ & \textcolor{gray}{$aecdb\,(R^{19})$} & \textcolor{gray}{$aedbc\,(R^{1})$} & \textcolor{gray}{$aedcb\,(R^{18})$} \\
$R^{37}$ & $acedb$ & $aecbd$ & \textcolor{gray}{$acebd\,(R^{34})$} & \textcolor{gray}{$acedb\,(R^{36})$} & $aecbd$ & \textcolor{gray}{$aecdb\,(R^{21})$} \\
$R^{38}$ & $acbed$ & $aedcb$ & $acbed$ & \textcolor{gray}{$acebd\,(R^{2})$} & $acedb$ & \textcolor{gray}{$aecbd\,(R^{2})$} & \textcolor{gray}{$aecdb\,(R^{2})$} & $aedcb$ \\
$R^{39}$ & $acdeb$ & $aecbd$ & \textcolor{gray}{$acdeb\,(R^{37})$} & \textcolor{gray}{$acebd\,(R^{35})$} & \textcolor{gray}{$acedb\,(R^{37})$} & $aecbd$ & \textcolor{gray}{$aecdb\,(R^{21})$} \\
$R^{40}$ & $acbed$ & $adceb$ & \textcolor{gray}{$acbde\,(R^{30})$} & \textcolor{gray}{$acbed\,(R^{30})$} & $acdbe$ & $acdeb$ & \textcolor{gray}{$acebd\,(R^{38})$} & \textcolor{gray}{$acedb\,(R^{14})$} & \textcolor{gray}{$adcbe\,(A)$} & $adceb$ \\
$R^{41}$ & $acbed$ & $acdeb$ & \textcolor{gray}{$acbde\,(R^{40})$} & \textcolor{gray}{$acbed\,(R^{40})$} & $acdbe$ & \textcolor{gray}{$acdeb\,(R^{39})$} & \textcolor{gray}{$acebd\,(R^{40})$} & \textcolor{gray}{$acedb\,(R^{14})$} \\
$R^{42}$ & $acbed$ & $acdbe$ & \textcolor{gray}{$acbde\,(R^{41})$} & \textcolor{gray}{$acbed\,(R^{41})$} & $acdbe$ \\
$R^{43}$ & $acdbe$ & $aecbd$ & \textcolor{gray}{$acbde\,(R^{39})$} & $acbed$ & \textcolor{gray}{$acdbe\,(R^{39})$} & \textcolor{gray}{$acdeb\,(R^{37})$} & \textcolor{gray}{$acebd\,(R^{35})$} & \textcolor{gray}{$acedb\,(R^{37})$} & $aecbd$ & \textcolor{gray}{$aecdb\,(R^{21})$} \\
$R^{44}$ & $abecd$ & $acdbe$ & $abcde$ & \textcolor{gray}{$abced\,(R^{15})$} & \textcolor{gray}{$abecd\,(R^{13})$} & \textcolor{gray}{$acbde\,(R^{43})$} & \textcolor{gray}{$acbed\,(R^{42})$} & \textcolor{gray}{$acdbe\,(R^{43})$} \\
$R^{45}$ & $abcde$ & $abecd$ & $abcde$ & \textcolor{gray}{$abced\,(R^{*})$} & \textcolor{gray}{$abecd\,(R^{44})$} \\
$R^{46}$ & $acdbe$ & $aecbd$ & \textcolor{gray}{$acbde\,(R^{39})$} & \textcolor{gray}{$acbed\,(R^{42})$} & \textcolor{gray}{$acdbe\,(R^{39})$} & \textcolor{gray}{$acdeb\,(R^{37})$} & \textcolor{gray}{$acebd\,(R^{35})$} & \textcolor{gray}{$acedb\,(R^{37})$} & $aecbd$ & \textcolor{gray}{$aecdb\,(R^{21})$} \\
$R^{47}$ & $abcde$ & $aecbd$ & \textcolor{gray}{$abcde\,(R^{46})$} & \textcolor{gray}{$abced\,(R^{*})$} & \textcolor{gray}{$abecd\,(R^{45})$} & \textcolor{gray}{$acbde\,(R^{*})$} & \textcolor{gray}{$acbed\,(R^{*})$} & \textcolor{gray}{$acebd\,(R^{*})$} & \textcolor{gray}{$aebcd\,(R^{1})$} & \textcolor{gray}{$aecbd\,(R^{6})$} \\
\end{longtable}
}
}

\textbf{Step 5:} Finally, we will show that the information we have inferred so far is contradictory. 
{\footnotesize
{\setlength{\tabcolsep}{3pt}
\noindent\begin{longtable}{c|cc|llllllll}
$R^{*}$ & $abcde$ & $acbed$ & $abcde$ (A)\\
$R^{1}$ & $aebdc$ & $aecbd$ & \multicolumn{2}{l}{$aecbd$ (A, Step 1)}\\
$R^{2}$ & $aecbd$ & $aedcb$ & \multicolumn{2}{l}{$aedcb$ (A, Step 2)}\\
$R^{3}$ & $abecd$ & $ecabd$ & \multicolumn{2}{l}{$ecabd$ (A, Step 3)}\\
$R^{4}$ & $acbed$ & $adceb$ & \multicolumn{2}{l}{$adcbe$ (A, Step 4)}\\
$R^{5}$ & $abecd$ & $aecbd$ & $abecd$ & \textcolor{gray}{$aebcd\,(R^{1})$} & $aecbd$ \\
$R^{6}$ & $aebcd$ & $aecbd$ & \textcolor{gray}{$aebcd\,(R^{1})$} & $aecbd$ \\
$R^{7}$ & $abcde$ & $acebd$ & $abcde$ & \textcolor{gray}{$abced\,(R^{*})$} & \textcolor{gray}{$acbde\,(R^{*})$} & \textcolor{gray}{$acbed\,(R^{*})$} & \textcolor{gray}{$acebd\,(R^{*})$} \\
$R^{8}$ & $abcde$ & $aecbd$ & $abcde$ & \textcolor{gray}{$abced\,(R^{*})$} & $abecd$ & \textcolor{gray}{$acbde\,(R^{*})$} & \textcolor{gray}{$acbed\,(R^{*})$} & \textcolor{gray}{$acebd\,(R^{*})$} & \textcolor{gray}{$aebcd\,(R^{1})$} & \textcolor{gray}{$aecbd\,(R^{7})$} \\
$R^{9}$ & $abdec$ & $aecbd$ & $abdec$ & $abecd$ & \textcolor{gray}{$abedc\,(R^{1})$} & \textcolor{gray}{$aebcd\,(R^{1})$} & \textcolor{gray}{$aebdc\,(R^{1})$} & \textcolor{gray}{$aecbd\,(R^{8})$} \\
$R^{10}$ & $abedc$ & $aecbd$ & $abecd$ & \textcolor{gray}{$abedc\,(R^{1})$} & \textcolor{gray}{$aebcd\,(R^{1})$} & \textcolor{gray}{$aebdc\,(R^{1})$} & \textcolor{gray}{$aecbd\,(R^{9})$} \\
$R^{11}$ & $abecd$ & $aecbd$ & $abecd$ & \textcolor{gray}{$aebcd\,(R^{1})$} & \textcolor{gray}{$aecbd\,(R^{10})$} \\
$R^{12}$ & $abecd$ & $aecdb$ & $abecd$ & \textcolor{gray}{$aebcd\,(R^{5})$} & \textcolor{gray}{$aecbd\,(R^{11})$} & \textcolor{gray}{$aecdb\,(R^{11})$} \\
$R^{13}$ & $abecd$ & $caebd$ & \textcolor{gray}{$abced\,(R^{3})$} & \textcolor{gray}{$abecd\,(R^{3})$} & $acbed$ & \textcolor{gray}{$acebd\,(R^{11})$} & \textcolor{gray}{$aebcd\,(R^{3})$} & \textcolor{gray}{$aecbd\,(R^{11})$} & $cabed$ & $caebd$ \\
$R^{14}$ & $abecd$ & $acedb$ & $abced$ & \textcolor{gray}{$abecd\,(R^{13})$} & $acbed$ & \textcolor{gray}{$acebd\,(R^{11})$} & \textcolor{gray}{$acedb\,(R^{12})$} & \textcolor{gray}{$aebcd\,(R^{5})$} & \textcolor{gray}{$aecbd\,(R^{11})$} & \textcolor{gray}{$aecdb\,(R^{11})$} \\
$R^{15}$ & $acbed$ & $acedb$ & $acbed$ & $acebd$ & \textcolor{gray}{$acedb\,(R^{14})$} \\
$R^{16}$ & $acedb$ & $aedcb$ & $acedb$ & \textcolor{gray}{$aecdb\,(R^{2})$} & $aedcb$ \\
$R^{17}$ & $abcde$ & $abecd$ & $abcde$ & \textcolor{gray}{$abced\,(R^{*})$} & $abecd$ \\
$R^{18}$ & $aebcd$ & $aecdb$ & \textcolor{gray}{$aebcd\,(R^{6})$} & $aecbd$ & \textcolor{gray}{$aecdb\,(R^{12})$} \\
$R^{19}$ & $aecbd$ & $aecdb$ & $aecbd$ & \textcolor{gray}{$aecdb\,(R^{18})$} \\
$R^{20}$ & $acbed$ & $acdbe$ & $acbde$ & \textcolor{gray}{$acbed\,(R^{4})$} & $acdbe$ \\
$R^{21}$ & $abecd$ & $eacbd$ & \textcolor{gray}{$abecd\,(R^{3})$} & \textcolor{gray}{$aebcd\,(R^{3})$} & \textcolor{gray}{$aecbd\,(R^{11})$} & $eabcd$ & \textcolor{gray}{$eacbd\,(R^{11})$} \\
$R^{22}$ & $abecd$ & $eabcd$ & \textcolor{gray}{$abecd\,(R^{21})$} & \textcolor{gray}{$aebcd\,(R^{11})$} & $eabcd$ \\
$R^{23}$ & $abecd$ & $eacdb$ & \textcolor{gray}{$abecd\,(R^{3})$} & \textcolor{gray}{$aebcd\,(R^{3})$} & \textcolor{gray}{$aecbd\,(R^{11})$} & \textcolor{gray}{$aecdb\,(R^{11})$} & $eabcd$ & \textcolor{gray}{$eacbd\,(R^{11})$} & \textcolor{gray}{$eacdb\,(R^{12})$} \\
$R^{24}$ & $eabcd$ & $eacdb$ & $eabcd$ & \textcolor{gray}{$eacbd\,(R^{23})$} & \textcolor{gray}{$eacdb\,(R^{23})$} \\
$R^{25}$ & $acebd$ & $aebdc$ & $acebd$ & \textcolor{gray}{$aebcd\,(R^{1})$} & \textcolor{gray}{$aebdc\,(R^{1})$} & $aecbd$ \\
$R^{26}$ & $acbed$ & $aedcb$ & \textcolor{gray}{$acbed\,(R^{4})$} & \textcolor{gray}{$acebd\,(R^{2})$} & \textcolor{gray}{$acedb\,(R^{15})$} & \textcolor{gray}{$aecbd\,(R^{2})$} & \textcolor{gray}{$aecdb\,(R^{2})$} & $aedcb$ \\
$R^{27}$ & $abecd$ & $aedcb$ & \textcolor{gray}{$abecd\,(R^{26})$} & \textcolor{gray}{$abedc\,(R^{26})$} & \textcolor{gray}{$aebcd\,(R^{2})$} & $aebdc$ & \textcolor{gray}{$aecbd\,(R^{2})$} & \textcolor{gray}{$aecdb\,(R^{2})$} & \textcolor{gray}{$aedbc\,(R^{12})$} & \textcolor{gray}{$aedcb\,(R^{12})$} \\
$R^{28}$ & $aebdc$ & $aedcb$ & $aebdc$ & \textcolor{gray}{$aedbc\,(R^{27})$} & \textcolor{gray}{$aedcb\,(R^{27})$} \\
$R^{29}$ & $acedb$ & $aedcb$ & \textcolor{gray}{$acedb\,(R^{26})$} & \textcolor{gray}{$aecdb\,(R^{2})$} & $aedcb$ \\
$R^{30}$ & $abced$ & $eabcd$ & \textcolor{gray}{$abced\,(R^{22})$} & \textcolor{gray}{$abecd\,(R^{22})$} & \textcolor{gray}{$aebcd\,(R^{22})$} & $eabcd$ \\
$R^{31}$ & $acebd$ & $eabcd$ & \textcolor{gray}{$acebd\,(R^{30})$} & \textcolor{gray}{$aebcd\,(R^{25})$} & $aecbd$ & $eabcd$ & \textcolor{gray}{$eacbd\,(R^{24})$} \\
$R^{32}$ & $acebd$ & $aebdc$ & \textcolor{gray}{$acebd\,(R^{31})$} & \textcolor{gray}{$aebcd\,(R^{1})$} & \textcolor{gray}{$aebdc\,(R^{1})$} & $aecbd$ \\
$R^{33}$ & $acedb$ & $aebdc$ & \textcolor{gray}{$acebd\,(R^{32})$} & \textcolor{gray}{$acedb\,(R^{29})$} & \textcolor{gray}{$aebcd\,(R^{1})$} & \textcolor{gray}{$aebdc\,(R^{1})$} & $aecbd$ & \textcolor{gray}{$aecdb\,(R^{16})$} & \textcolor{gray}{$aedbc\,(R^{1})$} & \textcolor{gray}{$aedcb\,(R^{28})$} \\
$R^{34}$ & $acedb$ & $aecbd$ & \textcolor{gray}{$acebd\,(R^{33})$} & \textcolor{gray}{$acedb\,(R^{33})$} & $aecbd$ & \textcolor{gray}{$aecdb\,(R^{19})$} \\
$R^{35}$ & $acdeb$ & $aecbd$ & \textcolor{gray}{$acdeb\,(R^{34})$} & \textcolor{gray}{$acebd\,(R^{34})$} & \textcolor{gray}{$acedb\,(R^{34})$} & $aecbd$ & \textcolor{gray}{$aecdb\,(R^{19})$} \\
$R^{36}$ & $acdbe$ & $aecbd$ & \textcolor{gray}{$acbde\,(R^{35})$} & \textcolor{gray}{$acbed\,(R^{20})$} & \textcolor{gray}{$acdbe\,(R^{35})$} & \textcolor{gray}{$acdeb\,(R^{34})$} & \textcolor{gray}{$acebd\,(R^{34})$} & \textcolor{gray}{$acedb\,(R^{34})$} & $aecbd$ & \textcolor{gray}{$aecdb\,(R^{19})$} \\
$R^{37}$ & $abecd$ & $acdbe$ & $abcde$ & \textcolor{gray}{$abced\,(R^{17})$} & \textcolor{gray}{$abecd\,(R^{14})$} & \textcolor{gray}{$acbde\,(R^{36})$} & \textcolor{gray}{$acbed\,(R^{20})$} & \textcolor{gray}{$acdbe\,(R^{36})$} \\
$R^{38}$ & $abcde$ & $aecbd$ & \textcolor{gray}{$abcde\,(R^{36})$} & \textcolor{gray}{$abced\,(R^{*})$} & $abecd$ & \textcolor{gray}{$acbde\,(R^{*})$} & \textcolor{gray}{$acbed\,(R^{*})$} & \textcolor{gray}{$acebd\,(R^{*})$} & \textcolor{gray}{$aebcd\,(R^{1})$} & \textcolor{gray}{$aecbd\,(R^{7})$} \\
$R^{39}$ & $abcde$ & $abecd$ & \textcolor{gray}{$abcde\,(R^{38})$} & \textcolor{gray}{$abced\,(R^{*})$} & \textcolor{gray}{$abecd\,(R^{37})$} \\
\end{longtable}
}
}

Since Step 5 only uses the assumption on $R^*$ and the insights proven in the previous steps, this proves our lemma.
\end{proof}

\begin{lemma}
    $f(R^+)\neq eacbd$.
\end{lemma}
\begin{proof}
    We again assume for contradiction that $f(R^+)=eacbd$ and derive a contradiction in multiple steps. \medskip

    \textbf{Step 1:} First, we will show that $f(R)=eacbd$ for the profile $R$ where one voter reports $ceabd$ and the other reports $eacbd$. The following derivation shows that we get an impossibility if $f(R)\neq eacbd$, thus proving our claim. Our assumption that $f(R)\neq eacbd$ is used at profile $R^{1}$ and $R^{41}$.
    {\footnotesize
    {\setlength{\tabcolsep}{3pt}
\noindent
\begin{longtable}{c|cc|llllllll}
$R^{*}$ & $abcde$ & $acbed$ & $abcde$ (A)\\
$R^{+}$ & $eabcd$ & $ecabd$ & $eacbd$ (A)\\
$R^{1}$ & $ceabd$ & $eacbd$ & $ceabd$ & \textcolor{gray}{$eacbd\,(A)$} & $ecabd$ \\
$R^{2}$ & $ceabd$ & $eabcd$ & $ceabd$ & \textcolor{gray}{$eabcd\,(R^{+})$} & \textcolor{gray}{$eacbd\,(R^{1})$} & \textcolor{gray}{$ecabd\,(R^{+})$} \\
$R^{3}$ & $ceabd$ & $eabdc$ & $ceabd$ & \textcolor{gray}{$eabcd\,(R^{2})$} & \textcolor{gray}{$eabdc\,(R^{2})$} & \textcolor{gray}{$eacbd\,(R^{1})$} & \textcolor{gray}{$ecabd\,(R^{2})$} \\
$R^{4}$ & $eabdc$ & $ecabd$ & \textcolor{gray}{$eabcd\,(R^{+})$} & \textcolor{gray}{$eabdc\,(R^{3})$} & $eacbd$ & \textcolor{gray}{$ecabd\,(R^{+})$} \\
$R^{5}$ & $eabdc$ & $eacbd$ & \textcolor{gray}{$eabcd\,(R^{4})$} & \textcolor{gray}{$eabdc\,(R^{4})$} & $eacbd$ \\
$R^{6}$ & $abcde$ & $acebd$ & $abcde$ & \textcolor{gray}{$abced\,(R^{*})$} & \textcolor{gray}{$acbde\,(R^{*})$} & \textcolor{gray}{$acbed\,(R^{*})$} & \textcolor{gray}{$acebd\,(R^{*})$} \\
\end{longtable}
}
    }

    \textbf{Step 2:} Our next goal is to show that $f(R)=aebcd$ for the profile $R$ where one voter reports $aebcd$ and the other reports $ecabd$. We hence assume that $f(R)\neq aebcd$ and derive a contradiction, as shown in the following table. The profile $R$ appears at Steps 2 and 11 (i.e., $R^2$ and $R^{11}$) of our derivation.
    {\footnotesize
    {\setlength{\tabcolsep}{3pt}
\noindent\begin{longtable}{c|cc|llllllll}
$R^{*}$ & $abcde$ & $acbed$ & $abcde$ (A)\\
$R^{+}$ & $eabcd$ & $ecabd$ & $eacbd$ (A)\\
$R^{1}$ & $ceabd$ & $eacbd$ & \multicolumn{2}{l}{$eacbd$ (A, Step 1)}\\
$R^{2}$ & $aebcd$ & $ecabd$ & \textcolor{gray}{$aebcd\,(A)$} & $aecbd$ & \textcolor{gray}{$eabcd\,(R^{+})$} & $eacbd$ & \textcolor{gray}{$ecabd\,(R^{+})$} \\
$R^{3}$ & $eacbd$ & $ecbad$ & $eacbd$ & \textcolor{gray}{$ecabd\,(R^{1})$} & $ecbad$ \\
$R^{4}$ & $abcde$ & $acebd$ & $abcde$ & \textcolor{gray}{$abced\,(R^{*})$} & \textcolor{gray}{$acbde\,(R^{*})$} & \textcolor{gray}{$acbed\,(R^{*})$} & \textcolor{gray}{$acebd\,(R^{*})$} \\
$R^{5}$ & $abcde$ & $aecbd$ & $abcde$ & \textcolor{gray}{$abced\,(R^{*})$} & $abecd$ & \textcolor{gray}{$acbde\,(R^{*})$} & \textcolor{gray}{$acbed\,(R^{*})$} & \textcolor{gray}{$acebd\,(R^{*})$} & \textcolor{gray}{$aebcd\,(R^{4})$} & \textcolor{gray}{$aecbd\,(R^{4})$} \\
$R^{6}$ & $abced$ & $aecbd$ & $abced$ & $abecd$ & $acbed$ & \textcolor{gray}{$acebd\,(R^{5})$} & \textcolor{gray}{$aebcd\,(R^{5})$} & \textcolor{gray}{$aecbd\,(R^{5})$} \\
\end{longtable}
}
    }

    \textbf{Step 3:} As the next step, we will show that $f(R)=aebcd$ for the profile $R$ where one voter reports $acebd$ and the other reports $aebcd$. Once again, we assume for contradiction that this is not true and derive a contradiction. The profile $R$ is called $R^3$ in the subsequent derivation. 
    {\footnotesize
    {\setlength{\tabcolsep}{3pt}
\noindent\begin{longtable}{c|cc|llllllll}
$R^{*}$ & $abcde$ & $acbed$ & $abcde$ (A)\\
$R^{+}$ & $eabcd$ & $ecabd$ & $eacbd$ (A)\\
$R^{1}$ & $ceabd$ & $eacbd$ & \multicolumn{2}{l}{$eacbd$ (A, Step 1)}\\
$R^{2}$ & $aebcd$ & $ecabd$ & \multicolumn{2}{l}{$aebcd$ (A, Step 2)}\\
$R^{3}$ & $acebd$ & $aebcd$ & $acebd$ & \textcolor{gray}{$aebcd\,(A)$} & \textcolor{gray}{$aecbd\,(R^{2})$} \\
$R^{4}$ & $abcde$ & $acebd$ & $abcde$ & \textcolor{gray}{$abced\,(R^{*})$} & \textcolor{gray}{$acbde\,(R^{*})$} & \textcolor{gray}{$acbed\,(R^{*})$} & \textcolor{gray}{$acebd\,(R^{*})$} \\
$R^{5}$ & $acbde$ & $acebd$ & $acbde$ & $acbed$ & \textcolor{gray}{$acebd\,(R^{4})$} \\
$R^{6}$ & $abcde$ & $acbde$ & $abcde$ & \textcolor{gray}{$acbde\,(R^{*})$} \\
\end{longtable}
}
    }

    \textbf{Step 4:} Fourthly, we will show that $f(R)=abecd$ for the profile $R$ where one voter reports $acbde$ and the other reports $aebcd$. In more detail, the following derivation shows that $f(R)\neq abecd$ results in a contradiction. The profile $R$ is called $R^9$ in the following table. 
    {\footnotesize
    {\setlength{\tabcolsep}{3pt}
\noindent\begin{longtable}{c|cc|llllllll}
$R^{*}$ & $abcde$ & $acbed$ & $abcde$ (A)\\
$R^{+}$ & $eabcd$ & $ecabd$ & $eacbd$ (A)\\
$R^{1}$ & $ceabd$ & $eacbd$ & \multicolumn{2}{l}{$eacbd$ (A, Step 1)}\\
$R^{2}$ & $aebcd$ & $ecabd$ & \multicolumn{2}{l}{$aebcd$ (A, Step 2)}\\
$R^{3}$ & $acebd$ & $aebcd$ & \multicolumn{2}{l}{$aebcd$ (A, Step 3)}\\
$R^{4}$ & $abcde$ & $acbde$ & $abcde$ & \textcolor{gray}{$acbde\,(R^{*})$} \\
$R^{5}$ & $aebcd$ & $aecbd$ & $aebcd$ & \textcolor{gray}{$aecbd\,(R^{3})$} \\
$R^{6}$ & $aebcd$ & $aecdb$ & $aebcd$ & \textcolor{gray}{$aecbd\,(R^{3})$} & $aecdb$ \\
\end{longtable}
}
    }

    \textbf{Step 5:} In our fifth step, we will prove that $f(R)=acbed$ for the profile $R$ where one voter reports $acbed$ and the other reports $eacbd$. To this end, we assume that $f(R)\neq acbed$ and derive a contradiction. Our assumption that $f(R)\neq acbed$ appears at profile $R^{12}$ and $R^{27}$. 
    {\footnotesize
    {\setlength{\tabcolsep}{3pt}
\noindent\begin{longtable}{c|cc|llllllll}
$R^{*}$ & $abcde$ & $acbed$ & $abcde$ (A)\\
$R^{+}$ & $eabcd$ & $ecabd$ & $eacbd$ (A)\\
$R^{1}$ & $ceabd$ & $eacbd$ & \multicolumn{2}{l}{$eacbd$ (A, Step 1)}\\
$R^{2}$ & $aebcd$ & $ecabd$ & \multicolumn{2}{l}{$aebcd$ (A, Step 2)}\\
$R^{3}$ & $acebd$ & $aebcd$ & \multicolumn{2}{l}{$aebcd$ (A, Step 3)}\\
$R^{4}$ & $acbde$ & $aebcd$ & \multicolumn{2}{l}{$abecd$ (A, Step 4)}\\
$R^{5}$ & $aebcd$ & $aecbd$ & $aebcd$ & \textcolor{gray}{$aecbd\,(R^{3})$} \\
$R^{6}$ & $abcde$ & $aebcd$ & \textcolor{gray}{$abcde\,(R^{4})$} & \textcolor{gray}{$abced\,(R^{4})$} & $abecd$ & \textcolor{gray}{$aebcd\,(R^{4})$} \\
\end{longtable}
}
    }

    \textbf{Step 6:} Finally, we derive a contradiction by showing that the insights of our previous steps are incompatible with each other. 
    {\footnotesize
    {\setlength{\tabcolsep}{3pt}
\noindent\begin{longtable}{c|cc|llllllll}
$R^{+}$ & $eabcd$ & $ecabd$ & $eacbd$ (A)\\
$R^{1}$ & $ceabd$ & $eacbd$ & \multicolumn{2}{l}{$eacbd$ (A, Step 1)}\\
$R^{2}$ & $aebcd$ & $ecabd$ & \multicolumn{2}{l}{$aebcd$ (A, Step 2)}\\
$R^{3}$ & $acebd$ & $aebcd$ & \multicolumn{2}{l}{$aebcd$ (A, Step 3)}\\
$R^{4}$ & $acbde$ & $aebcd$ & \multicolumn{2}{l}{$abecd$ (A, Step 4)}\\
$R^{5}$ & $acbed$ & $eacbd$ & \multicolumn{2}{l}{$acbed$ (A, Step 5)}\\
$R^{6}$ & $aebcd$ & $cabed$ & \textcolor{gray}{$abced\,(R^{4})$} & $abecd$ & \textcolor{gray}{$acbed\,(R^{4})$} & \textcolor{gray}{$acebd\,(R^{4})$} & \textcolor{gray}{$aebcd\,(R^{4})$} & \textcolor{gray}{$aecbd\,(R^{3})$} & \textcolor{gray}{$cabed\,(R^{4})$} & \textcolor{gray}{$caebd\,(R^{3})$} \\
$R^{7}$ & $cabed$ & $eacbd$ & $acbed$ & \textcolor{gray}{$acebd\,(R^{5})$} & \textcolor{gray}{$aecbd\,(R^{5})$} & \textcolor{gray}{$cabed\,(R^{6})$} & \textcolor{gray}{$caebd\,(R^{1})$} & \textcolor{gray}{$ceabd\,(R^{5})$} & \textcolor{gray}{$eacbd\,(R^{5})$} & \textcolor{gray}{$ecabd\,(R^{5})$} \\
\end{longtable}
}
    }
    
This concludes the proof of this lemma.
\end{proof}

\begin{lemma}
    $f(R^+)\neq eabcd$.
\end{lemma}
\begin{proof}
As usual, we will assume for contradiction that $f(R^+)=eabcde$. Howbever, in contrast to the previous two lemmas, we will derive a contradiction by means of a further case distinction. Specifically, we will consider the outcomes for the profile $R^1$ shown below. 
\begin{center}
\begin{tabular}{lll}
    $R^1$: $caebd$ & $cebad$
\end{tabular}
\end{center}

For this profile, only three rankings satisfy unanimity: $f(R^1)=caebd$, $f(R^1)=ceabd$, and $f(R^1)=cebad$. We will show that none of these outcomes is compatible with our other assumptions on $f$.\medskip

\textbf{Case 1}: We first assume that $f(R^1)=ceabd$. In this case, we immediately can derive a contradiction as witnessed by the following deduction.
{\footnotesize
{\setlength{\tabcolsep}{3pt}
\noindent\begin{longtable}{c|cc|llllllll}
$R^{*}$ & $abcde$ & $acbed$ & $abcde$ (A)\\
$R^{+}$ & $eabcd$ & $ecabd$ & $eabcd$ (A)\\
$R^{1}$ & $caebd$ & $cebad$ & $ceabd$ (A)\\
$R^{2}$ & $ceabd$ & $eabcd$ & \textcolor{gray}{$ceabd\,(R^{+})$} & $eabcd$ & \textcolor{gray}{$eacbd\,(R^{+})$} & \textcolor{gray}{$ecabd\,(R^{+})$} \\
$R^{3}$ & $ceabd$ & $ebcad$ & \textcolor{gray}{$ceabd\,(R^{2})$} & $cebad$ & $ebcad$ & $ecabd$ & $ecbad$ \\
$R^{4}$ & $abcde$ & $cabed$ & $abcde$ & \textcolor{gray}{$abced\,(R^{*})$} & \textcolor{gray}{$acbde\,(R^{*})$} & \textcolor{gray}{$acbed\,(R^{*})$} & $cabde$ & \textcolor{gray}{$cabed\,(R^{*})$} \\
$R^{5}$ & $acbde$ & $cabed$ & $acbde$ & $acbed$ & $cabde$ & \textcolor{gray}{$cabed\,(R^{4})$} \\
$R^{6}$ & $abcde$ & $acebd$ & $abcde$ & \textcolor{gray}{$abced\,(R^{*})$} & \textcolor{gray}{$acbde\,(R^{*})$} & \textcolor{gray}{$acbed\,(R^{*})$} & \textcolor{gray}{$acebd\,(R^{*})$} \\
\end{longtable}
}
}

\textbf{Case 2:} As the second case, we assume $f(R^1)=caebd$. In this case, we proceed with a third tier case distinction with respect to the profile $R^2$ where one voter reports $caebd$ and the other reports $ecabd$.\medskip

\textbf{Case 2.1:} First, we assume that $f(R^2)\neq caebd$. Then, our assumption become incompatible as shown by the following derivation.
{\footnotesize
{\setlength{\tabcolsep}{3pt}
\noindent\begin{longtable}{c|cc|llllllll}
$R^{*}$ & $abcde$ & $acbed$ & $abcde$ (A)\\
$R^{1}$ & $caebd$ & $cebad$ & $caebd$ (A)\\
$R^{2}$ & $caebd$ & $ecabd$ & \textcolor{gray}{$caebd\,(A)$} & \textcolor{gray}{$ceabd\,(R^{1})$} & $ecabd$ \\
$R^{3}$ & $cabed$ & $ecabd$ & \textcolor{gray}{$cabed\,(R^{2})$} & \textcolor{gray}{$caebd\,(R^{2})$} & \textcolor{gray}{$ceabd\,(R^{2})$} & $ecabd$ \\
$R^{4}$ & $cabde$ & $ecabd$ & \textcolor{gray}{$cabde\,(R^{3})$} & \textcolor{gray}{$cabed\,(R^{2})$} & \textcolor{gray}{$caebd\,(R^{2})$} & \textcolor{gray}{$ceabd\,(R^{2})$} & $ecabd$ \\
$R^{5}$ & $cabde$ & $ceabd$ & \textcolor{gray}{$cabde\,(R^{4})$} & \textcolor{gray}{$cabed\,(R^{4})$} & $caebd$ & $ceabd$ \\
$R^{6}$ & $cabde$ & $cebad$ & \textcolor{gray}{$cabde\,(R^{4})$} & \textcolor{gray}{$cabed\,(R^{4})$} & $caebd$ & \textcolor{gray}{$cbade\,(R^{4})$} & $cbaed$ & \textcolor{gray}{$cbead\,(R^{1})$} & \textcolor{gray}{$ceabd\,(R^{1})$} & \textcolor{gray}{$cebad\,(R^{1})$} \\
$R^{7}$ & $cabde$ & $ceabd$ & \textcolor{gray}{$cabde\,(R^{4})$} & \textcolor{gray}{$cabed\,(R^{4})$} & $caebd$ & \textcolor{gray}{$ceabd\,(R^{6})$} \\
\end{longtable}
}
}

\textbf{Case 2.2:} As our second subcase, we suppose that $f(R^2)=caebd$. We again derive a contradiction for this case, as shown in the following. By combining our two subcases, it follows that $f(R^1)=caebd$ must be wrong. 
{\footnotesize
{\setlength{\tabcolsep}{3pt}
\noindent\begin{longtable}{c|cc|llllllll}
$R^{*}$ & $abcde$ & $acbed$ & $abcde$ (A)\\
$R^{+}$ & $eabcd$ & $ecabd$ & $eabcd$ (A)\\
$R^{1}$ & $caebd$ & $cebad$ & $caebd$ (A)\\
$R^{2}$ & $caebd$ & $ecabd$ & $caebd$ (A)\\
$R^{3}$ & $ceabd$ & $eabcd$ & \textcolor{gray}{$ceabd\,(R^{+})$} & $eabcd$ & \textcolor{gray}{$eacbd\,(R^{+})$} & \textcolor{gray}{$ecabd\,(R^{+})$} \\
$R^{4}$ & $cebad$ & $eabcd$ & \textcolor{gray}{$ceabd\,(R^{+})$} & \textcolor{gray}{$cebad\,(R^{3})$} & $eabcd$ & \textcolor{gray}{$eacbd\,(R^{+})$} & $ebacd$ & $ebcad$ & \textcolor{gray}{$ecabd\,(R^{+})$} & \textcolor{gray}{$ecbad\,(R^{+})$} \\
$R^{5}$ & $cebad$ & $ebacd$ & \textcolor{gray}{$cebad\,(R^{4})$} & $ebacd$ & $ebcad$ & \textcolor{gray}{$ecbad\,(R^{4})$} \\
$R^{6}$ & $cebad$ & $ebcad$ & \textcolor{gray}{$cebad\,(R^{5})$} & $ebcad$ & $ecbad$ \\
\end{longtable}
}
}

\textbf{Case 3:} Lastly, we suppose that $f(R^1)=cebad$. To derive a contradiction in this case, we will use a further case distinction with respect to the profile $R^2$ where one voter reports $aecbd$ and the other reports $cabed$. It can be checked that there are $5$ possible outcomes for this profile $f(R^2)\in \{acbed, acebd, aecbd, cabed, caebd\}$. We will subsequently show that none of these outcomes is feasible by considering four separate cases. 
\medskip

\textbf{Case 3.1:} As the first case, we suppose that $f(R^2)\not\in \{acbed, acebd, aecbd\}$. In this case, the following derivation shows that our assumptions are incompatible. Note that the profile $R^2$ appears again at $R^{11}$ as we can conclude that $f(R^2)=f(R^{11})=cabed$ at this point. 
{\footnotesize
{\setlength{\tabcolsep}{3pt}
\noindent\begin{longtable}{c|cc|llllllll}
$R^{*}$ & $abcde$ & $acbed$ & $abcde$ (A)\\
$R^{+}$ & $eabcd$ & $ecabd$ & $eabcd$ (A)\\
$R^{1}$ & $caebd$ & $cebad$ & $cebad$ (A)\\
$R^{2}$ & $aecbd$ & $cabed$ & \textcolor{gray}{$acbed\,(A)$} & \textcolor{gray}{$acebd\,(A)$} & \textcolor{gray}{$aecbd\,(A)$} & $cabed$ & $caebd$ \\
$R^{3}$ & $aecbd$ & $caebd$ & $acebd$ & \textcolor{gray}{$aecbd\,(R^{2})$} & $caebd$ \\
$R^{4}$ & $acebd$ & $cebad$ & \textcolor{gray}{$acebd\,(R^{1})$} & \textcolor{gray}{$caebd\,(R^{1})$} & \textcolor{gray}{$ceabd\,(R^{1})$} & $cebad$ \\
$R^{5}$ & $acebd$ & $ceabd$ & \textcolor{gray}{$acebd\,(R^{4})$} & $caebd$ & $ceabd$ \\
$R^{6}$ & $ceabd$ & $eabcd$ & \textcolor{gray}{$ceabd\,(R^{+})$} & $eabcd$ & \textcolor{gray}{$eacbd\,(R^{+})$} & \textcolor{gray}{$ecabd\,(R^{+})$} \\
$R^{7}$ & $caebd$ & $eabcd$ & \textcolor{gray}{$acebd\,(R^{6})$} & $aebcd$ & \textcolor{gray}{$aecbd\,(R^{3})$} & \textcolor{gray}{$caebd\,(R^{6})$} & \textcolor{gray}{$ceabd\,(R^{+})$} & $eabcd$ & \textcolor{gray}{$eacbd\,(R^{+})$} & \textcolor{gray}{$ecabd\,(R^{+})$} \\
$R^{8}$ & $aebcd$ & $caebd$ & \textcolor{gray}{$acebd\,(R^{7})$} & $aebcd$ & \textcolor{gray}{$aecbd\,(R^{3})$} & \textcolor{gray}{$caebd\,(R^{7})$} \\
$R^{9}$ & $aebcd$ & $aecdb$ & $aebcd$ & \textcolor{gray}{$aecbd\,(R^{8})$} & $aecdb$ \\
$R^{10}$ & $aecbd$ & $caebd$ & $acebd$ & \textcolor{gray}{$aecbd\,(R^{2})$} & \textcolor{gray}{$caebd\,(R^{8})$} \\
$R^{11}$ & $aecbd$ & $cabed$ & \textcolor{gray}{$acbed\,(A)$} & \textcolor{gray}{$acebd\,(A)$} & \textcolor{gray}{$aecbd\,(A)$} & $cabed$ & \textcolor{gray}{$caebd\,(R^{10})$} \\
$R^{12}$ & $aebcd$ & $cabed$ & $abced$ & \textcolor{gray}{$abecd\,(R^{11})$} & \textcolor{gray}{$acbed\,(R^{8})$} & \textcolor{gray}{$acebd\,(R^{2})$} & \textcolor{gray}{$aebcd\,(R^{2})$} & \textcolor{gray}{$aecbd\,(R^{2})$} & \textcolor{gray}{$cabed\,(R^{8})$} & \textcolor{gray}{$caebd\,(R^{8})$} \\
$R^{13}$ & $abced$ & $cabed$ & $abced$ & \textcolor{gray}{$acbed\,(R^{11})$} & \textcolor{gray}{$cabed\,(R^{12})$} \\
$R^{14}$ & $abcde$ & $acebd$ & $abcde$ & \textcolor{gray}{$abced\,(R^{*})$} & \textcolor{gray}{$acbde\,(R^{*})$} & \textcolor{gray}{$acbed\,(R^{*})$} & \textcolor{gray}{$acebd\,(R^{*})$} \\
$R^{15}$ & $abcde$ & $aebcd$ & $abcde$ & \textcolor{gray}{$abced\,(R^{*})$} & \textcolor{gray}{$abecd\,(R^{12})$} & \textcolor{gray}{$aebcd\,(R^{14})$} \\
$R^{16}$ & $acdeb$ & $aebcd$ & $acdeb$ & \textcolor{gray}{$acebd\,(R^{8})$} & \textcolor{gray}{$acedb\,(R^{8})$} & \textcolor{gray}{$aebcd\,(R^{15})$} & \textcolor{gray}{$aecbd\,(R^{8})$} & $aecdb$ \\
$R^{17}$ & $acedb$ & $aebcd$ & \textcolor{gray}{$acebd\,(R^{8})$} & \textcolor{gray}{$acedb\,(R^{8})$} & \textcolor{gray}{$aebcd\,(R^{16})$} & \textcolor{gray}{$aecbd\,(R^{8})$} & $aecdb$ \\
$R^{18}$ & $aebcd$ & $aecdb$ & \textcolor{gray}{$aebcd\,(R^{17})$} & \textcolor{gray}{$aecbd\,(R^{8})$} & $aecdb$ \\
$R^{19}$ & $abced$ & $ceabd$ & $abced$ & \textcolor{gray}{$acbed\,(R^{13})$} & \textcolor{gray}{$acebd\,(R^{5})$} & \textcolor{gray}{$cabed\,(R^{13})$} & \textcolor{gray}{$caebd\,(R^{13})$} & \textcolor{gray}{$ceabd\,(R^{6})$} \\
$R^{20}$ & $abced$ & $acedb$ & $abced$ & \textcolor{gray}{$acbed\,(R^{13})$} & \textcolor{gray}{$acebd\,(R^{19})$} & \textcolor{gray}{$acedb\,(R^{19})$} \\
$R^{21}$ & $abecd$ & $acedb$ & $abced$ & \textcolor{gray}{$abecd\,(R^{17})$} & \textcolor{gray}{$acbed\,(R^{20})$} & \textcolor{gray}{$acebd\,(R^{20})$} & \textcolor{gray}{$acedb\,(R^{20})$} & \textcolor{gray}{$aebcd\,(R^{17})$} & \textcolor{gray}{$aecbd\,(R^{20})$} & \textcolor{gray}{$aecdb\,(R^{20})$} \\
$R^{22}$ & $abecd$ & $aecdb$ & \textcolor{gray}{$abecd\,(R^{18})$} & \textcolor{gray}{$aebcd\,(R^{18})$} & \textcolor{gray}{$aecbd\,(R^{9})$} & \textcolor{gray}{$aecdb\,(R^{21})$} \\
\end{longtable}
}
}

\textbf{Case 3.2:} Next, we assume that $f(R^2)=acbed$ and derive a contradiction as shown below.
{\footnotesize
{\setlength{\tabcolsep}{3pt}
\noindent\begin{longtable}{c|cc|llllllll}
$R^{+}$ & $eabcd$ & $ecabd$ & $eabcd$ (A)\\
$R^{1}$ & $caebd$ & $cebad$ & $cebad$ (A)\\
$R^{2}$ & $aecbd$ & $cabed$ & $acbed$ (A)\\
$R^{3}$ & $ceabd$ & $eabcd$ & \textcolor{gray}{$ceabd\,(R^{+})$} & $eabcd$ & \textcolor{gray}{$eacbd\,(R^{+})$} & \textcolor{gray}{$ecabd\,(R^{+})$} \\
$R^{4}$ & $caebd$ & $eabcd$ & \textcolor{gray}{$acebd\,(R^{3})$} & $aebcd$ & $aecbd$ & \textcolor{gray}{$caebd\,(R^{3})$} & \textcolor{gray}{$ceabd\,(R^{+})$} & $eabcd$ & \textcolor{gray}{$eacbd\,(R^{+})$} & \textcolor{gray}{$ecabd\,(R^{+})$} \\
$R^{5}$ & $aebcd$ & $caebd$ & \textcolor{gray}{$acebd\,(R^{4})$} & $aebcd$ & $aecbd$ & \textcolor{gray}{$caebd\,(R^{4})$} \\
$R^{6}$ & $aecbd$ & $caebd$ & $acebd$ & $aecbd$ & \textcolor{gray}{$caebd\,(R^{5})$} \\
$R^{7}$ & $acebd$ & $cebad$ & \textcolor{gray}{$acebd\,(R^{1})$} & \textcolor{gray}{$caebd\,(R^{1})$} & \textcolor{gray}{$ceabd\,(R^{1})$} & $cebad$ \\
$R^{8}$ & $acebd$ & $ceabd$ & \textcolor{gray}{$acebd\,(R^{7})$} & $caebd$ & $ceabd$ \\
$R^{9}$ & $aecbd$ & $cebad$ & \textcolor{gray}{$acebd\,(R^{1})$} & \textcolor{gray}{$aecbd\,(R^{7})$} & \textcolor{gray}{$caebd\,(R^{1})$} & \textcolor{gray}{$ceabd\,(R^{1})$} & $cebad$ & \textcolor{gray}{$eacbd\,(R^{7})$} & $ecabd$ & $ecbad$ \\
$R^{10}$ & $aecbd$ & $ceabd$ & \textcolor{gray}{$acebd\,(R^{8})$} & \textcolor{gray}{$aecbd\,(R^{9})$} & \textcolor{gray}{$caebd\,(R^{6})$} & \textcolor{gray}{$ceabd\,(R^{3})$} & $eacbd$ & $ecabd$ \\
$R^{11}$ & $aecbd$ & $ecabd$ & \textcolor{gray}{$aecbd\,(R^{10})$} & $eacbd$ & $ecabd$ \\
$R^{12}$ & $ceadb$ & $eabcd$ & \textcolor{gray}{$ceabd\,(R^{+})$} & \textcolor{gray}{$ceadb\,(R^{3})$} & $eabcd$ & \textcolor{gray}{$eacbd\,(R^{+})$} & $eacdb$ & \textcolor{gray}{$ecabd\,(R^{+})$} & \textcolor{gray}{$ecadb\,(R^{+})$} \\
$R^{13}$ & $ceadb$ & $eadbc$ & \textcolor{gray}{$ceadb\,(R^{12})$} & $eacdb$ & $eadbc$ & $eadcb$ & \textcolor{gray}{$ecadb\,(R^{12})$} \\
$R^{14}$ & $ceadb$ & $eadcb$ & \textcolor{gray}{$ceadb\,(R^{13})$} & $eacdb$ & $eadcb$ & \textcolor{gray}{$ecadb\,(R^{13})$} \\
$R^{15}$ & $eadcb$ & $ecadb$ & $eacdb$ & $eadcb$ & \textcolor{gray}{$ecadb\,(R^{14})$} \\
$R^{16}$ & $abced$ & $cabed$ & $abced$ & $acbed$ & \textcolor{gray}{$cabed\,(R^{2})$} \\
$R^{17}$ & $acbed$ & $cabed$ & $acbed$ & \textcolor{gray}{$cabed\,(R^{2})$} \\
$R^{18}$ & $acbde$ & $aecbd$ & $acbde$ & $acbed$ & \textcolor{gray}{$acebd\,(R^{2})$} & \textcolor{gray}{$aecbd\,(R^{2})$} \\
$R^{19}$ & $acbde$ & $eacbd$ & $acbde$ & $acbed$ & \textcolor{gray}{$acebd\,(R^{18})$} & \textcolor{gray}{$aecbd\,(R^{18})$} & \textcolor{gray}{$eacbd\,(R^{18})$} \\
$R^{20}$ & $acbed$ & $eacbd$ & $acbed$ & $acebd$ & \textcolor{gray}{$aecbd\,(R^{19})$} & \textcolor{gray}{$eacbd\,(R^{19})$} \\
$R^{21}$ & $acbed$ & $ceabd$ & $acbed$ & \textcolor{gray}{$acebd\,(R^{8})$} & \textcolor{gray}{$cabed\,(R^{17})$} & $caebd$ & $ceabd$ \\
$R^{22}$ & $abced$ & $ceabd$ & $abced$ & $acbed$ & \textcolor{gray}{$acebd\,(R^{8})$} & \textcolor{gray}{$cabed\,(R^{16})$} & $caebd$ & \textcolor{gray}{$ceabd\,(R^{3})$} \\
$R^{23}$ & $acbed$ & $ceabd$ & $acbed$ & \textcolor{gray}{$acebd\,(R^{8})$} & \textcolor{gray}{$cabed\,(R^{17})$} & $caebd$ & \textcolor{gray}{$ceabd\,(R^{22})$} \\
$R^{24}$ & $acebd$ & $ceabd$ & \textcolor{gray}{$acebd\,(R^{7})$} & $caebd$ & \textcolor{gray}{$ceabd\,(R^{23})$} \\
$R^{25}$ & $acbed$ & $ecabd$ & $acbed$ & \textcolor{gray}{$acebd\,(R^{21})$} & \textcolor{gray}{$aecbd\,(R^{20})$} & \textcolor{gray}{$cabed\,(R^{17})$} & $caebd$ & \textcolor{gray}{$ceabd\,(R^{23})$} & \textcolor{gray}{$eacbd\,(R^{20})$} & \textcolor{gray}{$ecabd\,(R^{20})$} \\
$R^{26}$ & $acebd$ & $ecabd$ & \textcolor{gray}{$acebd\,(R^{7})$} & \textcolor{gray}{$aecbd\,(R^{11})$} & $caebd$ & \textcolor{gray}{$ceabd\,(R^{24})$} & \textcolor{gray}{$eacbd\,(R^{25})$} & \textcolor{gray}{$ecabd\,(R^{25})$} \\
$R^{27}$ & $acebd$ & $ecbad$ & \textcolor{gray}{$acebd\,(R^{7})$} & \textcolor{gray}{$aecbd\,(R^{7})$} & \textcolor{gray}{$caebd\,(R^{7})$} & \textcolor{gray}{$ceabd\,(R^{7})$} & $cebad$ & \textcolor{gray}{$eacbd\,(R^{7})$} & \textcolor{gray}{$ecabd\,(R^{26})$} & \textcolor{gray}{$ecbad\,(R^{26})$} \\
$R^{28}$ & $cebad$ & $ecbad$ & $cebad$ & \textcolor{gray}{$ecbad\,(R^{27})$} \\
$R^{29}$ & $ceadb$ & $eacdb$ & \textcolor{gray}{$ceadb\,(R^{14})$} & $eacdb$ & $ecadb$ \\
$R^{30}$ & $acbed$ & $cebad$ & \textcolor{gray}{$acbed\,(R^{7})$} & \textcolor{gray}{$acebd\,(R^{1})$} & \textcolor{gray}{$cabed\,(R^{1})$} & \textcolor{gray}{$caebd\,(R^{1})$} & $cbaed$ & $cbead$ & \textcolor{gray}{$ceabd\,(R^{1})$} & $cebad$ \\
$R^{31}$ & $acbed$ & $cbead$ & \textcolor{gray}{$acbed\,(R^{30})$} & \textcolor{gray}{$cabed\,(R^{17})$} & $cbaed$ & $cbead$ \\
$R^{32}$ & $acbed$ & $cbaed$ & \textcolor{gray}{$acbed\,(R^{31})$} & \textcolor{gray}{$cabed\,(R^{17})$} & $cbaed$ \\
$R^{33}$ & $cebad$ & $ecabd$ & \textcolor{gray}{$ceabd\,(R^{1})$} & $cebad$ & $ecabd$ & \textcolor{gray}{$ecbad\,(R^{28})$} \\
$R^{34}$ & $aecbd$ & $cebad$ & \textcolor{gray}{$acebd\,(R^{1})$} & \textcolor{gray}{$aecbd\,(R^{7})$} & \textcolor{gray}{$caebd\,(R^{1})$} & \textcolor{gray}{$ceabd\,(R^{1})$} & $cebad$ & \textcolor{gray}{$eacbd\,(R^{7})$} & $ecabd$ & \textcolor{gray}{$ecbad\,(R^{28})$} \\
$R^{35}$ & $aecbd$ & $ceabd$ & \textcolor{gray}{$acebd\,(R^{8})$} & \textcolor{gray}{$aecbd\,(R^{9})$} & \textcolor{gray}{$caebd\,(R^{6})$} & \textcolor{gray}{$ceabd\,(R^{3})$} & \textcolor{gray}{$eacbd\,(R^{34})$} & $ecabd$ \\
$R^{36}$ & $aecbd$ & $ecabd$ & \textcolor{gray}{$aecbd\,(R^{10})$} & \textcolor{gray}{$eacbd\,(R^{35})$} & $ecabd$ \\
$R^{37}$ & $aecbd$ & $cbaed$ & \textcolor{gray}{$acbed\,(R^{32})$} & \textcolor{gray}{$acebd\,(R^{2})$} & \textcolor{gray}{$aecbd\,(R^{2})$} & \textcolor{gray}{$cabed\,(R^{2})$} & \textcolor{gray}{$caebd\,(R^{6})$} & $cbaed$ \\
$R^{38}$ & $cbaed$ & $ecabd$ & \textcolor{gray}{$cabed\,(R^{32})$} & \textcolor{gray}{$caebd\,(R^{37})$} & $cbaed$ & $cbead$ & \textcolor{gray}{$ceabd\,(R^{37})$} & $cebad$ & \textcolor{gray}{$ecabd\,(R^{37})$} & \textcolor{gray}{$ecbad\,(R^{37})$} \\
$R^{39}$ & $cbead$ & $ecabd$ & $cbead$ & \textcolor{gray}{$ceabd\,(R^{33})$} & $cebad$ & \textcolor{gray}{$ecabd\,(R^{38})$} & \textcolor{gray}{$ecbad\,(R^{33})$} \\
$R^{40}$ & $cebad$ & $ecabd$ & \textcolor{gray}{$ceabd\,(R^{1})$} & $cebad$ & \textcolor{gray}{$ecabd\,(R^{39})$} & \textcolor{gray}{$ecbad\,(R^{28})$} \\
$R^{41}$ & $acebd$ & $ecbad$ & \textcolor{gray}{$acebd\,(R^{7})$} & \textcolor{gray}{$aecbd\,(R^{7})$} & \textcolor{gray}{$caebd\,(R^{7})$} & \textcolor{gray}{$ceabd\,(R^{7})$} & $cebad$ & \textcolor{gray}{$eacbd\,(R^{7})$} & $ecabd$ & $ecbad$ \\
$R^{42}$ & $aecbd$ & $cebad$ & \textcolor{gray}{$acebd\,(R^{1})$} & \textcolor{gray}{$aecbd\,(R^{7})$} & \textcolor{gray}{$caebd\,(R^{1})$} & \textcolor{gray}{$ceabd\,(R^{1})$} & $cebad$ & \textcolor{gray}{$eacbd\,(R^{7})$} & \textcolor{gray}{$ecabd\,(R^{40})$} & \textcolor{gray}{$ecbad\,(R^{28})$} \\
$R^{43}$ & $cebad$ & $eacdb$ & \textcolor{gray}{$ceabd\,(R^{1})$} & \textcolor{gray}{$ceadb\,(R^{29})$} & $cebad$ & \textcolor{gray}{$eacbd\,(R^{7})$} & \textcolor{gray}{$eacdb\,(R^{42})$} & \textcolor{gray}{$ecabd\,(R^{40})$} & \textcolor{gray}{$ecadb\,(R^{40})$} & \textcolor{gray}{$ecbad\,(R^{28})$} \\
$R^{44}$ & $eacdb$ & $ecbad$ & \textcolor{gray}{$eacbd\,(R^{41})$} & \textcolor{gray}{$eacdb\,(R^{43})$} & $ecabd$ & \textcolor{gray}{$ecadb\,(R^{43})$} & $ecbad$ \\
$R^{45}$ & $eadcb$ & $ecbad$ & \textcolor{gray}{$eacbd\,(R^{41})$} & \textcolor{gray}{$eacdb\,(R^{44})$} & \textcolor{gray}{$eadcb\,(R^{44})$} & $ecabd$ & \textcolor{gray}{$ecadb\,(R^{14})$} & $ecbad$ \\
$R^{46}$ & $eadcb$ & $ecabd$ & \textcolor{gray}{$eacbd\,(R^{36})$} & \textcolor{gray}{$eacdb\,(R^{45})$} & \textcolor{gray}{$eadcb\,(R^{45})$} & $ecabd$ & \textcolor{gray}{$ecadb\,(R^{14})$} \\
$R^{47}$ & $eadcb$ & $ecadb$ & $eacdb$ & \textcolor{gray}{$eadcb\,(R^{46})$} & \textcolor{gray}{$ecadb\,(R^{14})$} \\
$R^{48}$ & $aedcb$ & $ecabd$ & \textcolor{gray}{$aecbd\,(R^{11})$} & \textcolor{gray}{$aecdb\,(R^{36})$} & \textcolor{gray}{$aedcb\,(R^{46})$} & \textcolor{gray}{$eacbd\,(R^{36})$} & \textcolor{gray}{$eacdb\,(R^{46})$} & \textcolor{gray}{$eadcb\,(R^{46})$} & $ecabd$ & \textcolor{gray}{$ecadb\,(R^{46})$} \\
$R^{49}$ & $aedcb$ & $ecadb$ & \textcolor{gray}{$aecdb\,(R^{48})$} & \textcolor{gray}{$aedcb\,(R^{48})$} & $eacdb$ & \textcolor{gray}{$eadcb\,(R^{47})$} & \textcolor{gray}{$ecadb\,(R^{15})$} \\
$R^{50}$ & $aedcb$ & $eacdb$ & \textcolor{gray}{$aecdb\,(R^{49})$} & \textcolor{gray}{$aedcb\,(R^{49})$} & $eacdb$ & \textcolor{gray}{$eadcb\,(R^{49})$} \\
$R^{51}$ & $acedb$ & $eacdb$ & $acedb$ & \textcolor{gray}{$aecdb\,(R^{50})$} & $eacdb$ \\
$R^{52}$ & $ceadb$ & $eadcb$ & \textcolor{gray}{$ceadb\,(R^{13})$} & $eacdb$ & \textcolor{gray}{$eadcb\,(R^{46})$} & \textcolor{gray}{$ecadb\,(R^{13})$} \\
$R^{53}$ & $ceadb$ & $eacdb$ & \textcolor{gray}{$ceadb\,(R^{14})$} & $eacdb$ & \textcolor{gray}{$ecadb\,(R^{52})$} \\
$R^{54}$ & $cedba$ & $eacdb$ & \textcolor{gray}{$ceadb\,(R^{29})$} & \textcolor{gray}{$cedab\,(R^{53})$} & $cedba$ & \textcolor{gray}{$eacdb\,(R^{43})$} & \textcolor{gray}{$ecadb\,(R^{43})$} & $ecdab$ & $ecdba$ \\
$R^{55}$ & $eacdb$ & $ecdba$ & \textcolor{gray}{$eacdb\,(R^{54})$} & \textcolor{gray}{$ecadb\,(R^{53})$} & $ecdab$ & $ecdba$ \\
$R^{56}$ & $caebd$ & $ecabd$ & $caebd$ & \textcolor{gray}{$ceabd\,(R^{26})$} & \textcolor{gray}{$ecabd\,(R^{26})$} \\
$R^{57}$ & $cebad$ & $eabcd$ & \textcolor{gray}{$ceabd\,(R^{1})$} & \textcolor{gray}{$cebad\,(R^{3})$} & \textcolor{gray}{$eabcd\,(R^{42})$} & \textcolor{gray}{$eacbd\,(R^{+})$} & \textcolor{gray}{$ebacd\,(R^{42})$} & $ebcad$ & \textcolor{gray}{$ecabd\,(R^{+})$} & \textcolor{gray}{$ecbad\,(R^{+})$} \\
$R^{58}$ & $cebad$ & $ebcad$ & \textcolor{gray}{$cebad\,(R^{57})$} & $ebcad$ & \textcolor{gray}{$ecbad\,(R^{28})$} \\
$R^{59}$ & $eacdb$ & $ecdab$ & \textcolor{gray}{$eacdb\,(R^{55})$} & \textcolor{gray}{$ecadb\,(R^{53})$} & $ecdab$ \\
$R^{60}$ & $caebd$ & $ebcad$ & \textcolor{gray}{$caebd\,(R^{1})$} & \textcolor{gray}{$ceabd\,(R^{1})$} & \textcolor{gray}{$cebad\,(R^{58})$} & $ebcad$ & \textcolor{gray}{$ecabd\,(R^{56})$} & \textcolor{gray}{$ecbad\,(R^{56})$} \\
$R^{61}$ & $aecbd$ & $ebcad$ & \textcolor{gray}{$aebcd\,(R^{9})$} & \textcolor{gray}{$aecbd\,(R^{9})$} & $eabcd$ & \textcolor{gray}{$eacbd\,(R^{9})$} & $ebacd$ & $ebcad$ & \textcolor{gray}{$ecabd\,(R^{60})$} & \textcolor{gray}{$ecbad\,(R^{58})$} \\
$R^{62}$ & $aecbd$ & $ebacd$ & \textcolor{gray}{$aebcd\,(R^{61})$} & \textcolor{gray}{$aecbd\,(R^{61})$} & $eabcd$ & \textcolor{gray}{$eacbd\,(R^{61})$} & $ebacd$ \\
$R^{63}$ & $aecbd$ & $eabcd$ & $aebcd$ & \textcolor{gray}{$aecbd\,(R^{62})$} & $eabcd$ & \textcolor{gray}{$eacbd\,(R^{+})$} \\
$R^{64}$ & $acbed$ & $eacdb$ & $acbed$ & $acebd$ & $acedb$ & \textcolor{gray}{$aecbd\,(R^{20})$} & \textcolor{gray}{$aecdb\,(R^{50})$} & \textcolor{gray}{$eacbd\,(R^{20})$} & \textcolor{gray}{$eacdb\,(R^{20})$} \\
$R^{65}$ & $adecb$ & $eacdb$ & \textcolor{gray}{$adecb\,(R^{50})$} & \textcolor{gray}{$aecdb\,(R^{50})$} & \textcolor{gray}{$aedcb\,(R^{50})$} & $eacdb$ & \textcolor{gray}{$eadcb\,(R^{50})$} \\
$R^{66}$ & $acdbe$ & $eacdb$ & $acdbe$ & \textcolor{gray}{$acdeb\,(R^{65})$} & $acedb$ & \textcolor{gray}{$aecdb\,(R^{50})$} & \textcolor{gray}{$eacdb\,(R^{64})$} \\
$R^{67}$ & $adceb$ & $eacdb$ & \textcolor{gray}{$acdeb\,(R^{65})$} & $acedb$ & \textcolor{gray}{$adceb\,(R^{65})$} & \textcolor{gray}{$adecb\,(R^{50})$} & \textcolor{gray}{$aecdb\,(R^{50})$} & \textcolor{gray}{$aedcb\,(R^{50})$} & \textcolor{gray}{$eacdb\,(R^{66})$} & \textcolor{gray}{$eadcb\,(R^{50})$} \\
$R^{68}$ & $acedb$ & $eacdb$ & $acedb$ & \textcolor{gray}{$aecdb\,(R^{50})$} & \textcolor{gray}{$eacdb\,(R^{67})$} \\
$R^{69}$ & $eabcd$ & $ecdab$ & $eabcd$ & \textcolor{gray}{$eacbd\,(R^{+})$} & $eacdb$ & \textcolor{gray}{$ecabd\,(R^{+})$} & \textcolor{gray}{$ecadb\,(R^{+})$} & $ecdab$ \\
$R^{70}$ & $ceadb$ & $eadbc$ & \textcolor{gray}{$ceadb\,(R^{12})$} & $eacdb$ & $eadbc$ & \textcolor{gray}{$eadcb\,(R^{52})$} & \textcolor{gray}{$ecadb\,(R^{12})$} \\
$R^{71}$ & $caebd$ & $eabcd$ & \textcolor{gray}{$acebd\,(R^{3})$} & $aebcd$ & \textcolor{gray}{$aecbd\,(R^{63})$} & \textcolor{gray}{$caebd\,(R^{3})$} & \textcolor{gray}{$ceabd\,(R^{+})$} & $eabcd$ & \textcolor{gray}{$eacbd\,(R^{+})$} & \textcolor{gray}{$ecabd\,(R^{+})$} \\
$R^{72}$ & $acedb$ & $eabcd$ & \textcolor{gray}{$acebd\,(R^{3})$} & \textcolor{gray}{$acedb\,(R^{71})$} & $aebcd$ & \textcolor{gray}{$aecbd\,(R^{63})$} & \textcolor{gray}{$aecdb\,(R^{51})$} & $eabcd$ & \textcolor{gray}{$eacbd\,(R^{+})$} & \textcolor{gray}{$eacdb\,(R^{68})$} \\
$R^{73}$ & $acedb$ & $eadbc$ & \textcolor{gray}{$acedb\,(R^{72})$} & \textcolor{gray}{$aecdb\,(R^{51})$} & $aedbc$ & $aedcb$ & \textcolor{gray}{$eacdb\,(R^{68})$} & $eadbc$ & \textcolor{gray}{$eadcb\,(R^{70})$} \\
$R^{74}$ & $acedb$ & $edacb$ & \textcolor{gray}{$acedb\,(R^{73})$} & \textcolor{gray}{$aecdb\,(R^{51})$} & $aedcb$ & \textcolor{gray}{$eacdb\,(R^{68})$} & \textcolor{gray}{$eadcb\,(R^{68})$} & $edacb$ \\
$R^{75}$ & $ceadb$ & $eabcd$ & \textcolor{gray}{$ceabd\,(R^{+})$} & \textcolor{gray}{$ceadb\,(R^{3})$} & $eabcd$ & \textcolor{gray}{$eacbd\,(R^{+})$} & \textcolor{gray}{$eacdb\,(R^{72})$} & \textcolor{gray}{$ecabd\,(R^{+})$} & \textcolor{gray}{$ecadb\,(R^{+})$} \\
$R^{76}$ & $eabcd$ & $ecdab$ & $eabcd$ & \textcolor{gray}{$eacbd\,(R^{+})$} & \textcolor{gray}{$eacdb\,(R^{59})$} & \textcolor{gray}{$ecabd\,(R^{+})$} & \textcolor{gray}{$ecadb\,(R^{+})$} & \textcolor{gray}{$ecdab\,(R^{75})$} \\
$R^{77}$ & $eadbc$ & $ecdab$ & \textcolor{gray}{$eacdb\,(R^{59})$} & $eadbc$ & \textcolor{gray}{$eadcb\,(R^{70})$} & \textcolor{gray}{$ecadb\,(R^{69})$} & \textcolor{gray}{$ecdab\,(R^{76})$} & $edabc$ & $edacb$ & \textcolor{gray}{$edcab\,(R^{76})$} \\
$R^{78}$ & $aedcb$ & $ecdab$ & \textcolor{gray}{$aecdb\,(R^{48})$} & \textcolor{gray}{$aedcb\,(R^{48})$} & \textcolor{gray}{$eacdb\,(R^{59})$} & \textcolor{gray}{$eadcb\,(R^{48})$} & \textcolor{gray}{$ecadb\,(R^{49})$} & \textcolor{gray}{$ecdab\,(R^{77})$} & $edacb$ & $edcab$ \\
$R^{79}$ & $aedcb$ & $edacb$ & \textcolor{gray}{$aedcb\,(R^{78})$} & \textcolor{gray}{$eadcb\,(R^{74})$} & $edacb$ \\
$R^{80}$ & $acedb$ & $edacb$ & \textcolor{gray}{$acedb\,(R^{73})$} & \textcolor{gray}{$aecdb\,(R^{51})$} & \textcolor{gray}{$aedcb\,(R^{79})$} & \textcolor{gray}{$eacdb\,(R^{68})$} & \textcolor{gray}{$eadcb\,(R^{68})$} & $edacb$ \\
$R^{81}$ & $ecabd$ & $edacb$ & \textcolor{gray}{$eacbd\,(R^{36})$} & \textcolor{gray}{$eacdb\,(R^{46})$} & \textcolor{gray}{$eadcb\,(R^{46})$} & $ecabd$ & \textcolor{gray}{$ecadb\,(R^{46})$} & $ecdab$ & \textcolor{gray}{$edacb\,(R^{46})$} & \textcolor{gray}{$edcab\,(R^{46})$} \\
$R^{82}$ & $ecadb$ & $edacb$ & \textcolor{gray}{$eacdb\,(R^{80})$} & \textcolor{gray}{$eadcb\,(R^{47})$} & \textcolor{gray}{$ecadb\,(R^{15})$} & $ecdab$ & \textcolor{gray}{$edacb\,(R^{81})$} & \textcolor{gray}{$edcab\,(R^{81})$} \\
$R^{83}$ & $ecdab$ & $edacb$ & \textcolor{gray}{$ecdab\,(R^{78})$} & \textcolor{gray}{$edacb\,(R^{82})$} & \textcolor{gray}{$edcab\,(R^{82})$} \\
\end{longtable}
}
}

\textbf{Case 3.3:} Thirdly, we show that $f(R^2)=acebd$ is not possible. As usual, we assume the opposite, i.e., that $f(R^2)=acebd$ and infer a contradiction with the following table. 
{\footnotesize
{\setlength{\tabcolsep}{3pt}
\noindent\begin{longtable}{c|cc|llllllll}
$R^{*}$ & $abcde$ & $acbed$ & $abcde$ (A)\\
$R^{+}$ & $eabcd$ & $ecabd$ & $eabcd$ (A)\\
$R^{1}$ & $caebd$ & $cebad$ & $cebad$ (A)\\
$R^{2}$ & $aecbd$ & $cabed$ & $acebd$ (A)\\
$R^{3}$ & $aecbd$ & $cbaed$ & \textcolor{gray}{$acbed\,(R^{2})$} & $acebd$ & \textcolor{gray}{$aecbd\,(R^{2})$} & \textcolor{gray}{$cabed\,(R^{2})$} & \textcolor{gray}{$caebd\,(R^{2})$} & \textcolor{gray}{$cbaed\,(R^{2})$} \\
$R^{4}$ & $acbed$ & $cbaed$ & $acbed$ & $cabed$ & \textcolor{gray}{$cbaed\,(R^{3})$} \\
$R^{5}$ & $acebd$ & $cebad$ & \textcolor{gray}{$acebd\,(R^{1})$} & \textcolor{gray}{$caebd\,(R^{1})$} & \textcolor{gray}{$ceabd\,(R^{1})$} & $cebad$ \\
$R^{6}$ & $acbed$ & $cebad$ & \textcolor{gray}{$acbed\,(R^{5})$} & \textcolor{gray}{$acebd\,(R^{1})$} & \textcolor{gray}{$cabed\,(R^{1})$} & \textcolor{gray}{$caebd\,(R^{1})$} & \textcolor{gray}{$cbaed\,(R^{4})$} & $cbead$ & \textcolor{gray}{$ceabd\,(R^{1})$} & $cebad$ \\
$R^{7}$ & $acbed$ & $cbead$ & \textcolor{gray}{$acbed\,(R^{6})$} & \textcolor{gray}{$cabed\,(R^{6})$} & \textcolor{gray}{$cbaed\,(R^{4})$} & $cbead$ \\
$R^{8}$ & $acbed$ & $cbaed$ & \textcolor{gray}{$acbed\,(R^{7})$} & $cabed$ & \textcolor{gray}{$cbaed\,(R^{3})$} \\
$R^{9}$ & $acbed$ & $cabed$ & \textcolor{gray}{$acbed\,(R^{8})$} & $cabed$ \\
$R^{10}$ & $acebd$ & $ceabd$ & \textcolor{gray}{$acebd\,(R^{5})$} & $caebd$ & $ceabd$ \\
$R^{11}$ & $aecbd$ & $caebd$ & $acebd$ & \textcolor{gray}{$aecbd\,(R^{2})$} & \textcolor{gray}{$caebd\,(R^{2})$} \\
$R^{12}$ & $ceabd$ & $eabcd$ & \textcolor{gray}{$ceabd\,(R^{+})$} & $eabcd$ & \textcolor{gray}{$eacbd\,(R^{+})$} & \textcolor{gray}{$ecabd\,(R^{+})$} \\
$R^{13}$ & $caebd$ & $eabcd$ & \textcolor{gray}{$acebd\,(R^{12})$} & $aebcd$ & \textcolor{gray}{$aecbd\,(R^{11})$} & \textcolor{gray}{$caebd\,(R^{11})$} & \textcolor{gray}{$ceabd\,(R^{+})$} & $eabcd$ & \textcolor{gray}{$eacbd\,(R^{+})$} & \textcolor{gray}{$ecabd\,(R^{+})$} \\
$R^{14}$ & $aebcd$ & $caebd$ & \textcolor{gray}{$acebd\,(R^{13})$} & $aebcd$ & \textcolor{gray}{$aecbd\,(R^{11})$} & \textcolor{gray}{$caebd\,(R^{11})$} \\
$R^{15}$ & $aebcd$ & $aecdb$ & $aebcd$ & \textcolor{gray}{$aecbd\,(R^{14})$} & $aecdb$ \\
$R^{16}$ & $acbde$ & $aecbd$ & $acbde$ & \textcolor{gray}{$acbed\,(R^{2})$} & $acebd$ & \textcolor{gray}{$aecbd\,(R^{2})$} \\
$R^{17}$ & $abcde$ & $acebd$ & $abcde$ & \textcolor{gray}{$abced\,(R^{*})$} & \textcolor{gray}{$acbde\,(R^{*})$} & \textcolor{gray}{$acbed\,(R^{*})$} & \textcolor{gray}{$acebd\,(R^{*})$} \\
$R^{18}$ & $acbde$ & $acebd$ & $acbde$ & \textcolor{gray}{$acbed\,(R^{16})$} & \textcolor{gray}{$acebd\,(R^{17})$} \\
$R^{19}$ & $acbde$ & $aecbd$ & $acbde$ & \textcolor{gray}{$acbed\,(R^{2})$} & \textcolor{gray}{$acebd\,(R^{18})$} & \textcolor{gray}{$aecbd\,(R^{2})$} \\
$R^{20}$ & $abcde$ & $aecbd$ & $abcde$ & \textcolor{gray}{$abced\,(R^{*})$} & \textcolor{gray}{$abecd\,(R^{19})$} & \textcolor{gray}{$acbde\,(R^{*})$} & \textcolor{gray}{$acbed\,(R^{2})$} & \textcolor{gray}{$acebd\,(R^{*})$} & \textcolor{gray}{$aebcd\,(R^{17})$} & \textcolor{gray}{$aecbd\,(R^{2})$} \\
$R^{21}$ & $abcde$ & $aebcd$ & $abcde$ & \textcolor{gray}{$abced\,(R^{*})$} & \textcolor{gray}{$abecd\,(R^{20})$} & \textcolor{gray}{$aebcd\,(R^{17})$} \\
$R^{22}$ & $aebcd$ & $cabed$ & $abced$ & $abecd$ & \textcolor{gray}{$acbed\,(R^{2})$} & \textcolor{gray}{$acebd\,(R^{14})$} & \textcolor{gray}{$aebcd\,(R^{21})$} & \textcolor{gray}{$aecbd\,(R^{2})$} & \textcolor{gray}{$cabed\,(R^{2})$} & \textcolor{gray}{$caebd\,(R^{2})$} \\
$R^{23}$ & $abced$ & $cabed$ & $abced$ & \textcolor{gray}{$acbed\,(R^{9})$} & \textcolor{gray}{$cabed\,(R^{22})$} \\
$R^{24}$ & $aebcd$ & $caebd$ & $acebd$ & $aebcd$ & \textcolor{gray}{$aecbd\,(R^{11})$} & \textcolor{gray}{$caebd\,(R^{11})$} \\
$R^{25}$ & $acdeb$ & $aebcd$ & $acdeb$ & \textcolor{gray}{$acebd\,(R^{14})$} & \textcolor{gray}{$acedb\,(R^{14})$} & \textcolor{gray}{$aebcd\,(R^{21})$} & \textcolor{gray}{$aecbd\,(R^{24})$} & $aecdb$ \\
$R^{26}$ & $acedb$ & $aebcd$ & \textcolor{gray}{$acebd\,(R^{14})$} & \textcolor{gray}{$acedb\,(R^{14})$} & \textcolor{gray}{$aebcd\,(R^{25})$} & \textcolor{gray}{$aecbd\,(R^{24})$} & $aecdb$ \\
$R^{27}$ & $aebcd$ & $aecdb$ & \textcolor{gray}{$aebcd\,(R^{26})$} & \textcolor{gray}{$aecbd\,(R^{14})$} & $aecdb$ \\
$R^{28}$ & $abced$ & $ceabd$ & $abced$ & \textcolor{gray}{$acbed\,(R^{23})$} & \textcolor{gray}{$acebd\,(R^{10})$} & \textcolor{gray}{$cabed\,(R^{23})$} & \textcolor{gray}{$caebd\,(R^{23})$} & \textcolor{gray}{$ceabd\,(R^{12})$} \\
$R^{29}$ & $abced$ & $acedb$ & $abced$ & \textcolor{gray}{$acbed\,(R^{23})$} & \textcolor{gray}{$acebd\,(R^{28})$} & \textcolor{gray}{$acedb\,(R^{28})$} \\
$R^{30}$ & $abecd$ & $acedb$ & $abced$ & \textcolor{gray}{$abecd\,(R^{26})$} & \textcolor{gray}{$acbed\,(R^{29})$} & \textcolor{gray}{$acebd\,(R^{29})$} & \textcolor{gray}{$acedb\,(R^{29})$} & \textcolor{gray}{$aebcd\,(R^{26})$} & \textcolor{gray}{$aecbd\,(R^{29})$} & \textcolor{gray}{$aecdb\,(R^{29})$} \\
$R^{31}$ & $abecd$ & $aecdb$ & \textcolor{gray}{$abecd\,(R^{27})$} & \textcolor{gray}{$aebcd\,(R^{27})$} & \textcolor{gray}{$aecbd\,(R^{15})$} & \textcolor{gray}{$aecdb\,(R^{30})$} \\
\end{longtable}
}
}

\textbf{Case 3.4:} Lastly, we assume that $f(R^2)=aecbd$ and derive again a contradiction, as demonstrated with the following table. 
{\footnotesize
{\setlength{\tabcolsep}{3pt}
\noindent\begin{longtable}{c|cc|llllllll}
$R^{*}$ & $abcde$ & $acbed$ & $abcde$ (A)\\
$R^{+}$ & $eabcd$ & $ecabd$ & $eabcd$ (A)\\
$R^{1}$ & $caebd$ & $cebad$ & $cebad$ (A)\\
$R^{2}$ & $aecbd$ & $cabed$ & $aecbd$ (A)\\
$R^{3}$ & $acebd$ & $cebad$ & \textcolor{gray}{$acebd\,(R^{1})$} & \textcolor{gray}{$caebd\,(R^{1})$} & \textcolor{gray}{$ceabd\,(R^{1})$} & $cebad$ \\
$R^{4}$ & $aecbd$ & $cbaed$ & \textcolor{gray}{$acbed\,(R^{2})$} & \textcolor{gray}{$acebd\,(R^{2})$} & $aecbd$ & \textcolor{gray}{$cabed\,(R^{2})$} & \textcolor{gray}{$caebd\,(R^{2})$} & \textcolor{gray}{$cbaed\,(R^{2})$} \\
$R^{5}$ & $aecbd$ & $cebad$ & \textcolor{gray}{$acebd\,(R^{2})$} & \textcolor{gray}{$aecbd\,(R^{3})$} & \textcolor{gray}{$caebd\,(R^{2})$} & \textcolor{gray}{$ceabd\,(R^{2})$} & \textcolor{gray}{$cebad\,(R^{4})$} & \textcolor{gray}{$eacbd\,(R^{3})$} & $ecabd$ & \textcolor{gray}{$ecbad\,(R^{4})$} \\
$R^{6}$ & $aecbd$ & $ecabd$ & \textcolor{gray}{$aecbd\,(R^{5})$} & \textcolor{gray}{$eacbd\,(R^{5})$} & $ecabd$ \\
$R^{7}$ & $acbde$ & $aecbd$ & \textcolor{gray}{$acbde\,(R^{2})$} & \textcolor{gray}{$acbed\,(R^{2})$} & \textcolor{gray}{$acebd\,(R^{2})$} & $aecbd$ \\
$R^{8}$ & $acbde$ & $eacbd$ & \textcolor{gray}{$acbde\,(R^{7})$} & \textcolor{gray}{$acbed\,(R^{7})$} & \textcolor{gray}{$acebd\,(R^{7})$} & $aecbd$ & $eacbd$ \\
$R^{9}$ & $acbed$ & $eacbd$ & \textcolor{gray}{$acbed\,(R^{8})$} & \textcolor{gray}{$acebd\,(R^{8})$} & $aecbd$ & $eacbd$ \\
$R^{10}$ & $abcde$ & $acebd$ & $abcde$ & \textcolor{gray}{$abced\,(R^{*})$} & \textcolor{gray}{$acbde\,(R^{*})$} & \textcolor{gray}{$acbed\,(R^{*})$} & \textcolor{gray}{$acebd\,(R^{*})$} \\
$R^{11}$ & $abcde$ & $aecbd$ & \textcolor{gray}{$abcde\,(R^{7})$} & \textcolor{gray}{$abced\,(R^{2})$} & $abecd$ & \textcolor{gray}{$acbde\,(R^{2})$} & \textcolor{gray}{$acbed\,(R^{2})$} & \textcolor{gray}{$acebd\,(R^{2})$} & \textcolor{gray}{$aebcd\,(R^{10})$} & \textcolor{gray}{$aecbd\,(R^{10})$} \\
$R^{12}$ & $abecd$ & $aecbd$ & $abecd$ & \textcolor{gray}{$aebcd\,(R^{11})$} & \textcolor{gray}{$aecbd\,(R^{11})$} \\
$R^{13}$ & $aecbd$ & $caebd$ & \textcolor{gray}{$acebd\,(R^{2})$} & $aecbd$ & \textcolor{gray}{$caebd\,(R^{2})$} \\
$R^{14}$ & $cebad$ & $ecabd$ & \textcolor{gray}{$ceabd\,(R^{1})$} & \textcolor{gray}{$cebad\,(R^{5})$} & $ecabd$ & \textcolor{gray}{$ecbad\,(R^{5})$} \\
$R^{15}$ & $ceabd$ & $eabcd$ & \textcolor{gray}{$ceabd\,(R^{+})$} & $eabcd$ & \textcolor{gray}{$eacbd\,(R^{+})$} & \textcolor{gray}{$ecabd\,(R^{+})$} \\
$R^{16}$ & $cbead$ & $ecabd$ & \textcolor{gray}{$cbead\,(R^{14})$} & \textcolor{gray}{$ceabd\,(R^{14})$} & \textcolor{gray}{$cebad\,(R^{14})$} & $ecabd$ & \textcolor{gray}{$ecbad\,(R^{14})$} \\
$R^{17}$ & $aecbd$ & $ceabd$ & \textcolor{gray}{$acebd\,(R^{2})$} & \textcolor{gray}{$aecbd\,(R^{5})$} & \textcolor{gray}{$caebd\,(R^{2})$} & \textcolor{gray}{$ceabd\,(R^{2})$} & \textcolor{gray}{$eacbd\,(R^{5})$} & $ecabd$ \\
$R^{18}$ & $acbed$ & $ecabd$ & \textcolor{gray}{$acbed\,(R^{9})$} & \textcolor{gray}{$acebd\,(R^{6})$} & \textcolor{gray}{$aecbd\,(R^{6})$} & \textcolor{gray}{$cabed\,(R^{16})$} & $caebd$ & \textcolor{gray}{$ceabd\,(R^{14})$} & \textcolor{gray}{$eacbd\,(R^{6})$} & $ecabd$ \\
$R^{19}$ & $acbed$ & $ceabd$ & \textcolor{gray}{$acbed\,(R^{18})$} & \textcolor{gray}{$acebd\,(R^{17})$} & \textcolor{gray}{$cabed\,(R^{18})$} & $caebd$ & $ceabd$ \\
$R^{20}$ & $abced$ & $ceabd$ & \textcolor{gray}{$abced\,(R^{19})$} & \textcolor{gray}{$acbed\,(R^{19})$} & \textcolor{gray}{$acebd\,(R^{17})$} & \textcolor{gray}{$cabed\,(R^{19})$} & $caebd$ & \textcolor{gray}{$ceabd\,(R^{15})$} \\
$R^{21}$ & $abced$ & $caebd$ & \textcolor{gray}{$abced\,(R^{20})$} & \textcolor{gray}{$acbed\,(R^{20})$} & \textcolor{gray}{$acebd\,(R^{20})$} & \textcolor{gray}{$cabed\,(R^{20})$} & $caebd$ \\
$R^{22}$ & $abecd$ & $caebd$ & \textcolor{gray}{$abced\,(R^{21})$} & \textcolor{gray}{$abecd\,(R^{21})$} & \textcolor{gray}{$acbed\,(R^{21})$} & \textcolor{gray}{$acebd\,(R^{13})$} & \textcolor{gray}{$aebcd\,(R^{12})$} & \textcolor{gray}{$aecbd\,(R^{12})$} & \textcolor{gray}{$cabed\,(R^{13})$} & \textcolor{gray}{$caebd\,(R^{13})$} \\
\end{longtable}
}
}

Since no valid ranking remains for $R^2$, this means that $f(R^1)=cebad$ is not possible either. Since we exhausted all cases for $f(R^1)$, this proves the lemma.
\end{proof}

\end{document}